\documentclass[11pt]{article}

\usepackage[margin=1in]{geometry}
\usepackage{mdframed}

\usepackage{color}
\usepackage{cite}
\usepackage{amsthm}
\usepackage{amsmath}
\usepackage{tablefootnote}
\usepackage{comment}
\usepackage{mdframed}
\usepackage{amssymb}
\usepackage{booktabs}
\usepackage{braket}
\usepackage{wrapfig} 
\usepackage{thmtools}
\usepackage{thm-restate}
\usepackage[colorlinks = true]{hyperref}
\usepackage[noabbrev]{cleveref}
\usepackage{graphicx}
\usepackage[font=large]{caption}
\usepackage{subcaption}
\usepackage[T1]{fontenc}
\usepackage[utf8]{inputenc}
\usepackage[fleqn]{mathtools}
\usepackage{mathtools}
\mathtoolsset{centercolon}
\usepackage{multicol}
\usepackage{authblk}
\usepackage{url}
\usepackage{float}
\usepackage{mathtools}

\usepackage{xcolor}
\definecolor{darkred}  {rgb}{0.5,0,0}
\definecolor{darkblue} {rgb}{0,0,0.5}
\definecolor{darkgreen}{rgb}{0,0.5,0}
\hypersetup{
  urlcolor   = blue,         % color of external links
  linkcolor  = darkblue,     % color of internal links
  citecolor  = darkgreen,    % color of links to bibliography
  filecolor  = darkred       % color of file links
}

\newtheorem{theorem}{Theorem}
\newtheorem{proposition}[theorem]{Proposition}
\newtheorem{definition}[theorem]{Definition}
\newtheorem{algorithm}[theorem]{Algorithm}
\newtheorem{lemma}[theorem]{Lemma}
\newtheorem{conjecture}[theorem]{Conjecture}
\newtheorem{corollary}[theorem]{Corollary}

\newtheorem{problem}[theorem]{Problem}

\newcommand{\be}{\begin{eqnarray}}
\newcommand{\ee}{\end{eqnarray}}

\newcommand{\semigeq}{\succeq}

\newcommand{\knote}[1]{\textcolor{red}{({\bf Kevin:} #1)}}
\newcommand{\onote}[1]{\textcolor{blue}{({\bf Ojas:} #1)}}

\renewcommand{\knote}[1]{}
\renewcommand{\onote}[1]{}

\begin{document}
\title{ Beating Random Assignment for Approximating Quantum $2$-Local Hamiltonian Problems }
\author{
  Ojas Parekh\footnote{Sandia National Laboratories, {\it email:} odparek@sandia.gov} \,\,and  
  Kevin Thompson\footnote{Sandia National Laboratories, {\it email:} kevthom@sandia.gov}
}

\date{}
\maketitle

\begin{abstract}
The quantum $k$-Local Hamiltonian problem is a natural generalization of classical constraint satisfaction problems ($k$-CSP) and is complete for QMA, a quantum analog of NP. Although the complexity of $k$-Local Hamiltonian problems has been well studied, only a handful of approximation results are known. For Max $2$-Local Hamiltonian where each term is a rank 3 projector, a natural quantum generalization of classical Max $2$-SAT, the best known approximation algorithm was the trivial random assignment, yielding a $0.75$-approximation. We present the first approximation algorithm beating this bound, a classical polynomial-time $0.764$-approximation. For strictly quadratic instances, which are maximally entangled instances, we provide a $0.801$ approximation algorithm, and numerically demonstrate that our algorithm is likely a $0.821$-approximation. We conjecture these are the hardest instances to approximate. We also give improved approximations for quantum generalizations of other related classical $2$-CSPs. Finally, we exploit quantum connections to a generalization of the Grothendieck problem to obtain a classical constant-factor approximation for the physically relevant special case of strictly quadratic traceless $2$-Local Hamiltonians on bipartite interaction graphs, where a inverse logarithmic approximation was the best previously known (for general interaction graphs).  Our work employs recently developed techniques for analyzing classical approximations of CSPs and is intended to be accessible to both quantum information scientists and classical computer scientists.
\end{abstract}

\knote{General Notes:}

\knote{Put table of contents}

\knote{ worry about punctuations on equations}

\knote{Make sure vector notation is $[..]$}

\knote{Delete et al.'s.  Probably easier to just remove references last names}

\knote{Worry about transposing vectors or put a note in the notations section}

\onote{Put table comparing classical/quantum CSP approximation results}

\onote{Include section showing how classical and quantum CSPs, and the corresponding rounding algorithms, relate}

\onote{Read over Kevin's changes}

\onote{Break up long sections with \\paragraphs}

\onote{Do content-preserving compression on the paper to make it shorter and easier to read}
\tableofcontents

\section{Introduction}

The design and analysis of approximation algorithms \cite{W11, V13} is an extensively studied area in theoretical computer science.  In this setting, we are given some (generally NP-hard) optimization problem, and we are tasked with producing a valid (feasible) solution with objective within some provable factor of the optimal objective value.  To understand this formally imagine we are given some optimization problem $\mathcal{P}$, which corresponds to an infinite set of problem instances $\{P_i\}$.  Each problem instance corresponds to a triple $P_i=(f_i, \mathcal{T}_i, i_n)$ where $\mathcal{T}_i \subseteq \{0, 1\}^{i_n}$ and each triple corresponds to an optimization problem of the form:
\begin{equation*}
OPT_i=\max_{v\in \mathcal{T}_i} f_i(v).
\end{equation*}
\noindent An approximation algorithm $\mathcal{A}$ acts on an efficient description of the instance to produce some feasible solution to the problem: $\mathcal{A}(f_i, \mathcal{T}_i, i_n)=v^*_i \in \mathcal{T}_i$ in time polynomial in the instance size (polynomial in $i_n$).  It is said that algorithm has approximation factor $\alpha$ for $0 < \alpha \leq 1$ if in the worst-case (over all instances), the solution produced by the algorithm is a factor of $\alpha$ off of the optimal answer:
\begin{equation*}
\min_i \frac{f_i(v_i^*)}{OPT_i} \geq \alpha.
\end{equation*}

Since we should not expect to solve NP-hard problems, the interesting question is then the approximability of NP-hard optimization problems, or the study of which approximation factors $\alpha$ are obtainable for different problems.  As one might expect, approximability is highly problem sensitive and there are many classes of natural problems with very different attainable approximation factors \cite{H97, D05, V13}.  So, in general approximation algorithms are developed in the context of specific problems, and many such algorithms are known \cite{V13, W11}.   

\paragraph{$2$-Local Hamiltonian.} In stark contrast, although QMA-hard quantum optimization problems arise naturally through well-known physically motivated problems \cite{Be07, S09}, they have very few known approximation algorithms with provable approximation factors \cite{B07, G12, B16, H17, B19, G19, H20, A20}\footnote{Here and throughout this paper we mean a {\it classical } algorithm which takes as input a classical description of a quantum problem and produces a classical description of a quantum state.  An approximation algorithm for a QMA-hard problem can have several natural meanings distinct from this (quantum input, quantum algorithm which produces classical output, etc.)}. The QMA-hard optimization studied in these works, as well as the problem we sill study here, is the {\it $2$-Local Hamiltonian problem}\cite{K06, K02}.  An instance of this problem is specified by a problem size, $n$, as well as a set of $2$-local interactions, $\{H_e\}$.  Each $H_e$ is some {\it local} Hamiltonian which can be written as the tensor product of $n-2$ identity terms with some nontrivial operator that acts on at most $2$ qubits, i.e. $H_e=\mathcal{O}_{ij}\otimes (\mathbb{I}_2)^{\otimes (n-2)}$.  The optimization problem corresponding to a particular instance is to find the smallest or largest eigenvalue, $\lambda_{min}$ or $\lambda_{max}$, of $H=\sum_e H_e$. Ideally, an algorithm solving this problem would also produce a description of or access to a corresponding eigenvector.  An approximation algorithm, $\mathcal{A}$, acts on the size of the problem ($n$) and a description of the local Hamiltonians $\{H_e\}$ to produce a classical description of a valid quantum state.  Once again we say that the algorithm achieves approximation factor $\alpha$ if:
\begin{equation*}
\frac{\text{Tr}[\left(\sum_e H_e \right)\mathcal{A}(n, \{H_e\})]}{\lambda_{max}\left( \sum_e H_e\right)}\geq \alpha \text{     for all instances}.
\end{equation*}
\noindent Generally we assume some property of the Hamiltonian which forces $\lambda_{max}\left( \sum_e H_e\right) > 0$ so that this is a sensible definition.  A common assumption \cite{K06,G12,H20} is that the terms $H_e$ are positive semi-definite (PSD) and nonzero. We note that when all of the terms $H_e$ are taken to be diagonal projectors (in say, the standard computational basis), the corresponding instance of $2$-Local Hamiltonian corresponds precisely to an instance of the classical $2$-Constraint-Satisfaction problem ($2$-CSP).  In this case, the 4 diagonal entries of $\mathcal{O}_{ij}$ correspond to the $\{0,1\}$ output values of a Boolean function on Boolean variables $x_i$ and $x_j$ corresponding to $i$ and $j$. See \Cref{sec:classical-quantum} for more details as well as a classical motivation for $2$-Local Hamiltonian. In addition \Cref{table:results} highlights classical $2$-CSP specializations of quantum $2$-local Hamiltonian problems for which approximation algorithms are known.

The $2$-Local Hamiltonian problem is interesting in many different contexts of physics and quantum information \cite{K06, K02, O12}.  This problem is manifestly interesting to physicists because the $2$-local nature of the problem matches the local nature of many physical systems (spin chains, Ising model, etc.).  Hence, the study of eigenstates and energies is of utmost importance, and has been since the beginnings of quantum mechanics itself \cite{B31}.  From a theoretical computer science perspective, the $2$-Local Hamiltonian problem is interesting for the same reasons that classical approximation algorithms are interesting.  Under standard complexity theoretic assumptions, we should not expect to be able to solve the problem, so the interesting direction is the study of the approximability of the problem.  Can we find rigorous approximation algorithms, and how well can we expect to be able to approximate the answer?  Moreover, which classes of instances admit constant-factor approximation algorithms?  

The generic $2$-Local Hamiltonian problem is a generalization of several classical optimization problems \cite{W03, B19} with very different approximability.  Maximum independent set is one such example \cite{W03}, and it is well known that such a problem cannot be approximated to within a constant factor unless P=NP, so we should not expect the $2$-Local Hamiltonian problem to have a constant-factor approximation algorithm which holds uniformly for all instances.  In light of this fact researchers make specific assumptions on the terms $H_e$, and attempt to find approximation algorithms under these assumptions.  The $2$-Local Hamiltonian instances we consider generalize a variety of classical optimization problems, including Max $2$-SAT, Max Cut, general Max $2$-CSP, and the Grothendieck problem (see Table~\ref{table:results} for our results). 
%Note that Max 2-SAT is another example, indeed the problem we study here generalizes Max 2-SAT, which can see by mapping the clauses in the instance to diagonal local Hamiltonians (``classical'' Hamiltonians).  A solution to the corresponding instance of Max 2-QSAT yields an optimal solution for the instance of Max 2-SAT (see \Cref{thm:6} in the appendix).  
\subsection{Previous Work}
In the interest of describing classical approximation algorithms for $2$-Local Hamiltonian, let $OPT=\lambda_{max}(\sum_{e \in E} H_e)$ be the largest eigenvalue of an instance of $2$-Local Hamiltonian, $H=\sum_{e \in E} H_e$, and let 
\begin{equation*}
OPT_{prod}=\max_{\substack{\ket{\phi_1}, ..., \ket{\phi_n}\\ \in \mathbb{C}^2}} \bra{\phi_1} \otimes ... \otimes \bra{\phi_n} H \ket{\phi_1} \otimes ... \otimes \ket{\phi_n}
\end{equation*}
\noindent be the product state\footnote{As is clear from the expression, a product state is a quantum state which factors according to tensor product of individual quantum states.  Such states have no entanglement and are considered ``classical'' states.  } with the largest objective value or \emph{energy}.

One common assumption is on the {\it geometry} of the interactions in $E$.  Bansal, Bravyi, and Terhal show that $2$-Local Hamiltonian on bounded-degree planar graphs admits a polynomial-time approximation scheme\footnote{This is an approximation algorithm that allows an arbitrarily good, but constant, approximation factor at the expense of an increase in runtime.} (PTAS) \cite{B07}, and Brand\~ao and Harrow generalize this to arbitrary planar graphs~\cite{B16}.  On the other end, for $k$-Local Hamiltonian on dense graphs, Gharibian and Kempe give a PTAS with respect to $OPT_{prod}$~\cite{G12}, and Brand\~ao and Harrow extend this result to obtain a PTAS for dense graphs with respect to $OPT$~\cite{B16}. Brand\~ao and Harrow also show the existence of good product-states or give product-state approximations for a variety of graph classes~\cite{B16}.

Many authors make assumptions on the form of the terms $H_e$.  One common assumption is that each $H_e$ is traceless \cite{B19, H17}, or equivalently that they each can be written as a linear combination of tensor products of Pauli operators, excluding the identity operator.  Unfortunately, this case is still general enough to capture problems with no constant factor approximation algorithms \cite{A05} (under complexity theoretic assumptions), although Bravyi, Gosset, K{\"o}nig, and Temme give an approximation algorithm for traceless $2$-Local Hamiltonian with guarantee that depends inverse logarithmically on the problem size~\cite{B19}. Harrow and Montanaro give an approximation algorithm for traceless $k$-Local Hamiltonian with respect to the maximum degree and size of the interaction hypergraph~\cite{H17}. Other authors make assumptions which force particular physically relevant forms for the terms $H_e$ \cite{G19}, or on the rank of the terms $H_e$ \cite{H20}.  

It is these latter works which our work is most easily compared to, so we provide descriptions of them here (the approximation guarantees are specified in \Cref{table:results}).  A unifying theme among them is that they rely upon a semi-definite program (SDP) to provide an upper bound on $OPT$ and then use generalization of a classical randomized rounding scheme to produce a product state~\cite{B19,G19,H20}.  Such an approach was first carried out by Brand\~ao and Harrow~\cite{B16}.

The work of Hallgren, Lee, and Parekh \cite{H20} assumes that the terms $H_e$ are all PSD, and that each one is some projector of fixed rank, which is also our primary problem of interest.  They solve an SDP relaxation of $OPT$ and use the solution they get to ``round'' to a valid quantum product state, say $\ket{\psi}\bra{\psi}$.  Employing results used in classical SDP rounding algorithms~\cite{B10, G95} in a black-box fashion, along with a rounding scheme designed to handle $1$-local terms, the authors are able to show that the quantum state from the rounding algorithm is within some fraction of the optimal state, or that $\bra{\psi} \sum_e H_e \ket{\psi} \geq \alpha  OPT$.  
%%% BEGIN COMMENTED OUT
\iffalse
This constitutes an approximation algorithm for the best {\it product state} but does not qualify as an approximation algorithm for the best generic quantum state, so to complete their work they use an existence result of Gharibian and Kempe \cite{G12} which shows that for a $2$-local qubit problem the best product state achieves energy at least $OPT/2$.  The resulting approximation factor is the product of the two, which is $\alpha/2$.  Naturally, $\alpha<1$ (otherwise they would be solving an NP-hard problem) so the approximation factor from their methods is at best $1/2$.
\fi
%%% END COMMENTED OUT
The approximation factors they obtain are worse than the maximally mixed state (random assignment) for all cases except when $rank(H_e)=1$, and here they are able to get the first approximation factor beating random assignment. They also provide an approximation algorithm when each $H_e$ is a product state, which is a QMA-hard class of $2$-Local Hamiltonian.

Another result which is particularly relevant to our work is that of Gharibian and Parekh \cite{G19}, who consider a QMA-hard $2$-Local Hamiltonian generalization of the classical Max Cut problem.  Here the authors assume a particular form for the $H_e$ by assuming that there are only specific terms in the Pauli decomposition.  
They use a simpler SDP relaxation than~\cite{H20}, and they round a solution of this SDP to a product state using a result of Bri{\"e}t, de Oliveira Filho, and Vallentin~\cite{BJ10} in a black box fashion to prove the approximation factor.  Improved approximation results for this problem, that go beyond using product states, have been recently obtained by Anshu, Gosset, and Morenz \cite{A20}.  

For traceless Hamiltonians, the main result of interest is that of Bravyi, Gosset, K{\"o}nig, and Temme~\cite{B19}.  This work builds off \cite{C04, H17} and provides an inverse logarithmic approximation factor for generic traceless Hamiltonians. Such Hamiltonians necessarily have a non-negative and a non-positive eigenvalue. Note that such a class generalizes all problems considered in this paper, since adding copies of the identity can only make the approximation factor better, however there is no reason to expect their analysis could be used to prove constant factor approximations for the classes we study.  The SDP relaxation is the same as \cite{G19}; however, the rounding scheme and analysis is a generalization of a classical approximation by Charikar and Wirth~\cite{C04}.
%They round to Bloch vectors with components in $\{\pm 1/\sqrt{3}\}$, much like in the standard Grothendieck problem \cite{B11}.  

\onote{Make the table more readable and better explain problems where appropriate}
\begin{table}
{\centering
\caption{\textbf{Summary of our and related results.}  The number of qubits or Boolean variables is $n$.  For readability, we omit weights $w_{ij} \geq 0$ that may be present in both $2$-local Hamiltonian ($2$-LH) and $2$-CSP problems.  \Cref{sec:classical-quantum} provides more details on the relationship between $2$-LH and $2$-CSP, as well as definitions for $X_i,Y_i,Z_i$.  An ``N'' denotes a numerical result, and the classical results are implicitly numerical. The abbreviation ``quad.''\ refers to strictly quadratic instances.}

\label{table:results}
\begin{tabular}[t]{@{}p{0.32\linewidth}p{0.24\linewidth}p{0.19\linewidth}p{0.22\linewidth}@{}}
\toprule
Max $2$-LH problem \newline
(QMA-hard)&
Max $2$-CSP \newline
specialization \newline
(NP-hard)&
Classical approx. \newline
for $2$-CSP&
Classical approx. \newline
for $2$-LH problem \newline
(product state)\\
\midrule
\textbf{Traceless}\newline
$\sum_{ij \in E} H_{ij}\otimes \mathbb{I}_{[n]\setminus\{ij\}}$\newline
$H_{ij}$ has no $\mathbb{I}$ terms&
\textbf{Classical Ising} \newline
max -$\sum_{ij \in E} z_iz_j$ \newline
$z_i \in \{\pm 1\}$&
$\Omega(\frac{1}{\log n})$\cite{C04}&
$\Omega(\frac{1}{\log n})^\dagger$\cite{B19}\\
&&&\\

\textbf{Bipartite Traceless}\newline
$\sum_{e \in E} H_{ij}\otimes \mathbb{I}_{[n]\setminus\{ij\}}$\newline
$H_{ij}$ has no $\mathbb{I}$ terms\newline
$E$ bipartite&
\textbf{Grothendieck}\newline
max -$\sum_{ij \in E} z_iz_j$ \newline
$z_i \in \{\pm 1\}$\newline
$E$ bipartite&
$0.561+\varepsilon$\newline\cite{B13}&
\textbf{$0.187^\dagger$ (quad.)}\\
&&&\\

\textbf{Positive/Rank 1}\newline
$\sum_{ij \in E} H_{ij}\otimes \mathbb{I}_{[n]\setminus\{ij\}}$\newline
$\mathbb{I} \succeq H_{ij} \succeq 0$\newline
($\equiv H_{ij}$ rank 1 projector)&
\textbf{Max $2$-CSP}\newline
($\equiv$ 1 satisfying\newline
assignment\newline
per clause)&
0.874~\cite{L02}&
0.25 (random)\newline
0.328~\cite{H20}\newline
\textbf{0.387}\newline
\textbf{0.467 (quad.)}\newline
\textbf{0.498 (quad., N)}\newline
0.5 (upper bound)\\
&&&\\

\textbf{Max Heisenberg}\newline
$\sum_{ij\in E} \mathbb{I} - X_iX_j - Y_iY_j - Z_iZ_j$\newline
(special case of above)&
\textbf{Max Cut}\newline
max $\sum_{ij \in E} 1-z_iz_j$\newline
$z_i \in \{\pm 1\}$&
0.878~\cite{G95}&
0.25 (random)\newline
0.498~\cite{G19}\newline
0.5 (upper bound)\newline
0.53*~\cite{A20}\\
&&&\\

\textbf{Rank 2}\newline
$\sum_{ij \in E} H_{ij}\otimes \mathbb{I}_{[n]\setminus\{ij\}}$\newline
$H_{ij}$ rank 2 projector&
\textbf{Max $2$-CSP}\newline
with 2 satisfying\newline 
assignments/clause&
0.874\cite{L02}&
0.5 (random)\newline
\textbf{0.565}\newline
\textbf{0.639 (quad.)}\newline
\textbf{0.653 (quad., N)}\newline
0.667 (upper bound)\\
&&&\\

\textbf{$2$-QSAT}\newline
$\sum_{ij \in E} H_{ij}\otimes \mathbb{I}_{[n]\setminus\{ij\}}$\newline
$H_{ij}$ rank 3 projector&
\textbf{Max $2$-SAT}\newline
($\equiv$ 3 satisfying\newline
assignments/clause)&
0.940~\cite{L02}&
0.75 (random)\newline
\textbf{0.764}\newline
\textbf{0.805 (quad.)}\newline
\textbf{0.821 (quad., N)}\newline
0.834 (upper bound)\\

\bottomrule
\end{tabular}}
*This exceeds the product-state upper bound because it is achieved by a classical approximation algorithm that rounds to a non-product state.\\
$^\dagger$For any traceless $2$-LH problem, we obtain a product-state approximation ratio that is $\frac{1}{3}$ of an approximation ratio for a related classical CSP, using the appropriate classical approximation algorithm as a black box (see \Cref{sec:21}).
\end{table}

\subsection{Overview of Our Work}\label{sec:6}
\onote{use \textbackslash paragraph's to enhance organization and readability}
\onote{we need to make this consistent with the rest}
There are two problems of interest to us.  The first is the general problem of finding the largest eigenvalue of a $2$-local traceless Hamiltonian on a bipartite interaction graph, and the second problem is finding the largest eigenvalue of a $2$-local Hamiltonian where all the local terms are projectors (eigenvalues are $0$ or $1$).  

\paragraph{Traceless Hamiltonians.}
For the traceless case, we consider Hamiltonians on bipartite interaction graphs that are ``strictly quadratic.'' Informally, the latter means that the $2$-local terms of the Hamiltonian do not contain any implicit $1$-local terms (see \Cref{def:2} in \Cref{sec:formal-problem-statement}).  The classical analog of a strictly quadratic traceless Hamiltonian is a multilinear quadratic polynomial that does not contain linear terms (see \Cref{sec:classical-quantum} for connections between Hamiltonians and multilinear polynomials).     
%Qualitatively, a Hamiltonian is bipartite if the observables present in the Hamiltonian can be partitioned into two sets such that an observable in one set only ever ``sees'' observables in the other set through the Hamiltonian (see \Cref{sec:21} for details).  A Hamiltonian is strictly quadratic if all its composite terms correspond to the interaction of two observables.  In other words, no observable ever contributes ``on its own'' to the objective.  

A natural classical analog of this problem is the symmetric Grothendieck problem \cite{BJ10, F20}.  In this problem the objective is to maximize a quadratic form of a set of variables, $\mathbf{z}^T A \mathbf{z}$, subject to the constraint that each of the variables $\mathbf{z}_i \in \{\pm 1\}$, and where we assume diagonal elements of $A$ are $0$.  The strictly quadratic nature of the problem is apparent, since the objective is a quadratic form.  It is also apparent that $A$ is traceless, but this is not the real reason the analogy is appropriate.  Note that, since $\mathbf{z}_i \in \{\pm 1\}$, we will always pick up the diagonal elements of $A$:  $\mathbf{z}_i^2 A_{ii}=A_{ii}$.  Hence, finding the exact solution to the case where $\text{Tr}[A]\neq 0$ is equivalent to finding the exact solution when $\text{Tr}[A]=0$.  On the quantum side, there is a standard decomposition for local Hamiltonians such that $H=\alpha \mathbb{I} +\beta \mathcal{O}$, where $\mathcal{O}$ is some traceless operator.  Since we are trying to find $\max \bra{\psi} (\alpha \mathbb{I} +\beta \mathcal{O}) \ket{\psi}$ for normalized $\ket{\psi}$, we will always pick up the constant $\alpha$ independent of our choice of $\ket{\psi}$, hence we get an equivalent optimization problem for any value of $\alpha$.  
%The argument is the same in the $\text{Tr}[H]=0$ and the $\text{Tr}[H] \neq 0$ case.  
This is why the analogy is appropriate, we are optimizing over a quadratic form with the extra identity contribution subtracted off.

%Indeed, the analogy is so appropriate that 
The rounding algorithm we use is a simple modification of a known algorithm \cite{BJ10} for solving a variant of the symmetric Grothendieck problem (which also uses the bipartite assumption).  Since we are able to use \cite{BJ10} in a black-box fashion, the technical details of that algorithm are not needed for this paper.  Essentially, we use this result to obtain $\{\pm 1\}$ variables such that the objective upper bounds the optimal quantum objective. Rounding to a quantum state is then easily accomplished by dividing these variables by a large enough constant that they can be taken to be Bloch vectors for a valid quantum state.  In the end, we obtain a $\frac{2\ln(1+\sqrt{2})}{3\pi}$-approximation algorithm.  Note that the best previously known result is a $\Omega(\frac{1}{\log(n)})$-approximation by Bravyi, Gosset, K{\"o}nig, and Temme~\cite{B19} on general graphs.  More generally, our approach allows one to obtain a product-state approximation algorithm for traceless instaces by using an approximation algorithm for a related classical CSP as a black box, losing a factor of $3$ in the approximation ratio.  This also gives a more direct means of obtaining the result of Bravyi, Gosset, K{\"o}nig, and Temme~\cite{B19}. Since this result is disparate from our main results, we present both the formal statements and analysis in \Cref{sec:21}.

\paragraph{$2$-Local projectors.}The main contribution of this work is the second problem we mentioned, finding the largest eigenvalue of a $2$-local Hamiltonian where each of the local terms are projectors.  In order to understand our rounding algorithm, we must first understand the local nature of the $2$-local Hamiltonian problem.  Let $H=\sum_e H_e$ be the $2$-local Hamiltonian where each $H_e$ is a local term affecting only two qubits, say qubits $i$ and $j$.  Let $\rho \in \mathbb{C}^{2^n \times 2^n}$ be the optimal density matrix.  The value of the objective can be calculated as $\text{Tr}[\rho H]= \sum_e \text{Tr}[\rho H_e]=\sum_{e} \text{Tr}[\rho_{ij} \widetilde{H_e}]$, where $\rho_{ij} \in \mathbb{C}^{4 \times 4}$ is the marginal density matrix on qubits $i,j$, and $\widetilde{H_e} \in \mathbb{C}^{4 \times 4}$ is the $2$-local part of the term $H_e \in \mathbb{C}^{2^n \times 2^n}$ that acts on $i,j$.  Here $\widetilde{H_e}$ plays a role analogous to a classical Boolean constraint on 2 variables. Hence, if we were given the set of marginals $\{\rho_{ij}\}$, and descriptions of the $\widetilde{H_e}$, we could calculate the objective on a polynomial sized classical computer since $\rho_{ij}, \widetilde{H_e}\in \mathbb{C}^{4\times 4}$.  This implies that if we were able to optimize the set $\{\rho_{ij}\}$ subject to the constraint that the $\rho_{ij}$ were valid marginals of a global density matrix, we would be able to solve the local Hamiltonian problem.  Indeed, the issue here is that deciding if a set of marginal density matrices is globally consistent is itself a QMA-complete problem \cite{L06, Br19}.  

A natural question is then whether or not, in polynomial time, we can impose some global constraint which is weaker than consistency.  Optimizing with this constraint would provide a relaxation, hence a polynomial-time-computable upper bound on the optimal objective.  Then one might use generalizations of standard classical recipes~\cite{G95} for deriving an approximation algorithm from the relaxation.  In the classical case, a problem of this form can be constrained using a  semidefinite constraint on a ``moment matrix'' (e.g.\ \cite{R12}).  A moment matrix ``tracks'' low order statistics of a global probability distribution.  These statistics have the property that they are only a function of marginal distributions and that the moment matrix is guaranteed to be PSD if the marginals are consistent (although the converse does hold in general).  Hence, defining the moment matrix and forcing it to be PSD gives a weaker condition than global consistency, which can easily be checked.  We will adopt the same approach, except we will be tracking local {\it quantum} statistics.  Our approach is related to existing hierarchies of quantum moment matrices~\cite{D08,P10}. We will have, as variables in our optimization problem, a set of marginal distributions $\{\rho_{ij}\}$, as well as an overall moment matrix $M$, where each entry of $M$ can be evaluated using at most one specific $\rho_{ij}$.  Just as in the classical case, the matrix $M$ will have the property that $M\semigeq 0$ for a consistent set of marginals.  

Once we have the relaxation, we can efficiently solve it (it will be a polynomially-large SDP) to obtain an optimal moment matrix $M^*$ and optimal marginals $\{\rho_{ij}^*\}$ for the relaxation.  The solution will have the property that $\{\rho_{ij}^*\}$ likely represent a globally inconsistent set of marginals, and the objective will be larger than the objective for the optimal quantum state (recall we are applying a condition that is weaker than consistency).  The task for us is to then generate a set of consistent density matrices $\{\rho_{ij}\}$ with quantifiable loss in objective.  The loss of objective will correspond to an approximation factor, as is the theme in many works.  The approach we take is to (randomly) generate single qubit marginals $\{\rho_i\}$ from the set of two qubit marginals $\{\rho_{ij}^*\}$ and output the consistent density matrix $\bigotimes_i \rho_i$.  Note that, just as in the classical case, non-overlapping marginals can always be assumed consistent.  Understanding the loss in objective then reduces to understanding the loss of objective due to the random rounding procedure. 

\paragraph{Analysis.}The analysis proceeds as one might expect, by linearity we can reduce the expectation 

\noindent ${\mathbb{E}[\text{Tr}[\bigotimes_i \rho_i\ \sum_e H_e]]=\sum_e \mathbb{E}[\text{Tr}[\bigotimes_i \rho_i\ H_e]]}$.  Then, a bound on the ``worst-case'' edge provides a bound on the expectation overall.  The issue with accomplishing this directly is the number of parameters involved.  Arguments concerning classical problems do not have to contend with this.  For example, Max Cut has all terms proportional to $(1-z_i z_j)$ for $z_i, z_j$ scalar variables.  Max $2$-SAT has $4$ kinds of clauses depending on negation of the variables; it has nowhere near the variability of a generic $2$-local projector.  Naively, it is determined by $16$ parameters, since $\widetilde{H_e}$ is a $4\times 4$ matrix.  Therefore the first task in our analysis is to reduce the terms $\widetilde{H_e}$ to a standard form, specified by a small number of parameters.  This is accomplished with some singular value decompositions, as well as exploiting the rotational invariance of standard multivariate normal distributions, and applying some results concerning $2$-qubit density matrices or projectors \cite{G16}.  Given the standard form, calculating the expected objective reduces to calculating the expectations of certain functions of multivariate normal variables.  This is analyzed using expansions in Hermite polynomials, and the resulting expressions are bounded or given in terms of special functions.  The analysis itself is partially inspired by \cite{BJ10, B18}, and we are able to obtain a result of \cite{BJ10} as a special case of our analysis.  The final results we obtain are summarized here informally:

\begin{theorem}[Informal]\label{thm:8}
Given a $2$-local Hamiltonian problem $\{H_e\}$ where all $H_e$ are proportional to $2$-local projectors with $\widetilde{H_e}$ of rank $k\in \{1, 2, 3\}$, we give a classical randomized polynomial-time algorithm with approximation ratio $\alpha(k)$ where
\begin{align*}
\alpha(k) =\begin{cases}
0.387 \text{ if $k=1$}\\
0.565 \text{ if $k=2$}\\
0.764 \text{ if $k=3$}.
\end{cases}
\end{align*}
\end{theorem}

\begin{theorem}[Informal]
If in addition to the assumptions of \Cref{thm:8}, the terms $H_e$ are strictly quadratic, we give a classical randomized polynomial-time algorithm with approximation ratio $\alpha(k)$, where
\begin{align*}
\alpha(k) =\begin{cases}
0.467 \text{ if $k=1$}\\
0.639 \text{ if $k=2$}\\
0.805 \text{ if $k=3$}.
\end{cases}
\end{align*}
\end{theorem}

Projectors on $2$ qubits have rank at most $4$, so there are three cases of interest for this problem, rank $1$, rank $2$, and rank $3$.  We obtain novel approximation factors for each case.  The decision version of the problem we consider is known as Quantum-SAT and was introduced in 2006 by Bravyi \cite{Br11}, and the approximability of Max Quantum-SAT was first considered in 2011 by Gharibian and Kempe \cite{G12}, who observed that the maximally mixed state trivially achieves an approximation ratio of $k/4$, where $k \in \{1,2,3\}$ is the rank of the projectors $\widetilde{H_e}$.  The only nontrivial result previously known is a 0.328-approximation for the $k=1$ case by Hallgren, Lee, and Parekh \cite{H20}.  The main difference between our work and the previous works \cite{G19, H20} is that we are able to directly analyze the expectation rather than appealing to other works in a black-box fashion.  Both of these works appealed to approximation results of Bri{\"e}t, de Oliveira Filho, and Vallentin \cite{B10, BJ10}.  Our analysis may be seen as a generalization of the result of Bri{\"e}t, de Oliveira Filho, and Vallentin~\cite{BJ10} employed in \cite{G19}.  \knote{COMMENTED OUT:  Along these lines there are two observations that we use.  The first (\Cref{lem:6}) is that the expected value of an edge under the rounding algorithm can be reduced to a simple form with some elementary considerations and a singular value decomposition.  The second, inspired by work on hyperplane rounding \cite{B18, Ba16}, is that this simple form for the expectation can be expanded in the Hermite polynomials.  Under this expansion, some simple calculus allows us to bound the value of the expectation and get good bounds on the approximation factor.  I THINK ITS REPETITIVE GIVEN THE PARAGRAPH BEFORE THE INFORMAL THEOREMS}

%These results apply to the more general Quantum-SAT problem in addition to the quadratic case we consider.  Our main result is captured in the following informal version of Theorem~\ref{thm:1}.  For our first result, we will take only the assumption taht the $H_e$ terms are projectors.  Under this assumption, we provide the following 

\subsection{Outlook and Conjectures}

\paragraph{Strictly quadratic instances.} We believe the strictly quadratic case is an interesting special case for several reasons.  As noted, one of the difficulties in analyzing rounding schemes for $2$-Local Hamiltonian is the sheer number of parameters involved.  The quadratic case reduces the number of parameters to consider, while still including physically relevant QMA-hard instances such as the Max Heisenberg model that serves as a quantum generalization of Max Cut~\cite{G19}.  Indeed we believe that quadratic instances allow one to glean insights and develop techniques that might otherwise be obscured in more general instances.  One of the first rigorous approximation algorithms for a $2$-local Hamiltonian that goes beyond product states was developed for a quadratic instance, and one that has considerably fewer parameters than those we consider here~\cite{A20}.   Moreover, maximally entangled instances are quadratic, and we conjecture these are the hardest cases to approximate.

Additionally, we believe that the analysis we provide is tighter for this case.  Since there are no linear terms in the objective, we can focus on one particular Hermite expansion and carefully bound it.  For the general case we must contend with understanding the linear and quadratic terms simultaneously, which makes the analysis much more difficult.  We conjecture that the true performance of our algorithm for the general case is:
\begin{conjecture}\label{conj:performance:informal}(Informal) Our rounding algorithm achieves approximation ratio:
\begin{align*}
\alpha(k) =\begin{cases}
0.498 \text{ if $k=1$}\\
0.653 \text{ if $k=2$}\\
0.821 \text{ if $k=3$}.
\end{cases}
\end{align*}
\end{conjecture}
%For classical $2$-CSPs, some of the approximation algorithms offering the best-known approximation guarantees, based on rounding solutions of SDP relaxations, were initially verified only by numerical means, while more recent efforts have attempted to mitigate this \cite{SJ09}.  For $2$-CSPs, such numerical results are accepted by the community as the numerical optimization to find a worst-case example typically only involves one parameter, as one may focus on constructing a single worst-case edge.  In the case of $2$-Local Hamiltonians, numerical verification becomes a riskier prospect as a n{\"a}ive analysis of a single edge involves 15 SDP variables and 15 parameters from the objective function.  The quadratic case of $2$-Local Hamiltonian we consider involves 9 SDP variables and 9 objective-function parameters for an edge.  One of our main contributions is a means to to perform such an analysis.

We have indeed attempted to verify the approximation ratios in Conjecture~\ref{conj:performance:informal} numerically, and via both sampling and integration methods, we observe that the worst-case instances for our algorithm match the values stated in the conjecture.  The difficulty in taking these encouraging results as fact is that sampling is of course not exhaustive, and the integrals exhibit poor convergence.  We do give upper bounds on $\alpha(k)$ in Theorem~\ref{thm:gap-examples} in the Appendix.  These are derived by furnishing instances on 2 qubits that demonstrate a gap between the maximum eigenvalue and the maximum objective value achieved by a product state.  These bounds are $\frac{1}{2}$ ($k=1$), $\frac{2}{3}$ ($k = 2$), and $\frac{5}{6}$ ($k = 3$), which are fairly close the values in Conjecture~\ref{conj:performance:informal}.  We suspect that one may be able to prove that the approximation ratios in Conjecture~\ref{conj:performance:informal} are best possible under the unique games conjecture.

%Our methods are most similar to \cite{G19}, but we make some important changes in the analysis and the rounding algorithm which allows us to prove approximation factors for less constrained problems.  The output of the SDP is a ``moment matrix'' which is some object which may or may not correspond to a quantum state.  It is exactly the task of the rounding algorithm to use this object to produce a valid quantum state with objective not that much worse.  We alter the SDP to force the moment matrix which we get out of the SDP to have valid ``marginal'' distributions on all pairs of qubits, as well as alter the way in which the moment matrix is converted to a product state.  

\paragraph{Significance of our work.} We give the first approximation algorithm beating random assignment for Max $2$-QSAT (and related problems).  We show how to move beyond numerical evaluation of approximation ratios for $2$-local Hamiltonian problems, which is not as critical in the classical case that enjoys only a handful of parameters.  This is accomplished by bringing analysis employing Hermite polynomials to bear on $2$-local Hamiltonian problems.   Moreover, our rounding scheme is a simple and natural generalization of hyperplane rounding.

\section{Semidefinite Relaxation and Rounding Approach} \label{sec:7}

In this section we present a rigorous but high-level overview of our approach, with technical lemmas deferred to later sections.  We define the main problems considered and our semidefinite relaxation and rounding algorithm.  We conclude by motivating the analysis that will occur in subsequent sections. 

\subsection{Quantum Information Notation}\label{sec:quantum-notation}
We adopt some standard notations used in quantum information~\cite{N10}.  The \emph{kets} $\ket{0}:=[0,1]$ and $\ket{1}:=[1,0]$ represent the standard basis vectors for $\mathbb{C}^2$, while the \emph{bras} $\bra{0}$ and $\bra{1}$ represent their conjugate transposes.  The $d \times d$ identity matrix is denoted by $\mathbb{I}_d$, and the subscript will be omitted when redundant.  We obtain the standard bases for $\mathbb{C}^{2^n}$ as $\ket{b_1b_2\ldots b_n} := \ket{b_1}\ket{b_2}\ldots\ket{b_n} := \ket{b_1}\otimes\ket{b_2}\otimes \ldots \otimes \ket{b_n}$, with $b_i \in \{0,1\}$. The Pauli matrices will have the usual definition:
\begin{equation}
\label{eq:paulis}
 \sigma^0=\mathbb{I}=\begin{bmatrix}
1 & 0 \\
0 & 1
\end{bmatrix},
\,\,\,\,\,\,
\sigma^1=\begin{bmatrix}
0 & 1 \\
1 & 0
\end{bmatrix},
\,\,\,\,\,\,
\sigma^2=\begin{bmatrix}
0 & -i \\
i & 0
\end{bmatrix}, \,\text{and}
\,\,\,\,\,\,
\sigma^3=\begin{bmatrix}
1 & 0 \\
0 & -1
\end{bmatrix}.
\end{equation}
\noindent We will generally use subscripts to indicate quantum subsystems.  If $\rho$ is a density matrix on $n$ qubits, for instance, $\rho_{ij}$ will correspond to the marginal density matrix on qubits $i$ and $j$, i.e. the partial trace $\rho_{ij}=\text{Tr}_{[n]\setminus \{i, j\}}[\rho]$ (e.g.~\cite{N10}, Section 2.4.3).  Similarly, $\sigma_i^j$ corresponds to Pauli $j$ on qubit $i$.  Subscripts will supercede position in many cases in the paper, for instance $\sigma_i^1 \otimes \sigma_j^2 \otimes \mathbb{I}_{[n]\setminus \{i, j\}}$ is meant as $\mathbb{I}\otimes \mathbb{I} \otimes ... \otimes \sigma^1\otimes ... \otimes \mathbb{I} \otimes \sigma^2 \otimes ... \otimes \mathbb{I}$ where $\sigma^1$ is at the $i$th position and $\sigma^2$ is at the $j$th position.  We encourage readers familiar with classical constraint satisfaction problems to consult \Cref{sec:classical-quantum}, which casts such problems as quantum local Hamiltonian problems.

\subsection{Formal Problem Statement} \label{sec:formal-problem-statement}
In the $2$-Local Hamiltonian problem we are given as input the problem size, $n$, and a list of classical descriptions of $2$-local terms, $\{H_e\}_{e \in E}$.  We allow multiedges, i.e.\ distinct edges $e$ and $e'$ on the same pair of qubits $i,j$; we will use the notation $e_1$ and $e_2$ to refer to the qubits on which $e$ acts.
In this context a term $H_e$ is $2$-local if it can be written in the form $H_e=\mathcal{O}_e\otimes \mathbb{I}_{[n]\setminus\{e_1, e_2\}}$, using the subscript notation from \Cref{sec:quantum-notation}.  Local Hamiltonians have polynomially-sized descriptions which can be given in terms of the local operators $\mathcal{O}_e$, but for our purposes the details of the description will not be important.  It is important to note at this point that we will use $rank(H_e)$ to mean $rank(\mathcal{O}_e)$.  The actual rank of $H_e$ is $rank(\mathcal{O}_e)2^{n-2}$, but for ease of exposition we will say that the ``rank'' of a $2$-local term is equal to the rank of its non-trivial part.  We are tasked with determining the largest eigenvalue of the Hamiltonian $H=\sum_{e \in E} H_e$:
\begin{problem}[QLH($n$, $\{H_e\}_{e \in E}$)]\label{prob:QLH}
Given a problem size, $n$, as well as a classical description of a set of $2$-local terms $\{H_e\}$ with $H_e \in \mathbb{C}^{2^n \times 2^n}$ Hermitian, find:

\begin{equation*}
\lambda_{max}(H):= \max_{\ket{\phi}\in [\mathbb{C}^2]^{\otimes n}} \text{Tr} \left[\sum_{e \in E} H_e \ \ket{\phi}\bra{\phi} \right]=\max_{\substack{\rho \in \mathbb{C}^{2^n}\times \mathbb{C}^{2^n}\\Tr(\rho)=1, \ \rho \semigeq 0}} \text{Tr} \left[\sum_{e \in E} H_e\ \rho \right].
\end{equation*}
\end{problem}

The main problems of interest to us are instances of QLH where each term is a projector and the special case where each projector is strictly quadratic.  The strictly quadratic case precludes non-identity $1$-local terms (i.e. $\mathcal{O}_i \otimes \mathbb{I}_{[n]\setminus\{i\}}$ with $\mathcal{O}_i \not= \mathbb{I}$) that may be implicit in a $2$-local term.  
\begin{definition}[Strictly Quadratic]\label{def:2}
Let $H_e$ be a $2$-local term on $n$ qubits.  Write $H_e=\mathcal{O}_e\otimes \mathbb{I}_{[n]\setminus \{e_1, e_2\}}$ for some nontrivial operator $\mathcal{O}_e$.  Express $\mathcal{O}_e$ in the Pauli basis as:
\begin{equation}\label{eq:strictly-quadratic}
\mathcal{O}_e=\sum_{k, l=0}^3 \alpha_{k, l} \sigma^k \otimes \sigma^l.
\end{equation}
We say that $H_e$ is a strictly quadratic if $\alpha_{k, 0}=0$ for all $k\neq 0$, and $\alpha_{0, l}=0$ for all $l\neq 0$.
\end{definition}
\noindent Note that the coefficients in \Cref{eq:strictly-quadratic} may be obtained as $\alpha_{k, l} = \text{Tr}[\sigma^k \otimes \sigma^l\ \mathcal{O}_e]/4$ and are real since $\mathcal{O}_e$ is Hermitian.
%To illustrate the meaning of ``strictly quadratic'' in this context, we may express each term $H_{ij}$ in the Pauli basis by writing $H_e=\mathcal{O}_{i, j}\otimes \mathbb{I}_{[n]\setminus \{i, j\}}$ where $\mathcal{O}_{ij}=\sum_{k, l=0}^3 \alpha_{k, l}\sigma^k\otimes \sigma^l$.  We say that $H_{ij}$ is {\it strictly quadratic} if $\alpha_{k, l}=0$ when exactly one of $k$, $l$ is non-zero.

We focus our attention on QLH restricted to projectors and strictly quadratic projectors, both of which remain QMA-hard~\cite{P15}. In this case $\mathcal{O}_e = w_e P_e$, where $P_e$ is a 2-qubit projector, and $w_e \geq 0$ is a weight. There are three interesting cases, depending on the rank of $P_e$. We will obtain approximation factors for each.  
\begin{problem}[QLHP($n$, $k$, $\{H_e\}_{e \in E}$)]\label{prob:QLHP}
Given a problem size, $n$, as well as a classical description of a set of $2$-local terms $\{H_e  = w_e P_e \otimes \mathbb{I}_{[n]\setminus\{e_1,e_2\}}\}$ with $w_e \geq 0$ and $P_e \in \mathbb{C}^{4 \times 4}$ a $2$-qubit projector of rank at least $k$, find $\lambda_{max}(\sum_{e \in E} H_e)$.
\end{problem}
\noindent It is worth mentioning that since any $\mathcal{O}_e \semigeq 0$ can be written as a positive combination of rank 1 projectors, QLHP with $k=1$ captures instances of QLH where each $H_e \semigeq 0$. 
  
\subsection{Semidefinite Relaxation}  

We employ a semidefinite programming relaxation for QLH (\Cref{prob:QLH}) that is a refinement of the now standard SDP relaxation that has been used in designing approximation algorithms~\cite{B16,B19,G19}.  Our relaxation is related to one used by Hallgren, Lee, and Parekh~\cite{H20} and may be viewed as a specialization of noncommutative Lasserre hierarchies proposed for quantum information applications~\cite{D08,P10}.

In this section we assume, for the sake of exposition, that there is a single edge $ij$ on any pair of qubits $i,j \in [n]$; however, the relaxation and rounding algorithm may be readily extended to handle general instances of QLH with multiedges. Suppose we have an instance of QLH on $n$ qubits.  As previously stated, the first set of variables in our SDPs will be marginal density matrices $\{\rho_{ij}\}$.  Since there are $n$ qubits, there are $\binom{n}{2}$ many of these, and each of them is a $4\times 4$ Hermitian matrix.  While we cannot impose global consistency, we can force each $\rho_{ij}$ to be a valid density matrix on its own: $\text{Tr}[\rho_{ij}]=1$ and $\rho_{ij}\semigeq 0$ for all $i,j \in [n]$.  We could also explicitly force overlapping marginals to be consistent on single qubit density matrices, however this will be implicit through our use of moment matrices.  

\paragraph{Moment matrices.} Suppose we have a quantum state on $n$ qubits $\ket{\psi}\in \mathbb{C}^{2^n}$. Consider the $1$-local Pauli operators, $\mathcal{M} = \{\mathbb{I}\} \cup \{\sigma^k_i \otimes \mathbb{I}_{[n]\setminus\{j\}} \mid k \in [3], i \in [n]\}$. We apply each of the $3n+1$ Pauli operators $\mathcal{O} \in \mathcal{M}$ on $\ket{\psi}$ to obtain columns of a matrix $V = \{\mathcal{O}\ket{\psi}\}_{\mathcal{O} \in \mathcal{M}} \in \mathbb{C}^{2^n \times (3n+1)}$. We call $M := V^\dagger V \in \mathbb{C}^{(3n+1)\times(3n+1)}$ the \emph{moment matrix} of $\ket{\psi}$ with respect to $\mathcal{M}$; note that $M$ is Hermitian and $M \semigeq 0$ by construction.  The notation $M(\mathcal{O}, \mathcal{P})$ refers to the entry of $M$ at the row and column corresponding to $\mathcal{O},\mathcal{P} \in \mathcal{M}$ respectively.  We have $M(\mathcal{O}, \mathcal{P}) = \bra{\psi}\mathcal{O}\mathcal{P}\ket{\psi}$, for all $\mathcal{O},\mathcal{P} \in \mathcal{M}$, so that $M$ captures all the $2$-local Pauli statistics of $\ket{\psi}$.  In particular the quantity $\bra{\psi}H\ket{\psi}$ is a linear function of the entries of $M$ for a $2$-local Hamiltonian $H$.  If we let $\mathcal{M}_k$ consist of all the $k$-local tensor products of Paulis, instead of just the $1$-local ones, the corresponding moment matrix $M_k$, of size $O(n^k)$ by $O(n^k)$, includes all the $2k$-local Pauli statistics.  We may obtain SDP relaxations for QLH problems by constructing a relaxed $M_k \semigeq 0$ that satisfies linear constraints of the form $\text{Tr}[A M_k] = b$, that a true moment matrix would satisfy.  This corresponds to the $k$th level of noncommutative Lasserre hierarchies introduced for quantum information~\cite{D08,P10}.  Our approach relaxes $M_1$ and adds additional constraints enforcing positivity of $2$-local marginals; the relaxation we obtain sits between the $k=1$ and $k=2$ levels of the noncommutative Lasserre hierarchy.

\paragraph{SDP Relaxation.} We define a (relaxed) moment matrix $M$, which will track local statistics of the set of marginals $\{\rho_{ij}\}$.    Let $M$ be a symmetric, $(3n+1)\times (3n+1)$ real matrix whose rows and columns correspond to operators in $\mathcal{M}$.  Entries of $M$ will correspond to coefficients of the marginal density matrices $\{\rho_{ij}\}$ in the Pauli basis.  We use the notation $M(\sigma_i^k,\sigma_j^l)$ to refer to entries of $M$ for $i,j \in [n]$ and $k,l \in [3]$; in addition we have a row and column of $M$ indexed by $\mathbb{I}$. We set $M(\sigma_i^k,\sigma_j^l)=\text{Tr} [\sigma^k \otimes  \sigma^l \ \rho_{ij}]$ for $(i, k)$, $(j, l)$ in $[n]\times [3]$.  In addition we set $M(\mathbb{I},\mathbb{I}) = 1$, and $M(\sigma_i^k,\mathbb{I})=\text{Tr}[\sigma^k \otimes \mathbb{I} \ \rho_{ij}]$ for all $(i, k) \in [n]\times [3]$ and $j \in [n]$.  Note that this constraint forces consistent single-qubit marginals since 
\begin{equation*}
\text{Tr}_u[\rho_{iu}]=\text{Tr}_v [\rho_{iv}] \Leftrightarrow \text{Tr}[\sigma^l \otimes \mathbb{I}\ \rho_{iu}]=\text{Tr}[\sigma^l \otimes \mathbb{I}\ \rho_{iv}] \ \forall l.
\end{equation*}

Since $M$ contains all local information of $\{\rho_{ij}\}$, we can use $M$ to evaluate the objective of our SDP.  In this direction, we will define a weight matrix for each edge $H_{ij}=w_{ij}\mathcal{O}_{ij}\otimes \mathbb{I}_{n\setminus \{i, j\}}$, where $\mathcal{O}_{ij} \in \mathbb{C}^{4 \times 4}$, and $w_{ij}$ is a scalar weight. We define the $(3n+1) \times (3n+1)$ matrix $C_{ij}$, which contains the coefficients of $\mathcal{O}_{ij}$ in the Pauli basis:
\begin{gather}\label{eq:52}
    C_{ij}(\sigma_i^k,\sigma_j^l)=C_{ij}(\sigma_j^l, \sigma_i^k)=Tr[\sigma^k \otimes \sigma^l \ \mathcal{O}_{ij}]/8 \quad \forall k,l \in [3],\\
    \notag C_{ij}(\sigma_i^k, \mathbb{I})=C_{ij}(\mathbb{I}, \sigma_i^k)=Tr[\sigma^k \otimes \mathbb{I} \ \mathcal{O}_{ij}]/8 \quad \forall k \in [3],\\
    \notag C_{ij}(\sigma_j^l, \mathbb{I})=C_{ij}(\mathbb{I}, \sigma_j^l)=Tr[\mathbb{I} \otimes \sigma^l \ \mathcal{O}_{ij}]/8 \quad \forall l \in [3],
\end{gather}
and all other entries of $C_{ij}$ are 0.
%\begin{equation}\label{eq:9}
%C^e_{ik, jl}=
%\text{Tr}[ \sigma_i^k \otimes \sigma_j^l \otimes \mathbb{I}_{[n]\setminus\{i, j\}} \mathcal{O}_{ij} ]/8 
%\end{equation}
%Note that the matrix $C^e$ can be partitioned into $3\times 3$ blocks, $C_{i, j}^e$, corresponding to the submatrix of $C^e$ containing $C_{ij, kl}^e$ for all $k$, $l$.  For an edge of the form $H_e=\mathcal{O}_{ij} \otimes \mathbb{I}_{[n]\setminus \{i, j\}}$, only two of the blocks will be nonzero, namely $C_{i, j}^e$ and $C_{j, i}^e$.  As an illustration, consider a term of the form $H_e=\mathcal{O}_{ij}\otimes \mathbb{I}_{2^{n-2}} $.  Let $\mathcal{O}_{ij}$ and the marginal quantum state $\rho_{ij}$ have Pauli decompositions:
To illustrate application of the matrix $C_{ij}$, suppose $\mathcal{O}_{ij}$ and the marginal density matrix $\rho_{ij}$ have Pauli decompositions:
\begin{equation*}
\mathcal{O}_{ij} = \sum_{k, l=0}^3 \alpha_{kl} \sigma^k \otimes \sigma^l \quad \text{and} \quad \rho_{ij}=\frac{1}{4} \sum_{k, l=0}^3 \beta_{kl} \sigma^k \otimes \sigma^l.
\end{equation*}
Since, for $k,l\geq 0$, $(\sigma^k)^2=\mathbb{I}$ and $\text{Tr}[\sigma^k \sigma^l]=0$ when $k\neq l$, the value we gain from edge $ij$, ignoring the weight $w_{ij}$, is written as:
\begin{equation}\label{eq:SDP-objective-value}
\text{Tr} [\mathcal{O}_{ij} \rho_{ij}]=\alpha_{00}\beta_{00}+\sum_{\mathclap{\substack{k, l:\\(k\neq 0) \lor (l \neq 0)}}}  \alpha_{kl}\beta_{kl} = \frac{\text{Tr}[\mathcal{O}_{ij}]}{4}+\text{Tr}[C_{ij} M].
\end{equation}
With these facts in hand, we may finally give the main SDP relaxation in this work:

%\noindent Hence, just as in \cite{G19, B19}, the value of an edge can be evaluated as\knote{Ref didnt like following equation since it is technically redundant, but I like the emphasis it gives}:
%\begin{equation}
%\text{Tr} [H_e \rho]=\frac{rank(\mathcal{O}_{ij})}{4}+\sum_{\substack{k, l:\\(k\neq 0) \lor (l \neq 0)}}  \gamma_{kl}\eta_{kl}=\frac{rank(\mathcal{O}_{ij})}{4}+\text{Tr}[ M C^e].
%\end{equation}

%We need to describe one more set of constraints before we can give the main semidefinite relaxation in our work.  In addition to the variables already described, we will also include a marginal density matrix $\rho_{ij}$ for each pair of qubits in the problem.  This will be used to force any submatrix of $M$ which corresponds to a particular pair of qubits to represent a valid quantum state.  In the relaxation, these marginals will likely be globally inconsistent\footnote{If we were able to impose consistency in our semi definite program we would be able to solve the (QMA-hard) problem.  Obviously we cannot expect to do this and, indeed, determining consistency of density matrices is a well-known QMA-complete problem \cite{L06}}. 

\begin{problem}\label{prob:5}
Given an instance of QLH (\Cref{prob:QLH}) on $n$ qubits with local terms $\{H_{ij} = w_{ij}\mathcal{O}_{ij}\otimes \mathbb{I}_{[n]\setminus \{i, j\}}\}$, let $C_{ij}$ be defined according to \Cref{eq:52} for each $ij \in E$.  Solve the following SDP:
\begin{alignat}{2}
\label{sdp-relax:obj} \max \sum_{ij\in E} &w_{ij} \mathrlap{\left(\frac{\text{Tr}[\mathcal{O}_{ij}]}{4}+\text{Tr}[C_{ij} M]\right)}\\
s.t.\qquad
 \label{sdp-relax:diag-id} M(\mathbb{I},\mathbb{I}) &= 1,\\
 \label{sdp-relax:diag} M(\sigma_i^k,\sigma_i^k) &= 1 \quad && \forall i\in [n] \text{ and } k \in [3],\\
 \label{sdp-relax:diag-block-zero} M(\sigma_i^k,\sigma_i^l) &= 0 \quad && \forall i\in [n] \text{ and } k\not=l \in [3],\\
 \label{sdp-relax:2-marginals} M(\sigma_i^k,\sigma_j^l) &= \text{Tr}[\sigma^k \otimes \sigma^l\ \rho_{ij}] \quad &&\forall ij \in E \text{ and }k,l \in [3],\\
 \label{sdp-relax:1-marginals-i} M(\sigma_i^k, \mathbb{I}) &= \text{Tr}[\sigma^k \otimes \mathbb{I}\ \rho_{ij}] \quad &&\forall ij\in E \text{ and } k\in [3],\\
 \label{sdp-relax:1-marginals-j} M(\sigma_j^l, \mathbb{I}) &= \text{Tr}[\mathbb{I} \otimes \sigma^l\ \rho_{ij}] \quad &&\forall ij\in E \text{ and } l\in [3],\\
 \label{sdp-relax:p_ij-trace} \text{Tr}[\rho_{ij}] &= 1 \quad &&\forall ij \in E,\\
 \label{sdp-relax:p_ij-pos} \mathcal{H}(\mathbb{C}^{4 \times 4}) \ni \rho_{ij} &\semigeq 0 \quad &&\forall ij \in E,\\
 \label{sdp-relax:M-pos} \mathcal{S}(\mathbb{R}^{(3n+1)\times(3n+1)})\ni M &\semigeq 0,
\end{alignat}
where $\mathcal{S}(\cdot)$ and $\mathcal{H}(\cdot)$ refer to the symmetric and Hermitian matrices, respectively.
\end{problem}

%For the case of strictly quadratic projectors, we will use a similar relaxation, the main difference is that we no longer need a special index $0$, since all the terms in the objective are strictly $2-$local.  For this problem, $C^e\in \mathbb{R}^{3n\times 3n}$ is defined by:
%\begin{align}\label{eq:9}
%C^e_{ik, jl}=
%\text{Tr}[ \sigma_i^k \otimes \sigma_j^l \otimes \mathbb{I}_{[n]\setminus\{i, j\}} \mathcal{O}_{ij} ]/8\text{ if $i\neq j$}\\
%\nonumber C^e_{ik, il}=0
%\end{align}
%Similarly, we will only need the moment matrix $M$ to be $3n\times 3n$, since the objective is only a function of the quadratic terms.  The relaxation we obtain is:
%\begin{problem}\label{prob:2}
%Given an instance of \Cref{prob:1} on $n$ qubits with local terms $\{H_e\}$ let $C^e$ be defined according to \Cref{eq:9} for each $e$.  Solve the following SDP:
%\begin{gather}
%\max \,\,\, \sum_{e\in E} c_e \left(\frac{rank(H_e)}{4}+\text{Tr}[C^e M]\right)\\
%\nonumber s.t. \,\,\,\, M_{ij, kl}=\text{Tr}[\sigma_i^j \otimes \sigma_k^l \rho_{ik}]\text{  for all $(i, k) \in [n] \times [n]$ and $(j, l) \in [3]\times [3]$}\\
%\nonumber \text{Tr}[\rho_{jl}]=1\text{  for all $(j, l) \in [n]\times [n]$}\\
%\nonumber M_{ij, ik}=0 \text{   and    } M_{ij, ij}=1\text{ for all $i\in [n]$, $(j, k)\in [3]\times [3]$}\\
%\nonumber M\semigeq 0, \,\,\, \rho_{jl}\semigeq 0 \text{  for all $(j, l) \in [n]\times [n]$}\\
%\nonumber M \in \mathbb{R}^{3n\times 3n}, \text{  symmetric} \\
%\nonumber \rho_{ij}\in \mathbb{C}^{4\times 4}, \text{  Hermitian}
%\end{gather}

%\end{problem}

\begin{theorem}
The mathematical program of \Cref{prob:5} is an efficiently computable semidefinite program that provides an upper bound on $\lambda_{max}(\sum_{ij\in E} H_{ij})$.
\end{theorem}
\begin{proof}
Constraints \eqref{sdp-relax:diag-id}--\eqref{sdp-relax:p_ij-trace} are linear equalities on the entries of PSD matrices $M$ and $\rho_{ij}$ $\forall ij \in E$, hence we do indeed have an SDP.  Since there are polynomially many variables of polynomial size, the usual considerations show computational efficiency, i.e. the program can be solved to arbitrary additive precision in polynomial time \cite{B04}. 

A larger matrix $X \semigeq 0$, consisting of $M$ and the $\rho_{ij}$ as its diagonal blocks may be used to put the SDP into a more standard form (e.g.~\cite{B04}, Section 4.6.2).  Although the $\rho_{ij}$ are complex, the SDP may be solved as a real SDP by appealing to the standard approach of tracking the real and imaginary parts separately and observing $X \semigeq 0$ if and only if
\begin{equation*}
\begin{bmatrix}
\text{Re}(X) & -\text{Im}(X)\\
\text{Im}(X) & \text{Re}(X)
\end{bmatrix}\semigeq 0.
\end{equation*}

Let $\ket{\psi}$ be an eigenvector corresponding to $\lambda_{max}(\sum_{ij\in E} H_{ij})$, and let $\rho_{ij}^*$, $\forall ij\in E$, be the 2-qubit marginal density matrices of $\rho=\ket{\psi}\bra{\psi}$, so that Constraints \eqref{sdp-relax:p_ij-trace} and \eqref{sdp-relax:p_ij-pos} are satisfied for the $\rho^*_{ij}$.  In addition consider the moment matrix $M$ for $\ket{\psi}$ with respect to $\mathcal{M}$, as described above.  The matrix $M$ satisfies Constraints \eqref{sdp-relax:diag-id}, \eqref{sdp-relax:diag}, \eqref{sdp-relax:2-marginals}--\eqref{sdp-relax:1-marginals-j}, and \eqref{sdp-relax:M-pos} by the definition of a moment matrix, since 
\begin{equation}\label{eq:psi-rho-moments}
\bra{\psi}\sigma_i^k \otimes \sigma_j^l \otimes \mathbb{I}_{[n]\setminus\{i,j\}} \ket{\psi} = \text{Tr}[\sigma^k \otimes \sigma^l \ \rho_{ij}^*]\text{, for }0 \leq k,l \leq 3. 
\end{equation} 

Constraint~\eqref{sdp-relax:diag-block-zero} is the only one that remains. Note that the real part of $M$, $M^* := \text{Re}(M) \semigeq 0$ since $M \semigeq 0$. By \Cref{eq:psi-rho-moments}, for any $j \in [n]$ and $k\not= l \in [3]$, $M(\sigma_i^k, \sigma_i^l) = \pm i\text{Tr}[\sigma^m \otimes \mathbb{I} \ \rho_{ij}^*]$, where $m \in [3]\setminus\{k,l\}$.  The quantities in \Cref{eq:psi-rho-moments} are real since tensor products of Pauli operators are Hermitian.  This implies that $M(\sigma_i^k, \sigma_i^l)$ for $k \not= l \in [3]$ is imaginary and more generally that $M^*$ and the $\rho_{ij}^*$ satisfy all the constraints. 

Consider the objective value for this solution, $\sum_{ij\in E} w_{ij} (\text{Tr}[\mathcal{O}_{ij}]/4 + \text{Tr}[C_{ij}M^*])=$
\begin{equation*}
\sum_{ij \in E} w_{ij} \text{Tr}[\mathcal{O}_{ij} \rho_{ij}^*] = 
\text{Tr}\left[\sum_{ij \in E} H_{ij}\ \rho\right] =
\lambda_{max}\left(\sum_{ij\in E} H_{ij}\right),
\end{equation*}
where the first equality follows from \Cref{eq:SDP-objective-value}. It follows that the optimal solution to \Cref{prob:5} has value at least that of the optimal solution of QLH.  
%%% BEGIN COMMENTED OUT
\iffalse
Further, we can impose $ M_{ij, ik}=0$, $M_{ij, ij}=1$ and $M\semigeq 0$ without loss of generality from \cite{G19}\footnote{Technically, \cite{G19} applies to a slightly different relaxation, but it is easy to extend to this case.}\knote{Reason matrix is PSD sounds cryptic, revise}.  The first condition holds essentially because those entries are imaginary for a generic quantum state (and hence do not contribute to the real objective), the second condition holds because $M_{ij, ij}=Tr(\mathbb{I}\rho)$, and the third condition holds because $v^T M v=Tr(S^2 \rho)$ for some Hermitian $S$ and for some globally consistent $\rho$ (see \cite{G19} for details).  We can also assume $M_{0,0}=Tr(\mathbb{I} \rho)=1$ for obvious reasons.

We can see that \Cref{prob:5} is a relaxation of \Cref{prob:1}: Given optimal $\rho^*$ for \Cref{prob:1}, we can set $\rho_{ij}=\rho_{ij}^*$, $M_{ij, kl}=\text{Tr}[\sigma_i^j \otimes \sigma_k^l\rho_{ik}^*]$ for $i\neq k$, $M_{ij, ik}=\delta_{j, k}$, and $M_{0, ij}=Tr[\sigma_i^j \rho]=M_{ij, 0}$ to obtain a valid solution to \Cref{prob:5} (\cite{G19} implies resulting $M$ will be PSD).  
\fi
%%% END COMMENTED OUT
\end{proof}

%Note also that for the strictly quadratic case we opted to not include constraints forcing the marginals to have consistent single qubit density matrices, i.e. $Tr_j(\rho_{ij})=Tr_k (\rho_{ik})$.  We could have included this constraint, it is linear, but our analysis makes no use of it.  The reason for this is clear once we understand that constraints of this form track $1$-moments of the density matrices.  Indeed,
%$$
%Tr_j(\rho_{ij})=Tr_k (\rho_{ik}) \Leftrightarrow Tr(\sigma^l_i\otimes \mathbb{I}_j \,\,\rho_{ij})=Tr(\sigma_i^l \otimes \mathbb{I}_k \,\, \rho_{ik}) \,\,\, \forall l.
%$$
%The quadratic nature of the objective allows us to effectively ignore the $1$-Local terms, which is why these constraints can be disregarded.  For the same reasons we could have removed the constraint $M_{ij, ik}=0$ for $j\neq k$ from this relaxation, however we have left them in so the reader realizes the Cholesky vectors corresponding to a single qubit can be taken as orthonormal, without loss of generality.  Our analysis makes no use of this, but we believe it is a potentially very useful constraint.  Extensions of these ideas may crucially require it.  \Cref{prob:5} does implicitly force consistency of single qubit density matrices, since the relaxation forces $M_{kl, 0}$ terms to correspond to $1-$local moments.  

\subsection{Rounding Approach and Formal Statement of Results}

\paragraph{Overview.} In classical SDP-based rounding schemes, one typically seeks to randomly ``round'' unit vectors $\mathbf{v}_i \in \mathbb{R}^d$ to scalars $z_i \in \{\pm 1\}$ so that the expected value of $z_i z_j$ approximates $\mathbf{v}_i\cdot\mathbf{v}_j$.  The seminal hyperplane rounding scheme of Goemans and Williamson~\cite{G95} achieves this by selecting a random unit vector $\mathbf{r} \in \mathbb{R}^d$ and setting $z_i = \mathbf{r} \cdot \mathbf{v}_i / |\mathbf{r} \cdot \mathbf{v}_i |$.

Rounding solutions from SDP relaxations for QLH to product states generalizes this approach.  Recall that a product state has the form $\ket{\psi} = \ket{\psi_1}\otimes \ldots \otimes \ket{\psi_n}$ where each $\ket{\psi_i} \in \mathbb{C}^2$ is a local state on qubit $i$.  We obtain a density matrix $\rho = \ket{\psi}\bra{\psi} = \ket{\psi_1}\bra{\psi_1}\otimes \ldots \otimes \ket{\psi_n}\bra{\psi_n}$, which is a tensor product of single-qubit density matrices $\rho_i := \ket{\psi_i}\bra{\psi_i}$.  Any such $\rho_i$ may be expressed in the Pauli basis as $\rho_i = \frac{1}{2}(\mathbb{I} + \theta_{i1}\sigma^1 + \theta_{i2}\sigma^2 + \theta_{i3}\sigma^3)$, where $\theta_{ik} = \text{Tr}[\sigma^k \rho_i] \in \mathbb{R}$ and $\sum_{k \in [3]} \theta_{ik}^2 = 1$. In particular, product states with $\theta_{i1} = \theta_{i2} = 0$ and $\theta_{i3}^2 = 1$ correspond precisely to the classical setting (see \Cref{sec:classical-quantum} for an explicit connection between the two).  Product states exhibit no entanglement, and we may specifying $\theta_{ik}$ independently for each qubit $i$.  However, instead of producing a single $z_i^2 = 1$ as in the classical case, we must produce a unit vector $\mathbf{\theta_i} = [\theta_{i1}, \theta_{i2}, \theta_{i3}] \in \mathbb{R}^3$ for each $i \in [n]$.  Bri\"et, de Oliveira Filho, and Vallentin were the first to consider such generalizations of scalars to unit vectors, in the context of the Grothendieck problem~\cite{B10,BJ10}, and their analysis has fueled recent approximation algorithms for QLH~\cite{G19,H20}.

The classical Goemans-Williamson rounding scheme obtains the unit vectors $\mathbf{v}_i$ from a Cholesky decomposition of a PSD matrix $V^T V = R \semigeq 0$.  Taking the $\mathbf{v}_i$ to be columns of $V$ yields $R_{i,j} = \mathbf{v}_i \cdot \mathbf{v}_j$. Recent approximation algorithms~\cite{B19,G19,H20} for QLH have mimicked this approach, as do we.  Let $M^* \semigeq 0$ be an optimal solution to \Cref{prob:5} (the $\rho_{ij}^*$ are not necessary to describe the rounding algorithm).  We find a Cholesky decomposition $V^T V = M^*$, and let $\mathbf{v}_{ik} \in \mathbb{R}^d$ be the column of $V$ associated with $\sigma_i^k$ for $i \in [n], k \in [3]$; we may assume $d \leq 3n+1$.
In addition we let $\mathbf{v}_0$ be the column of $V$ corresponding to $\mathbb{I}$.  These are unit vectors as a consequence of Constraints \eqref{sdp-relax:diag-id} and \eqref{sdp-relax:diag}.

We will employ the same rounding algorithm for both the general and strictly quadratic cases.  While previous related works~\cite{B19,G19,H20} have in some cases had to rely on more sophisticated rounding schemes because they have been amenable to analysis, we are able to shed light on what is arguably the most natural generalization of the Goemans-Williamson approach.
%The main idea is to use the Cholesky vectors of $M^*$ to round to Bloch vectors as in \cite{G19}.  There the authors break up the Cholesky vectors up into sets of $3$, where each set corresponds to a particular qubit, i.e.
%\begin{equation}
%[\mathbf{v}_{11}, \mathbf{v}_{12}, \mathbf{v}_{13}], \,\,\, %[\mathbf{v}_{21}, \mathbf{v}_{22}, \mathbf{v}_{23}],...
%\end{equation}
We draw $\mathbf{r}\sim \mathcal{N}(0, \mathbb{I}_d)$, i.e. a multivariate distribution over $d$ independent and standard Gaussian variables. For each qubit, we obtain the desired vector $\mathbf{\theta}_i = [\theta_{i1},\theta_{i2},\theta_{i3}]$ as:
\begin{equation*}
    [\mathbf{v}_{i1}\cdot \mathbf{r},\ \mathbf{v}_{i2}\cdot \mathbf{r},\ \mathbf{v}_{i3}\cdot \mathbf{r}]/Q_i,
\end{equation*}
where $Q_i:=\sqrt{(\mathbf{v}_{i1}\cdot{} \mathbf{r})^2+(\mathbf{v}_{i2}\cdot{} \mathbf{r})^2+(\mathbf{v}_{i3}\cdot{} \mathbf{r})^2}$ is a normalization. 

The classical Max Cut problem corresponds to a strictly quadratic Hamiltonian (see \Cref{sec:classical-quantum} for justification); however, classical Max $2$-SAT and more general Max $2$-CSP have $1$-local terms (i.e., linear terms in $\{\pm 1\}$ variables).  In contrast, strictly quadratic instances of QLHP serve as a quantum generalization of Max $2$-SAT and Max $2$-CSP that have no $1$-local terms.  In order to obtain effective classical approximations in the presence of $1$-local terms, an additional vector $\mathbf{v}_0$ is necessary, representing (scalar) identity. Generally, the vector $\mathbf{v}_0$ is used in conjunction with more sophisticated rounding schemes (e.g. \cite{L02}) to obtain positive expectation from the $1$-local terms.  For the quantum case, relatively simple approaches suffice to get good approximations~\cite{H20}.  Using the vector $\mathbf{v}_0$ is not necessary for the strictly quadratic case, and including it does not affect its approximation.

\paragraph{Rounding algorithm.} The rounding approach described above produces single-qubit density matrices:
\begin{equation*}
\rho_i=\frac{1}{2}\left(\mathbb{I}+\frac{\mathbf{v}_{i1}\cdot \mathbf{r}}{Q_i}\sigma^1+\frac{\mathbf{v}_{i2}\cdot \mathbf{r}}{Q_i}\sigma^2+\frac{\mathbf{v}_{i3}\cdot \mathbf{r}}{Q_i}\sigma^3\right).
\end{equation*}
Hence, on any $1$-local term, $\mathbb{E}[\text{Tr}[\sigma^k \rho_i]]=\mathbb{E} [\mathbf{v}_{ik} \cdot \mathbf{r}/Q_i]=0$ since $Q_i$ is an even function and $\mathbf{v}_{ik} \cdot{} \mathbf{r}$ is an odd function in each entry of $\mathbf{r}$.  Thus, in order to get a nontrivial approximation on $1$-local terms, we will use the vector $\mathbf{v}_0$ to globally flip the sign of the $\mathbf{\theta}_i$ vectors of all qubits, i.e. $\mathbf{v}_{ik}\cdot \mathbf{r}/Q_i \rightarrow sign(\mathbf{v}_0 \cdot \mathbf{r}) (\mathbf{v}_{ik}\cdot \mathbf{r}/Q_i)$.  Since $sign(\mathbf{v}_0 \cdot \mathbf{r})\in \{\pm 1\}$, for quadratic objective terms this factor will cancel out, but for $1$-local terms we will gain additional objective from the correlation of $\mathbf{v}_0 \cdot \mathbf{r}$ and $\mathbf{v}_{ik} \cdot \mathbf{r}$.  Formally, we can state the rounding algorithm, which applies to any instance of QLH, as follows.

\begin{algorithm}Hyperplane rounding for $2$-Local Hamiltonian:\label{alg:2} 
\begin{enumerate}
\item Given some instance of \Cref{prob:QLH} formulate and solve the corresponding instance of \Cref{prob:5}.  Let $M^*$ be the optimal moment matrix obtained from \Cref{prob:5}.  
\item  Find the Cholesky decomposition of $M^*$, obtaining Cholesky vectors $\mathbf{v}_0 \in \mathbb{R}^d$ and $\{\mathbf{v}_{ik} \in \mathbb{R}^d\}$ such that $M^*(\sigma_i^k,\sigma_j^l)=\mathbf{v}_{ik} \cdot \mathbf{v}_{jl}$ and $M^*(\mathbb{I},\sigma_i^k)=\mathbf{v}_0 \cdot \mathbf{v}_{ik}$ for $i,j \in [n]$ and $k,l \in [3]$.
\item  Let $\mathbf{r}$ be a random vector with $\mathbf{r}\sim \mathcal{N}(0, \mathbb{I}_d)$.
\item  For each qubit $i$, set $Q_i=\sqrt{(\mathbf{v}_{i1}\cdot \mathbf{r})^2+(\mathbf{v}_{i2}\cdot \mathbf{r})^2+(\mathbf{v}_{i3} \cdot \mathbf{r})^2}$, and set $\theta_{ik}=sign(\mathbf{v}_0\cdot  \mathbf{r})(\mathbf{v}_{ik}\cdot \mathbf{r}/Q_i)$ for $k \in [3]$.
\item  Output the (pure) state:
\begin{equation*}
\rho =\bigotimes_{i=1}^n \frac{1}{2}(\mathbb{I} +\theta_{i1} \sigma^1 +\theta_{i2} \sigma^2+\theta_{i3}\sigma^3).
\end{equation*}
\end{enumerate}
\end{algorithm}

%In the strictly quadratic case we use essentially the same algorithm, with the main difference being that we do not have to cover the $1-$local terms, and hence do not have to use an additional index $0$:

%\begin{algorithm}\label{alg:1}
%\begin{enumerate}
%\item  Given some instance of \Cref{prob:1} formulate and solve the corresponding instance of \Cref{prob:2}.  Let $M^*$ be the optimal moment matrix obtained from \Cref{prob:2}.  
%\item  Find the Cholesky decomposition of $M^*$, obtaining Cholesky vectors $\{\mathbf{v}_{ij}\}$ such that $M_{ij, kl}^*=\mathbf{v}_{ij}^T \mathbf{v}_{kl}$.
%\item  Let $\mathbf{r}$ be a random vector with i.i.d entries distributed according to $\mathbf{r}_i\sim \mathcal{N}(0, 1)$, with size matching the vectors $\mathbf{v}_{ij}$.  
%\item  For each qubit, $i$, set $Q_i=\sqrt{(\mathbf{r}\cdot{} \mathbf{v}_{i1})^2+(\mathbf{r}\cdot{} \mathbf{v}_{i2})^2+(\mathbf{r}\cdot{} \mathbf{v}_{i3})^2}$ and set $\theta_{i1}=(\mathbf{r}\cdot{} \mathbf{v}_{i1})/Q_i$, $\theta_{i2}=(\mathbf{r}\cdot{} \mathbf{v}_{i2})/Q_i$, $\theta_{i3}=(\mathbf{r}\cdot{} \mathbf{v}_{i3})/Q_i$.
%\item  Output the (pure) state:
%\begin{equation}
%\rho =\bigotimes_{i=1}^n \left(\frac{\mathbb{I} +\theta_{i1} \sigma^1 +\theta_{i2} \sigma^2+\theta_{i3}\sigma^3}{2} \right)
%\end{equation}
%\end{enumerate}
%\end{algorithm}
We will give the following approximation guarantees for QLHP:

\begin{theorem}\label{thm:7}
Fix $k\in \{1, 2, 3\}$.  Suppose we are given an instance of QLHP (\Cref{prob:QLHP}), $\{H_e\}$ where $H_e= w_e P_e\otimes \mathbb{I}_{[n]\setminus\{e_1,e_2\}}$ for $w_e\geq 0$ and $P_e$ a projector of rank at least $k$, for all $e \in E$. Let $M^*$ be the optimal moment matrix for the corresponding SDP relaxation, \Cref{prob:5}, and let $\rho$ be the random output of \Cref{alg:2}.  Then,
\begin{equation*}
\mathbb{E} \left[ \text{Tr} \left[ \sum_{e \in E} H_e\ \rho\right]\right]\geq \alpha(k) \left( \sum_{e \in E} w_e\left(\frac{rank(P_e)}{4}+\text{Tr}[C_e M^*]\right)\right)
\geq \alpha(k)\, \lambda_{max} \left(\sum_{e \in E}  H_e \right),\\
\end{equation*}
where
\begin{alignat*}{1}
\alpha(k) =\begin{cases}
2/\pi-1/4 \approx 0.387 &\text{if $k=1$}\\
16/(9\pi) \approx 0.565 &\text{if $k=2$}\\
3/8+11/(9\pi) \approx 0.764 &\text{if $k=3$}.
\end{cases}
\end{alignat*}

\end{theorem}

\begin{theorem}\label{thm:1}
If, in addition to the assumptions of \Cref{thm:7}, the $P_e$ are strictly quadratic projectors, then the random output of \Cref{alg:2} satisfies:
\begin{equation*}
\mathbb{E} \left[ \text{Tr} \left[ \sum_{ij \in E} H_{ij} \rho\right]\right]
\geq \alpha(k) \lambda_{max} \left(\sum_{ij \in E}  H_{ij} \right),
\end{equation*}
where 
\begin{alignat*}{1}
\alpha(k) =\begin{cases}
22/(15\pi) \approx 0.467 &\text{if $k=1$}\\
1/3+24/(25 \pi) \approx 0.639 &\text{if $k=2$}\\
1/2+388/(405 \pi) \approx 0.804 &\text{if $k=3$}.
\end{cases}
\end{alignat*}
\end{theorem}

The above results are rigorous, but non-optimal.  The quadratic analysis depends crucially on an expansion of a particular expectation in Hermite polynomials.  One can consider a higher order Hermite series to more accurately capture the expectation and achieve a better approximation factor.  We have such results, but opt to not include them in the interest of the reader.  Higher orders bring increased tedium, and our technique should be clear enough at the end of the paper that an interested reader could do the higher order calculation.  

One can ask, why not include a high enough order that the result becomes essentially optimal?  The issue is that polynomial expansions often converge slowly in the presence of discontinuities \cite{P19}.  Indeed, computationally we have determined that to get essentially optimal results one would need to go to high enough order that the polynomial expansion would become intractable.  One can determine the optimal approximation factor by using a high order expansion and numerically optimizing or simply by randomly sampling over some ``net'' of the parameter space.  As stated in the introduction, our observed approximation factors under these approaches are:
\begin{conjecture}
\Cref{alg:2} constitutes a randomized $\beta(k)$ approximation algorithm for strictly quadratic projectors where
\begin{align}
\beta(k) =\begin{cases}
0.498 \text{ if $k=1$}\\
0.653 \text{ if $k=2$}\\
0.821 \text{ if $k=3$}.
\end{cases}
\end{align}
\end{conjecture}
\noindent Proving an approximation factor as large as the observed performance of our algorithm (say with techniques from algebraic geometry) is the subject of future work.  

\subsection{Analysis Overview} \label{sec:analysis-overview}
We present an overview of our analysis for the strictly quadratic case, which will also carry over to the general case with additional bookkeeping and bounding for the $1$-local terms.  Suppose we are given an instance of QLHP (\Cref{prob:QLHP}) on which we execute \Cref{alg:2} to produce a random solution $\rho$.  
For $i,j \in [n]$, the $2$-qubit marginals of $\rho$ are 
\begin{equation}
\rho_{ij} = \frac{1}{4}(\mathbb{I}+\theta_{i1} \sigma^1 +\theta_{i2}\sigma^2 +\theta_{i3}\sigma^3)\otimes(\mathbb{I}+\theta_{j1} \sigma^1 +\theta_{j2}\sigma^2 +\theta_{j3}\sigma^3),
\end{equation}
and the objective value of $\rho$ is $\sum_{e \in E} \text{Tr}[H_e \rho] = \sum_{e \in E} w_e APX_e$, where $APX_e := \text{Tr}[P_e \rho_{e_1e_2}]$ is the unweighted contribution to the objective value from edge $e$. Let $M^*$ and $\rho_{e_1e_2}^*$ for $e \in E$ be the SDP solution obtained by \Cref{alg:2}, and let $SDP_e = \text{Tr}[P_e \rho_{e_1e_2}^*]$ be the unweighted contribution to the SDP objective value from edge $e$. The approximation ratio, which we seek to bound from below, is consequently:
\begin{equation*}
\alpha = \mathbb{E} \left[ \frac{\sum_{e \in E} w_e APX_e}
  {\sum_{e \in E} w_e SDP_e} \right]
= \frac{\sum_{e \in E} w_e \mathbb{E}[APX_e]}
  {\sum_{e \in E} w_e SDP_e }.
\end{equation*}
Observe that $APX_e \geq 0$ and $SDP_e \geq 0$ since $P_e$, $\rho_{e_1e_2}$, and $\rho_{e_1e_2}^*$ are all PSD.  
Since all the terms in the denominator are nonnegative, it follows from an elementary argument that
\begin{equation*}
\alpha \geq \frac{\sum_{e \in E} w_e \mathbb{E}[APX_e]}{\sum_{e \in E} w_e SDP_e} \geq \min_{e \in E} \frac{\mathbb{E}[APX_e]}{SDP_e}.
\end{equation*}
Thus it suffices to bound the approximation ratio for the worst case occurring on a single edge.

\paragraph{Bounding a worst-case edge.} We now focus our attention on a single edge $e=12$ on qubits $1,2$.  
We collect the vectors $\mathbf{v}_{ik}$, obtained from a Cholesky decomposition of the SDP solution $M^*$, into matrices $V_i = [\mathbf{v}_{i1}, \mathbf{v}_{i2}, \mathbf{v}_{i3}] \in \mathbb{R}^{d \times 3}$, for $i = 1,2$.  We define an objective matrix $C \in \mathbb{R}^{3 \times 3}$, containing scaled $2$-local Pauli-basis coefficients of $P_{12}$, in the vein of \Cref{eq:52}: $C(\sigma_1^k,\sigma_2^l) := \text{Tr}[\sigma^k \otimes \sigma^l\ P_{12}]$, for $k,l \in [3]$ (note that $C$ is not symmetric).  With these definitions in hand, observe that $V_i^T V_i = \mathbb{I}_3$, by the SDP constraints \eqref{sdp-relax:diag} and \eqref{sdp-relax:diag-block-zero}, and that $4\text{Tr}[C_{12}M^*] = \text{Tr}[V_1CV_2^T]$.  The hyperplane rounding produces unit vectors $\mathbf{\theta}_i^T = [\theta_{i1},\theta_{i2},\theta_{i3}] = \mathbf{r}^T V_i/||V_i^T\mathbf{r}||$. In terms of these variables, we have:
\begin{gather*}
    SDP_e = \text{Tr}[P_{12} \rho_{12}^*] = \frac{1}{4}\left( rank(P_{12}) + \text{Tr}[V_1CV_2^T] \right), \text{ and }\\
    \mathbb{E}[APX_e] = \mathbb{E}[\text{Tr}[P_{12} \rho_{12}]] = \frac{1}{4}\left( rank(P_{12}) + \mathbb{E}  \left[ \mathbf{\theta}_1^T C \mathbf{\theta}_2 \right]\right) = \frac{1}{4}\left( rank(P_{12}) + \mathbb{E}_{\mathbf{r}}  \left[ \frac{\mathbf{r}^T V_1 C V_2^T\mathbf{r}}{||V_1^T \mathbf{r}|| \ ||V_2^T \mathbf{r}||}\right] \right),
\end{gather*}
by \Cref{eq:SDP-objective-value} and because $\text{Tr}[P_{12}] = rank(P_{12})$, since $P_{12}$ is a projector. Thus, setting $k = rank(P_{12})$, the quantity we seek to bound is
\begin{equation*}
    \alpha \geq \min_{V_1,V_2,C} \frac{k + \mathbb{E}_{\mathbf{r}}  \left[ \frac{\mathbf{r}^T V_1 C V_2^T\mathbf{r}}{||V_1^T \mathbf{r}|| \ ||V_2^T \mathbf{r}||}\right]} {k + \text{Tr}[V_1CV_2^T]}.
\end{equation*}
The bulk of our analysis lies in (i) simplifying the above to reduce the number of parameters in the minimization and expectation (\Cref{sec:22}), and (ii) deriving analytical bounds on the expectation (\Cref{sec:23}).

\paragraph{Simplifying the Gaussian expectation.}
The first simplification follows from observing that $V_i^T \mathbf{r} \in \mathbb{R}^3$ are multivariate Gaussians for $i=1,2$ since they are linear combinations of Gaussians, $\mathbf{r}\sim \mathcal{N}(0,\mathbb{I})$. If we let $\mathbf{z}^T = [z_1, z_2, z_3] = \mathbf{r}^T V_1$ and $(\mathbf{z'})^T = [z_1', z_2', z_3'] = \mathbf{r}^T V_2$, then $[\mathbf{z},\mathbf{z'}]\sim \mathcal{N}(0,\Sigma)$, where
\begin{equation*}
\Sigma = \begin{bmatrix}
\mathbb{I} & V_1^T V_2\\
V_2^T V_1 & \mathbb{I}
\end{bmatrix} \in \mathbb{R}^{6 \times 6}.
\end{equation*}
The Gaussians $z_i$ are mutually independent as well as the $z_i'$, and the covariance between $\mathbf{z}$ and $\mathbf{z'}$ is given by $M = V_1^T V_2 \in \mathbb{R}^{3 \times 3}$. Our bound now depends on a constant number of parameters, the 18 entries of $C$ and $M$:
\begin{equation}\label{eq:overview:gaussian-approx-bound}
    \alpha \geq \min_{M,C} \frac{k + \mathbb{E}_{\mathbf{z},\mathbf{z'}}  \left[ \frac{\mathbf{z}^T C \mathbf{z'}}{||\mathbf{z}|| \ ||\mathbf{z'}||}\right]} {k + \text{Tr}[C^T M]}.
\end{equation}
For classical hyperplane rounding algorithms, $C$ and $M$ simply reduce to scalars, and one may resort to a numerical argument to furnish the desired bound. However, in the case of QLH above, numerical bounds exhibit poor precision or convergence due to the number of parameters.  Thus we press on, and our next observation is that only the singular values of $M$ matter for the analysis. The above arguments are detailed in \Cref{lem:6} in \Cref{sec:16}, which also shows that we may assume:
\begin{equation}\label{eq:overview:diag-C-M}
C = \begin{bmatrix}
p & 0 & 0\\
0 & q & 0\\
0 & 0 & r\\
\end{bmatrix}\text{ and } 
M = \begin{bmatrix}
a & 0 & 0\\
0 & b & 0\\
0 & 0 & c\\
\end{bmatrix},
\end{equation}
where $a,b,c$ are the singular values of $V_1^T V_2$.  This reduction to 6 parameters puts analysis of $\alpha$ within reach.  The special case when $a=b=c$ turns out to be equivalent to the recently studied quantum analog of Max Cut related to the quantum Heisenberg model~\cite{G19,A20}.  For this case, a representation of the expectation, 
\begin{equation}\label{eq:overview:expectation}
\mathbb{E}_{\mathbf{z},\mathbf{z'}}  \left[ \frac{\mathbf{z}^T C \mathbf{z'}}{ ||\mathbf{z}|| \ ||\mathbf{z'}||} \right]
\end{equation}
as a hypergeometric function follows from work of Bri\"et, de Oliveira Filho, and Vallentin (the expectation ends up being equivalent to the one in Lemma 2.1 from~\cite{BJ10}, when $\mathbf{u} \cdot \mathbf{v}$ in the lemma equals $a=b=c$). To the best of our knowledge, no elementary representation is known when $a,b,c$ may be distinct.  We appeal to Hermite analysis to express the expectation~\eqref{eq:overview:expectation} as a polynomial that we are subsequently able to bound; this is carried out in Lemmas \ref{lem:26}-\ref{lem:21}/\Cref{sec:23}.

\paragraph{Introducing constraints from positivity.}  The matrices $C$ and $M$ from~\eqref{eq:overview:diag-C-M} are related to the quadratic Pauli-basis coefficients of $P_{12}$ and $\rho_{12}^*$, respectively. The other ingredient of our analysis of the bound~\eqref{eq:overview:gaussian-approx-bound} is restricting $C$ and $M$ based on the facts that $P_{12} \semigeq 0$ and $\rho_{12}^* \semigeq 0$, which is undertaken in \Cref{sec:15} and \Cref{lem:14}.  The bound we obtain is
\begin{align}\label{eq:84}
\alpha \geq \min_{\substack{[a, b, c]\in \mathcal{S}\\ [p, q, r]\in \mathcal{P}_k}}\frac{k+\mathbb{E}_{\mathbf{z},\mathbf{z'}} \left[ \frac{p z_1 z_1' +q z_2 z_2' +r z_3 z_3'}{\sqrt{(z_1^2+z_2^2+z_3^2)((z_1')^2+(z_2')^2+(z_3')^2)}}\right]}{k+a p+bq+cr},
\end{align}
where $\mathcal{S}$ and $\mathcal{P}_k$ are specific polytopes ($\mathcal{S}$ is a simplex as is $\mathcal{P}_k$ for $k\not=2$) derived from the positivity of $P_{12}$ and $\rho_{12}^*$.  Finally, Lemmas \ref{lem:26} and \ref{lem:21} in \Cref{sec:23} derive the bounds in the main theorems \ref{thm:7} and \ref{thm:1}, respectively.

\paragraph{Analysis roadmap.}  First, we describe the notations which will be used throughout the paper in \Cref{sec:notation}.  Following this, in \Cref{sec:22}, we will describe our reduction of the expectation of an edge to a standard form (\Cref{sec:16}), as well as the tools required for the proof of these lemmas (primarily \Cref{lem:14} and singular value decomposition).  In \Cref{sec:23} we will cover the main technical lemmas in this work, \Cref{lem:26} and \Cref{lem:21}.  These are simply lower bounds on a constrained minimization of a particular expectation, e.g.\ \Cref{eq:84}.  Proving these lemmas is the brunt of the technical work in this paper, and will require tools from orthogonal polynomials (\Cref{sec:9}), optimization (\Cref{sec:10}), and some elementary calculus (\Cref{sec:11}).

\section{Notations}\label{sec:notation}
We will adopt several mostly standard notations throughout the paper.  Note that $[n]:=\{1, 2, ..., n\}$ takes its standard definition, as well as the double factorial: $a!! := a(a-2)...4\cdot{}2$ if $a$ is even and $a!!:= a(a-2)(a-4)...3\cdot{} 1$ if $a$ is odd.  By convention, $(-1)!!=1=0!!$.  The notation $\mathbf{z} \sim \mathcal{N} (\mu, \Sigma)$ will be used to indicate that $\mathbf{z}$ is a multivariate normal distribution with mean $\mu$ and covariance matrix $\Sigma$.  Consistent with the notation from \cite{O14}, we define $L^2(\mathbb{R}^n, \gamma)$ as the set of multivariate functions which are square integrable under the Gaussian measure.  More formally $f:\mathbb{R}^n \rightarrow \mathbb{R}$ is an element of $L^2(\mathbb{R}^n, \gamma)$ if:
\begin{equation}
\int_{\mathbb{R}^n} |f(\mathbf{x})|^2 d\gamma(\mathbf{x}) := \int_{\mathbb{R}^n} |f(\mathbf{x})|^2  \frac{e^{-||\mathbf{x}||^2/2}}{(2\pi)^{n/2}} d\mathbf{x} < \infty.
\end{equation}
\noindent We define an inner product on this space as:
\begin{equation}\label{eq:8}
\langle f, g \rangle =\int_{\mathbb{R}^n} f(\mathbf{x}) g(\mathbf{x}) d\gamma(\mathbf{x}).
\end{equation}
\noindent (Not to be confused with ``Bra-ket'' notation, meaning will always be obvious from context)

We will make use of the Cholesky decomposition \cite{V13}.  Let $M$ be a real, symmetric, $n\times n$ PSD matrix.  The Cholesky decomposition guarantees the existence of efficiently-computable vectors $\{\mathbf{v}_i\}_{i=1}^n$, each of which has dimension at most $n$, such that $M_{ij}=\mathbf{v}_i^T \mathbf{v}_j$.  We will refer to the $\mathbf{v}_i$ simply as the Cholesky vectors of $M$.  

The Gamma and Hypergeometric functions will take their standard definitions:
\begin{equation}
\Gamma(z):= \int_{0}^\infty x^{z-1} e^{-x} dx, \,\,\,\,\,\,\,\text{and}\,\,\,\,\,\,  \,_2F_1 \left[\begin{matrix} a, \,\, b\\ c\end{matrix} ; \,z\right]=\sum_{n=0}^\infty \frac{(a)_n (b)_n}{(c)_n} \frac{z^n}{n!},
\end{equation}
where the rising factorial $(x)_n = 1$ if $n = 0$, or $(x)_n = x(x+1)\ldots(x+n-1)$ if $n > 0$.

We will need to define a convex hull.  Given some set $S=\{\mathbf{v}_1, ..., \mathbf{v}_m\}$ of points in $\mathbb{R}^p$, we denote the convex hull of $S$ as $conv(S)$ or $conv \,\,S$:
\begin{equation}
conv(S)=\{\lambda_1 \mathbf{v}_1+...+\lambda_m \mathbf{v}_m: \sum_i \lambda_i=1, \lambda_i \geq 0 \,\forall\, i\}.
\end{equation}
\noindent We will reserve $\mathcal{S}$, $-\mathcal{S}$ and $\mathcal{T}$ for specific convex hulls:
\begin{gather}\label{eq:71}
\mathcal{S}=conv\,\,\{[-1, -1, -1], [-1, 1, 1], [1, -1, 1], [1, 1, -1]\}\\
\nonumber -\mathcal{S}=conv\,\,\{[1, 1, 1], [1, -1, -1], [-1, 1, -1], [-1, -1, 1]\}\\
\nonumber \mathcal{T}=conv\{[2, 0, 0], [0, 2, 0], [0, 0, 2], [-2, 0, 0], [0, -2, 0], [0, 0, -2]\}
\end{gather}

We reserve $\Sigma': \mathbb{R}^n \rightarrow \mathbb{R}^{2n\times 2n}$ for the matrix-valued function:
\begin{equation}\label{eq:44}
\Sigma'(a_1, a_2, ..., a_n)= \begin{bmatrix}
\mathbb{I}_n & \begin{matrix}a_1 & 0 & \hdots & 0\\0 & a_2 &  & \vdots\\ \vdots &  & \ddots & 0 \\ 0 & \hdots & 0 & a_n\end{matrix} \\
\begin{matrix}a_1 & 0 & \hdots & 0\\0 & a_2 &  & \vdots\\ \vdots &  & \ddots & 0 \\ 0 & \hdots & 0 & a_n\end{matrix}& \mathbb{I}_n
\end{bmatrix}.
\end{equation}
The vector-valued function $diag: \mathbb{R}^{n\times n}\rightarrow \mathbb{R}^n$ has its standard definition: $diag(M)=[M_{11}, M_{22}, ..., M_{nn}]$.

\section{Standard Form for the Expectation of an Edge}\label{sec:22}

\paragraph{Overview.}  Recall we will use linearity of expectation to reduce the expected approximation factor for the algorithm as a whole to the expected approximation factor for the ``worst-case'' edge.  The next step is to reduce the number of parameters we need to consider by converting a generic edge to a ``standard'' form.  Roughly speaking, this is accomplished by applying a singular value decomposition to the moment matrix, followed by the application of some well-known facts concerning the structure of $2$-qubit density matrices and projectors.  In this section we will first give formal statements for the properties of $2$-qubit states we will need (\Cref{sec:15}), then integrate these facts with our SVD argument to provide the main lemmas for this reduction (\Cref{sec:16}).  

\subsection{Moment Matrix Description of Two Qubit Quantum States} \label{sec:15}
Moment matrices for $2$-qubit states are well studied objects \cite{H96, K01, S11, J14, G16} (see \cite{G16} for an extensive list of references), and the lemmas in this section are all presented in some form in these previous works.  It is a well-known fact that tensor product of Pauli matrices form a complete basis for Hermitian operators.  Hence, any Hermitian matrix is completely described by it's coefficients in this basis, and we can associate this Hermitian matrix with some real matrix containing the coefficients.  As an example let us consider a state, $\rho$, or a projector $P$, on two qubits.  Let them have Pauli decompositions: 
\begin{equation}\label{eq:34}
\rho=\frac{1}{4} \sum_{l, m=0}^3 \Gamma_{lm} \sigma^l \otimes \sigma^m\,\,\,\,\,\,\text{or}\,\,\,\,\,\, P=\frac{1}{4} \sum_{l, m=0}^3 \Gamma_{lm} \sigma^l \otimes \sigma^m.
\end{equation}
\noindent We will think of $\Gamma_{lm}$ as comprising a real matrix, which gives a complete description of the corresponding quantum state/projector.
\begin{equation}\label{eq:5}
\Gamma=\begin{bmatrix}
1 & \Gamma_{01} & \Gamma_{02} & \Gamma_{03}\\
\Gamma_{10} & \Gamma_{11} & \Gamma_{12} & \Gamma_{13} \\
\Gamma_{20} & \Gamma_{21} & \Gamma_{22} & \Gamma_{23} \\
\Gamma_{30} & \Gamma_{31} & \Gamma_{32} & \Gamma_{33} 
\end{bmatrix}
=\begin{bmatrix}
1 & \mathbf{v}^T\\
\mathbf{u} & R
\end{bmatrix}.
\end{equation}
By the above we mean that the matrix $R$ corresponds to the submatrix of $\Gamma$ with indices in $\{1, 2, 3\}$, $\mathbf{u}=[\Gamma_{10}, \Gamma_{20}, \Gamma_{30}]^T$ and $\mathbf{v}=[\Gamma_{01}, \Gamma_{02}, \Gamma_{03}]^T$.  We will refer to the vectors $\mathbf{u}$ and $\mathbf{v}$ as the $1$-local parts and $R$ as the quadratic part.  

On several occasions in this work, we will be given an arbitrary $4\times 4$ real matrix $\Gamma$ or a $3\times 3$ matrix $R$.  We will say $\Gamma$ (resp. $R$) is a {\it valid $2$-moment ( resp. valid quadratic $2$-moment)} if it corresponds to a physical density matrix or projector as in \Cref{eq:34}.  

\begin{definition}
\begin{enumerate}
\item  Given some density matrix $\rho$ or projector $P$ as in \Cref{eq:34} we will say that $\Gamma$ (resp. $R$) as in \Cref{eq:5} is it's $2$-moment (resp. quadratic $2$-moment).
\item  Similarly if $\Gamma$ (resp. $R$) is a given $4\times 4$ (resp. $3\times 3$) real matrix, we say $\Gamma$ (resp. $R$) is a {\it valid $2$-moment (resp. valid quadratic $2$-moment)} for a density matrix or projector if it is the $2$-moment (resp. quadratic $2$-moment) of a (physical) density matrix or projector.
\end{enumerate} 
\end{definition}

This notation will be important because moments are more restricted than an arbitrary $3\times 3$ or $4\times 4$ matrix, and these restrictions will imply simplifications which will be crucial in the analysis.

\begin{lemma}\label{lem:14}
Let $P=1/4 \sum_{l, m=0}^3 \Gamma_{lm} \sigma^l \otimes \sigma^m$ be some $2$-qubit projector of rank $k$, and let $\rho=1/4 \sum_{l, m=0}^3 \Delta_{lm} \sigma^l \otimes \sigma^m$ be a density matrix on $2$ qubits.  Let $G$ be the quadratic $2$-moment for $P$, and let $D$ be the quadratic $2$-moment for $\rho$ with $D=L\Sigma N^T$ for $\Sigma$ diagonal and $L$, $N\in SO(3)$.  
\begin{enumerate}
\item  (\cite{H96}, equation $8$)  Let $P$, $Q$ $\in SO(3)$.  Then, $G'=P G Q^T$ (resp. $D'=P D Q^T$) is a valid quadratic $2$-moment for state $U_1 \otimes U_2 \rho U_1^\dagger \otimes U_2^\dagger$ (resp. projector $U_1 \otimes U_2 P U_1^\dagger \otimes U_2^\dagger$)  for some local unitaries $U_1$, $U_2$.
\item  (\cite{H96}, Proposition $1$)  The diagonal elements of $\Sigma$ reside in $\mathcal{S}$ (recall $\mathcal{S}$ was defined in \Cref{eq:71}).

\item \begin{align}\label{eq:59}
\text{If  } \begin{cases}
k=1\\
k=2\\
k=3
\end{cases}
\text{ then }  [\Gamma_{11}, \Gamma_{22}, \Gamma_{33}] \in 
\begin{cases}
\mathcal{S}\\
\mathcal{T}\\
-\mathcal{S}
\end{cases},
\end{align}
\noindent where $\mathcal{S}$, $\mathcal{T}$ and $-\mathcal{S}$ are defined in \Cref{eq:71}.  
\item Further, it holds that 
\begin{equation}\label{eq:60}
    -k \leq 4 \text{Tr}[P \rho] -k \leq 4-k,
\end{equation}

\item \begin{equation}\label{eq:63}
    -k \leq \Sigma_{11}\Gamma_{11}+\Sigma_{22}\Gamma_{22}+\Sigma_{33}\Gamma_{33} \leq 4-k,
\end{equation}

\item  \begin{equation}\label{eq:65}
    ||[\Gamma_{10}, \Gamma_{20}, \Gamma_{30}]||, \,\, ||[\Gamma_{01}, \Gamma_{02}, \Gamma_{03}]|| \leq \begin{cases} 1 \text{ if $k=1$ or $3$}\\ 2 \text{ if $k=2$}\end{cases},
\end{equation}\\
\item  and 
\begin{equation}\label{eq:64}
    ||[\Delta_{10}, \Delta_{20}, \Delta_{30}]||, \,\, ||[\Delta_{01}, \Delta_{02}, \Delta_{03}]|| \leq 1.
\end{equation}
\end{enumerate}
\end{lemma}
\begin{proof}
See Appendix, page \pageref{page:1}.
\end{proof}

\subsection{SVD Argument}\label{sec:16}

%We will set up some definitions, and then proceed to state and prove the main lemma.  To fix notations say we are interested in a particular edge $e=(i, j, \psi_{ij})$ and let $C^e$ be the corresponding value matrix.  We can imagine breaking up $C^e$ into $3\times 3$ blocks by grouping rows/columns which correspond to the same qubit.  For instance, the $(i, j)th$ block would correspond to the $3\times 3 $ matrix:
%$$
%\begin{bmatrix}
%C_{iX, jX}^e & C_{iX, jY}^e & C_{iX, jZ}^e\\
%C_{iY, jX}^e & C_{iY, jY}^e & C_{iY, jZ}^e\\
%C_{iZ, jX}^e & C_{iZ, jY}^e & C_{iZ, jZ}^e
%\end{bmatrix}
%$$
%\noindent The main lemma follows:
\begin{mdframed}
\begin{lemma}[Quadratic Part of Expectation]\label{lem:6}
Let $V_1=[\mathbf{v}_{11}, \mathbf{v}_{12}, \mathbf{v}_{13}]\in \mathbb{R}^{d\times 3}$ and $V_2=[\mathbf{v}_{21}, \mathbf{v}_{22}, \mathbf{v}_{23}]\in \mathbb{R}^{d\times 3}$ be real matrices with normalized columns ($|| \mathbf{v}_{ij}||^2=1$) such that $V_1^T V_2$ is a valid quadratic $2$-moment for some density matrix.  Let $C$ be a quadratic $2$-moment for some projector of rank $k$,  and let $\mathbf{r}\sim \mathcal{N}(0, \mathbb{I_d})$ be the same size as the vectors $\mathbf{v}_{ij}$.  Then,
\begin{align}\label{eq:22}
\mathbb{E}_{\mathbf{r}}\left[\frac{ \mathbf{r}^T V_1 C V_2^T \mathbf{r}}{||V_1^T \mathbf{r}|| \,\, ||V_2^T \mathbf{r}||}\right]=\mathbb{E} \left[ \frac{p z_1 z_1' +q z_2 z_2' +r z_3 z_3'}{\sqrt{(z_1^2+z_2^2+z_3^2)((z_1')^2+(z_2')^2+(z_3')^2)}}\right]\\
{\rm and }\,\, Tr\left( V_1 C V_2^T \right)=a p+bq+cr\label{eq:23}.
\end{align}
\noindent with $[z_1, z_2, z_3, z_1', z_2', z_3'] \sim \mathcal{N}(0, \Sigma'(a, b, c))$ for some constants $(a, b, c, p, q, r)$.  Further, we can assume without loss of generality that $[a, b, c]\in \mathcal{S}$ and $[p, q, r] \in \mathcal{P}_k$ where 
\begin{equation}\label{eq:66}
\mathcal{P}_k=\begin{cases}
\mathcal{S} \text{ if } k=1\\
\mathcal{T} \text{ if } k=2\\
-\mathcal{S} \text{ if } k=3
\end{cases}.
\end{equation}
\end{lemma}
\end{mdframed}

\begin{proof}
\noindent Take a singular value decomposition of $V_1^T V_2=L \Sigma N^T$.  Let us assume that $\det(L)=1=\det(N)$ by possibly absorbing signs into 
\begin{equation}\label{eq:3}
\Sigma=\begin{bmatrix} a & 0 & 0\\
0 & b & 0\\
0 & 0 & c\end{bmatrix}.
\end{equation}
\noindent Since $L$ and $N$ are special orthogonal matrices we can re-write \Cref{eq:22} (resp. \Cref{eq:23}) as:
\begin{gather}\label{eq:24}
\mathbb{E}_{\mathbf{r}}\left[\frac{ \mathbf{r}^T V_1 C V_2^T \mathbf{r}}{||V_1^T \mathbf{r}|| \,\, ||V_2^T \mathbf{r}||}\right]=\mathbb{E}_\mathbf{r}\left[ \frac{ \mathbf{r}^T (V_1 L) (L^T C N) (N^T V_2^T) \mathbf{r}}{||L ^T V_1^T \mathbf{r}|| \,\, ||N^T V_2^T \mathbf{r}||}\right],\\
Tr\left[V_1 P V_2^T\right]=\text{Tr}\left[ (V_1 L) (L^T C N) (N^T V_2^T)\right].\label{eq:25}
\end{gather}
\noindent Setting 
\begin{equation}\label{eq:6}
W_1=V_1 L, \,\,\,\,\,\, W_2=V_2 N, \,\,\,\text{and}\,\,\, C'= L^T C N,
\end{equation}
\noindent we can re-write \Cref{eq:24} (resp. \Cref{eq:25}) as:
\begin{gather}\label{eq:26}
=\mathbb{E}_{\mathbf{r}}\left[\frac{ \mathbf{r}^T W_1 C' W_2^T \mathbf{r}}{||W_1^T \mathbf{r}|| \,\, ||W_2^T \mathbf{r}||}\right]\\
=\text{Tr}\left[ W_1 C' W_2^T\right]\label{eq:27}
\end{gather}
Defining the random variables $[z_1, z_2, z_3]^T :=W_1^T \mathbf{r}$ and $[z_1', z_2', z_3']^T :=W_2^T \mathbf{r}$, we can re-write \Cref{eq:26} (resp. \Cref{eq:27}) as:
\begin{gather}\label{eq:28}
=\mathbb{E} \left[ \frac{[z_1, z_2, z_3] C' [z_1', z_2', z_3']^T }{\sqrt{(z_1^2+z_2^2+z_3^2)((z_1')^2+(z_2')^2+(z_3')^2)}}\right],\\
=\text{Tr}\left[ W_1 C' W_2^T\right].\label{eq:29}
\end{gather}
Since the variable $\mathbf{r}$ is a multivariate Gaussian distribution, linear combinations of it's components are also distributed according to the multivariate Gaussian.  Since each of $(z_1, z_2, z_3, z_1', z_2', z_3')$ is the sum of normal random variables with zero mean, they must each have zero mean.  Observe that, by construction, $W_1^T W_2=\Sigma$.  So, if we let $\mathbf{w}_1^1$ be the first column of $W_1$ and $\mathbf{w}_1^2$ be the first column of $W_2$ then 
\begin{align*}
\mathbb{E}[z_1 z_1']=\mathbb{E}[ (\mathbf{r}\cdot{} \mathbf{w}_1^1) (\mathbf{r} \cdot{} \mathbf{w}_1^2)]=\mathbb{E}[\mathbf{r}^T \mathbf{w}_1^1 (\mathbf{w}_1^2)^T \mathbf{r}]=\mathbb{E}[\text{Tr}(\mathbf{r}^T \mathbf{w}_1^1 (\mathbf{w}_1^2)^T \mathbf{r})]=\mathbb{E}[\text{Tr}(\mathbf{r} \mathbf{r}^T \mathbf{w}_1^1 (\mathbf{w}_1^2)^T)]\\
\nonumber =\text{Tr}[\mathbb{E}(\mathbf{r} \mathbf{r}^T) \mathbf{w}_1^1 (\mathbf{w}_1^2)^T]=a.
\end{align*}
\noindent Similar calculations show that the covariance of these random variables are equal to the overlap of the corresponding vectors.  Hence, $[z_1, z_2, z_3, z_1', z_2', z_3']\sim \mathcal{N}(\mathbf{0}, \Sigma'(a, b, c))$ with $\Sigma'$ defined as in \Cref{eq:44}.

Observe that 
\begin{equation*}
\mathbb{E}\left[ \frac{z_1 C'_{1,2} z_2'}{\sqrt{(z_1^2+z_2^2+z_3^2)((z_1')^2+(z_2')^2+(z_3')^2)}} \right]=0,
\end{equation*}
\noindent since it is the expectation of an odd function with respect to $z_1$, and similarly for the other off-diagonal elements.  Let $(p, q, r)$ be the diagonal elements of $C'$.  Then we have shown that  \Cref{eq:28} is equal to 
\begin{equation*}
\mathbb{E} \left[ \frac{p z_1 z_1' +q z_2 z_2' +r z_3 z_3'}{\sqrt{(z_1^2+z_2^2+z_3^2)((z_1')^2+(z_2')^2+(z_3')^2)}}\right].
\end{equation*}
\noindent Additionally, it is easy to see $\text{Tr}\left[W_1 C' W_2^T\right]=a p+bq+cr$ by taking the trace starting with the columns of $W_1$.

To complete the lemma, we need to establish that we can restrict $(a, b, c, p, q, r)$ as described in the statement.  $[p, q, r]$ are the diagonal elements of a valid $2$-moment for a projector, by \Cref{lem:14}, since we obtained $C'$ by multiplying it on the left and right by special orthogonal matrices.  Hence, \Cref{eq:59} establishes $[p, q, r] \in \mathcal{P}_k$.  Since $\Sigma$ is obtained from $V_1^T V_2$ by multiplying on the right and left by special orthogonal matrices, \Cref{lem:14} item $2$ implies $diag(\Sigma) \in \mathcal{S}$.

\end{proof}
\begin{mdframed}
\begin{lemma}[Linear part of Expectation]\label{lem:18}
Let $V=[\mathbf{v}_1, \mathbf{v}_2, \mathbf{v}_3]\in \mathbb{R}^{d\times 3}$, $\mathbf{v}\in \mathbb{R}^d$  with $V V^T=\mathbb{I}$ and $\mathbf{v}^T \mathbf{v}=1$.  Further, let $\mathbf{p}\in \mathbb{R}^3$, $a=||V^T\mathbf{v}||$, $\mathbf{r}\sim \mathcal{N}(\mathbf{0}, \mathbb{I}_d)$.  Let $U\in SO(3)$ satisfy $U V^T\mathbf{v}=[a, 0, 0]^T$, and define $\tilde{\mathbf{p}}_1=(U \mathbf{p})_1$.   Then,
\begin{align}\label{eq:70}
    \mathbb{E}_\mathbf{r} \left[\frac{(\mathbf{r}^T V \mathbf{p}) (\mathbf{v}^T \mathbf{r})}{||\mathbf{r}^T V|| \cdot{} |\mathbf{r}^T \mathbf{v}|}\right]=\mathbb{E}_{X, Y, Z, S} \left[\frac{\tilde{\mathbf{p}}_1 X S}{\sqrt{X^2+Y^2+Z^2}|S|} \right],
\end{align}
where $[S, X, Y, Z]\sim \mathcal{N}(\mathbf{0}, \Sigma)$ with 
\begin{equation*}
\Sigma= \begin{bmatrix}
1 & a & 0 & 0\\
a & 1 & 0 & 0\\
0 & 0 & 1 & 0\\
0 & 0 & 0 & 1
\end{bmatrix}.
\end{equation*}
Additionally, $Tr[V \mathbf{p} \mathbf{v}^T]=a\tilde{\mathbf{p}}_1$.
\end{lemma}
\end{mdframed}

\begin{proof}
Define $\tilde{V}=V U^T$.  Then,
\begin{align}\label{eq:69}
\mathbb{E}_\mathbf{r} \left[\frac{(\mathbf{r}^T V \mathbf{p}) (\mathbf{v}^T \mathbf{r})}{||\mathbf{r}^T V|| \cdot{} |\mathbf{r}^T \mathbf{v}|}\right]=\mathbb{E}_\mathbf{r} \left[\frac{(\mathbf{r}^T \tilde{V} \tilde{\mathbf{p}}) (\mathbf{v}^T \mathbf{r})}{||\mathbf{r}^T \tilde{V}|| \cdot{} |\mathbf{r}^T \mathbf{v}|}\right].
\end{align}
Let $\tilde{V}=[\mathbf{v}_X, \mathbf{v}_Y, \mathbf{v}_Z]$, and define $S=\mathbf{v}^T \mathbf{r}$, $X=\mathbf{v}_X^T \mathbf{r}$, $Y=\mathbf{v}_Y^T \mathbf{r}$ and $Z=\mathbf{v}_Z^T \mathbf{r}$.  \Cref{eq:69} can then be written as:
\begin{equation*}
    =\mathbb{E} \left[ \frac{[X, Y, Z] \tilde{\mathbf{p}} S}{\sqrt{X^2+Y^2+Z^2}|S|}\right],
\end{equation*}
with $\tilde{\mathbf{p}}=U \mathbf{p}$.  \Cref{eq:70} follows with the observation that
\begin{equation*}
    \mathbb{E} \left[\frac{Y S}{\sqrt{X^2+Y^2+Z^2}|S|}\right]=0=\mathbb{E} \left[\frac{Z S}{\sqrt{X^2+Y^2+Z^2}|S|}\right].
\end{equation*}

\noindent To finish the lemma, observe $Tr[V\mathbf{p} \mathbf{v}^T]=\mathbf{v}^T V \mathbf{p}=\mathbf{v}^T \tilde{V} \tilde{\mathbf{p}}=a \tilde{\mathbf{p}}_1$.
\end{proof}

\section{Proof of the Main Technical Lemmas}\label{sec:23}
In this section we give a proof and the requisite tools for the proofs of the main technical lemmas.  Recall we have some expectations from the rounding algorithm, and we are interested in finding lower bounds for these expectations over some set of the parameters.  We seek proofs of the following:
\begin{mdframed}
\begin{lemma}[Strictly Quadratic Case]\label{lem:26}
Let $\mathcal{S}$, $\mathcal{T}$, and $-\mathcal{S}$ be defined according to \Cref{eq:71}.  Further let $k\in \{1, 2, 3\}$ be fixed, and let 
\begin{equation}
    \mathcal{P}_k=\begin{cases}
    \mathcal{S} {\rm \text{ if $k=1$ }}\\
    \mathcal{T} {\rm \text{ if $k=2$ }}\\
    -\mathcal{S} {\rm \text{ if $k=3$ }}
    \end{cases}
\end{equation}
Let $[z_1, z_2, z_3, z_1', z_2', z_3']\sim \mathcal{N}(0, \Sigma'(a, b, c))$.  Then, we have the following lower bound:
\begin{align}\label{eq:72}
    \min_{\substack{[a, b, c]\in \mathcal{S}\\ [p, q, r]\in \mathcal{P}_k}}&\frac{k+\mathbb{E} \left[ \frac{p z_1 z_1' +q z_2 z_2' +r z_3 z_3'}{\sqrt{(z_1^2+z_2^2+z_3^2)((z_1')^2+(z_2')^2+(z_3')^2)}}\right]}{k+a p+bq+cr}  
\geq \begin{cases}
22/(15\pi) \approx 0.467 \text{ if $k=1$}\\
1/3+24/(25 \pi) \approx 0.639 \text{ if $k=2$}\\
1/2+388/(405 \pi) \approx 0.804 \text{ if $k=3$}
\end{cases}.
\end{align}

\end{lemma}
\end{mdframed}

\begin{mdframed}
\begin{lemma}[General Case]\label{lem:21}
Let $k\in \{1, 2, 3\}$ be fixed.  Let $a$, $b$, $c$, $p$, $q$, $r$, $t_i$, $d_i$, $t_j$, $d_j$ be constants, such that $|a|, |b|, |c|, |d_i|, |d_j| \leq 1$, and which satisfy:
\begin{gather}\label{eq:73}
    -k \leq ap+bq+rc \leq 4-k,\\
    \nonumber -k \leq ap+bq+rc+u_id_i+u_j d_j \leq 4-k,\\
    \nonumber |t_i d_i| \leq l,\\
    \nonumber \text{ and }|t_j d_j| \leq l,
\end{gather}
where $l$ is $1$ if $k=1$ or $3$ and $l=2$ if $k=2$.  Further, let $[z_1, z_2, z_3, z_1', z_2', z_3']\sim \mathcal{N}(0, \Sigma'(a, b, c))$, let $[s_1, x_1, x_2, x_3] \sim \mathcal{N}(0, \Sigma_1)$ for 
\begin{equation*}
     \Sigma_1=\begin{bmatrix} 1 & d_i & 0 & 0 \\d_i & 1 & 0 & 0\\ 0 & 0 & 1 & 0\\0 & 0 & 0 & 1\end{bmatrix},
\end{equation*}
and let $[s_2, y_1, y_2, y_3]\sim \mathcal{N}(0, \Sigma_2)$ for 
\begin{equation*}
    \Sigma_2=\begin{bmatrix} 1 & d_j & 0 & 0 \\d_j & 1 & 0 & 0\\ 0 & 0 & 1 & 0\\0 & 0 & 0 & 1\end{bmatrix}.
\end{equation*}
\textbf{Then,}
\begin{align}\label{eq:74}
    \bigg(k + \mathbb{E} \left[\frac{p z_1 z_1'+q z_2 z_2'+r z_3 z_3'}{\sqrt{(z_1^2+z_2^2+z_3^2)((z_1')^2+(z_2')^2+(z_3')^2)}} \right] +\mathbb{E} \left[ \frac{t_i x_1 s_1}{\sqrt{x_1^2+x_2^2+x_3^2}|s_1|}\right]\\
    \nonumber +\mathbb{E} \left[\frac{t_j y_1 s_2}{\sqrt{y_1^2+y_2^2+y_3^2}|s_2|} \right] \bigg) /\bigg(k+ap +bq+cr+t_i d_i +t_j d_j \bigg)\\
    \nonumber \geq \begin{cases} 2/\pi-1/4 \approx 0.387 \text{ if $k=1$}\\   16/(9\pi) \approx 0.565 \text{ if $k=2$}\\  3/8+11/(9\pi) \approx 0.764 \text{ if $k=3$}\end{cases}
\end{align}
\end{lemma}
\end{mdframed}
Demonstrating these lemmas will require several tools.  Roughly we will need some results in Hermite polynomial analysis (\Cref{sec:9}), some elementary considerations for optimizing rational functions over polytopes (\Cref{sec:10}), and some analytical bounds for specific functions of interest (\Cref{sec:11}).  Finally we conclude with the proofs of these two technical statements, (\Cref{sec:12} and \Cref{sec:13}).  We give rough outlines of the two proofs in this section, so that the reader may have context for the following results:

\paragraph{Quadratic Case}  Given an expectation of the form \Cref{eq:72}, the first task is to reduce the optimization of $[p, q, r]\in \mathcal{P}_k$ to the optimization of $[p, q, r]$ over only the extreme points of $\mathcal{P}_k$.  This is accomplished by fixing $[a, b, c]$ and minimizing with respect to $[p, q, r]$.  From this perspective, we are minimizing a rational function on $(p, q, r)$ and it will easily follow (using \Cref{sec:10}) that we may lower bound the expectation by fixing $[p, q, r]$ to be an extreme point of $\mathcal{P}_k$.  The next step is to lower bound the resulting expression over all choices $[a, b, c] \in \mathcal{S}$.  Using Hermite polynomials (\Cref{sec:9}), we may evaluate the expectation in terms of a convergent series.  Bounding the error obtained from a truncation of this series (\Cref{sec:14}), and using calculus to bound the truncation (\Cref{sec:11} and the Appendix) yield the final result.  

\paragraph{General Case}  Given an expectation of the form of \Cref{eq:74} with parameters restricted as in \Cref{eq:73}, we can once again evaluate the expectation using an expansion in the Hermite polynomials.  The resulting expression yields a convergent series for the quadratic part, as in the previous case, as well as a hypergeometric function for the linear part.  Cutting off the Hermite expansion to linear order, and bounding the special function to linear order we obtain a lower bound which is once again a rational function in the parameters (subject to the assumed constraints).  Then, we can ``coarse grain'' the optimization to get an optimization in fewer parameters as a lower bound, and apply  the same bounds for optimizing functions over polytopes (\Cref{sec:10} and \Cref{sec:11}) to complete the proof.

\subsection{Hermite Polynomials}\label{sec:9}
\subsubsection{Definitions and Expansions} 
Hermite polynomials \cite{O14, S72} are orthogonal polynomials which provide an orthonormal basis for $L^2(\mathbb{R}, \gamma)$ under the inner product introduced in \Cref{eq:8}.  One can obtain these polynomials by applying the Gram-Schmidt procedure to the set of standard monomials $\{1, z, z^2, z^3, ...\}$ using the inner product in \Cref{eq:8} or alternatively they can be defined using a generating function:
\begin{equation}\label{eq:7}
e^{tz -t^2/2}=\sum_{j=0}^\infty \frac{h_j(z) t^j}{\sqrt{j !}}.
\end{equation}
Hermite polynomials have a well-known explicit expression \cite{S72}, and since we will make use of the formula in many proofs, we define the Hermite polynomials according to this expression:
\begin{definition}\label{def:3}
The $nth$ Hermite polynomial is defined as:
\begin{gather*}
    h_n(z)=\frac{\sqrt{n!}}{2^{n/2}} \sum_{l=0}^{n/2} \frac{(-1)^{n/2-l} (\sqrt{2}z)^{2 l}}{(2l)! (n/2-l)!}
    {\rm \text{ if $n$ is even, }} \\ 
    \nonumber {\rm and} \,\,\, h_n(z)=\frac{\sqrt{n!}}{2^{n/2}} \sum_{l=0}^{(n-1)/2}\frac{(-1)^{(n-1)/2-l} (\sqrt{2} z)^{2 l+1}}{(2l+1)!((n-1)/2-l)!} \text{ if $n$ is odd.}
\end{gather*}
\end{definition}
The generalization of these polynomials to the multivariate case is simply found by taking products of single variable Hermite polynomials:
\begin{definition}
Given some multi-index $\mu \in[\mathbb{Z}_{\geq 0}]^n$, we define the multivariate Hermite polynomials as $h_{\mu} (\mathbf{z}) :=\prod_{i=1}^n h_{\mu_i} (\mathbf{z}_i)$.
\end{definition}

As stated previously, our interest in these polynomials stems from the fact that they are a basis for the space $L^2(\mathbb{R}^n, \gamma)$.  Hence, any function is this space can be expanded in the Hermite polynomials:
\begin{proposition}\label{prop:1}
\begin{enumerate}
\item  \cite{O14} Let $f\in L^2(\mathbb{R}^n, \gamma)$.  We can define $f(\mathbf{z})=\sum_{\mu \in [\mathbb{Z}_{\geq 0}]^n} \hat{f}_\mu h_\mu(\mathbf{z})$ where $\hat{f}_\mu =\langle f, h_\mu \rangle=\int f(\mathbf{z}) h_{\mu}(\mathbf{z}) d\gamma(\mathbf{z})$ and convergence is in $L^2(\mathbb{R}^n, \gamma)$ (converges in norm defined by \Cref{eq:8}).
%\item  This expansion satisfies the ``Plancharel formula'':
%\begin{equation}
%\forall f, g \in L^2(\mathbb{R}^n, \gamma), \langle f, g\rangle=\sum_{\mu} \hat{f}_\mu \hat{g}_\mu 
%\end{equation}
\item  \cite{O14} If $z$, $z'$ are bivariate normal with $\mathbb{E}[z z']=\rho$ and $|\rho| \leq 1$, then $\langle h_i(z), h_j(z') \rangle =\delta_{ij} \rho^i$ where $\delta_{ij}$ is the standard discrete delta function.  
\item  Let $f\in L^2(\mathbb{R}^3, \gamma)$ with Hermite expansion $f(\mathbf{z})=\sum_{i, j, k} \hat{f}_{i, j, k} h_i(z_1)h_j(z_2) h_k(z_3)$.  Let $(z_1, z_2, z_3, z_1', z_2', z_3') \sim \mathcal{N}(0, \Sigma'(a, b, c))$ with $|a|$, $|b|$, $|c|$ $\leq 1$.  Then, $\mathbb{E}[f(\mathbf{z}) f(\mathbf{z'})]=\sum_{i, j, k} \hat{f}_{i, j, k}^2 a^i b^j c^k $.
\end{enumerate}
\end{proposition}
For completeness, we supply a short proof of item $3$ in the appendix on Page~\pageref{page:4} (although it is implicit in \cite{O14}).   

\subsubsection{Hermite Expansions of Interest}
We wish to apply the general theory of Hermite polynomials to two specific functions of interest:
\begin{equation*}
   g(z)=sign(z)=\begin{cases}  1 \text{ if $z\geq 0$}\\-1 \text{ if $z<0$}\end{cases}\,\,\,\,\,\text{and }\,\,\,\,\,  f(z_1, z_2, z_3)=\frac{z_1}{\sqrt{z_1^2+z_2^2+z_3^2}}.
\end{equation*}
First observe that the Hermite expansions of both of these functions are well defined, since both $\int g(z)^2 d\gamma(z)$ and $\int f(\mathbf{z})^2 d\gamma(\mathbf{z})$ are $ < \infty$\footnote{The latter case holds because $\int f(z_1, z_2, z_3)^2+f(z_2, z_1, z_3)^2+f(z_3, z_1, z_2)^2 d\gamma(\mathbf{z})=1$, so by symmetry $\int f(\mathbf{z})^2 d\gamma(\mathbf{z})=1/3$.}.  The first function has a well-known Hermite expansion:
\begin{lemma}[\cite{S72}]\label{lem:23}
Let $g(z)=sign(z)=\begin{cases}  1 \text{ if $z\geq 0$}\\-1 \text{ if $z<0$}\end{cases}$.  Then, $g(z)=\sum_i \hat{g}_i h_i(z)$ where
\begin{equation*}
    \hat{g}_i=\begin{cases} 0 \text{ if $i$ is even}\\ \frac{\sqrt{i!}(-1)^q }{2^{q-1/2} \sqrt{\pi} q! (1+2 q)} \text{ if $i$ is odd } \end{cases},
\end{equation*}
\noindent with $q=(i-1)/2$.
\end{lemma}

For the second function, let us define:
\begin{equation}\label{eq:75}
    f(z_1, z_2, z_3):=\sum_{(i, j, k)} \hat{f}_{i, jk}\,\, h_i(z_1)h_j(z_2)h_k(z_3)
\end{equation}
We have chosen the notation $\{\hat{f}_{i, jk}\}$ rather than $\{\hat{f}_{i, j, k}\}$ because $\hat{f}_{i, j, k}=\hat{f}_{i, k, j}$, which follows from the fact that $f(z_1, z_2, z_3)=f(z_1, z_3, z_2)$ and orthonormality of the Hermite polynomials.  Indeed,
\begin{align*}
    \hat{f}_{i, jk}=\int f(z_1, z_2, z_3) h_i(z_1) h_j(z_2) h_k(z_3) d\gamma(\mathbf{z})=\int f(z_1, z_3, z_2) h_i(z_1) h_j(z_2) h_k(z_3) d\gamma(\mathbf{z})=\hat{f}_{i, kj}.
\end{align*}
We can also see that $\hat{f}_{i, jk}=0$ if $i$ is even, $j$ is odd or $k$ is odd.  This follows from the explicit form of the Hermite polynomials (\Cref{def:3}) as well as the fact that the integral of an anti-symmetric function along the entire real line is zero.  Having stated some symmetries of the coefficients we now state their analytic values:
\begin{lemma}\label{lem:24}
Let $f$ have Hermite expansion $\{\hat{f}_{i, jk}\}$ as defined in \Cref{eq:75}.  If $i$ is odd, $j$ is even, and $k$ is even, the Hermite coefficients take values:
\begin{equation}\label{eq:50}
    \hat{f}_{i, jk}=2\sqrt{\frac{2}{\pi}} \frac{\sqrt{i! j! k!}(-1)^{p}}{(i-1)!! j!! k!! (1+2p)(3+2p)}
\end{equation}
for $p=(i+j+k-1)/2$.  Otherwise, the coefficients are zero.  
\end{lemma}
\begin{proof}
This is a special case of \Cref{lem:25} found in the Appendix.  
\end{proof}

With these facts in hand, we can evaluate one of the expectations of interest to us:
\begin{lemma}\label{lem:19}
Let $(z, z_1, z_2, z_3) \sim \mathcal{N}(0, \Sigma)$ with 
\begin{equation*}
\Sigma=\begin{bmatrix}
1 & a & 0 & 0\\
a & 1 & 0 & 0\\
0 & 0 & 1 & 0\\
0 & 0 & 0 & 1
\end{bmatrix}.
\end{equation*}
Then,
\begin{equation}\label{eq:56}
    \mathbb{E} \left[ \frac{sign(z) z_1}{\sqrt{z_1^2+z_2^2+z_3^2}}\right]=\frac{4a}{3\pi} \, _2 F_1\left[\begin{matrix}1/2, \,\, 1/2\\5/2 \end{matrix};a^2 \right].
\end{equation}
\end{lemma}
\begin{proof}
See page \pageref{page:2} in the Appendix.  
\end{proof}

This falls out of the Hermite expansions we have so far, plus some elementary combinatorial identities.  For some intuition on how this fact is used in the analysis, note that the expectation in \Cref{eq:56} looks a lot like the component of one of the Bloch vectors in the rounding algorithm.  Indeed, analyzing that expectation will be exactly the purpose of this lemma.

\subsubsection{Bounds on Hermite Expansions}\label{sec:14}

While exact computations like \Cref{lem:19} would be desireable for all expansions of interest, in order to establish rigorous results we will need bounds in some cases.  The first lemma we describe concerns the symmetries of an inner product of random vectors with covariance as defined in \Cref{eq:44}:
\begin{lemma}\label{lem:11}
Let $[z_1, z_2, z_3, z_1', z_2', z_3'] \sim \mathcal{N}( 0, \Sigma'(a, b, c)) $ and let $\{\hat{f}_{i, jk}\}$ be the Hermite expansion of $f(\mathbf{z})=z_1/\sqrt{z_1^2+z_2^2+z_3^2}$.  Let $\mathbf{v}_1$, $\mathbf{v}_2$ be vectors with components defined as $(\mathbf{v}_1)_i:=z_i/\sqrt{z_1^2+z_2^2+z_3^2}$ and $(\mathbf{v}_2)_i=z_i'/\sqrt{(z_i')^2+(z_2')^2+(z_3')^2}$ for $i\in \{1, 2, 3\}$.  Further, define the following families of symmetric polynomials:
\begin{gather}\label{eq:12}
p_{i, jk}(a, b, c)=\begin{cases}
a^i b^j c^k+a^i b^k c^j +b^i a^j c^k +b^i a^k c^j+c^i a^j b^k+c^i a^k b^j \,\,\,\, \text{  if $k\neq j$ }\\
a^i b^j c^j+b^i a^j c^j +c^i a^j b^j \,\,\,\, \text{  if $k=j$}
\end{cases},\\
{\rm and} \,\, u_{i, jk}(a, b, c) =\begin{cases}   
a^i b^j c^k+a^i b^k c^j \,\,\,\, \text{  if $k\neq j$ }\\
a^i b^j c^j \,\,\,\, \text{  if $k=j$}
\end{cases}.
\end{gather}
Then, 
\begin{gather}\label{eq:30}
\mathbb{E}[\mathbf{v}_1\cdot{} \mathbf{v}_2]=\sum_{\substack{i, j, k\\ j \leq k}} \hat{f}_{i, jk}^2 p_{i, jk}(a, b, c),\\
{\rm and }\,\,\mathbb{E} \left[(\mathbf{v}_1)_1 (\mathbf{v}_2)_1\right]=\sum_{\substack{i, j, k\\ j \leq k}} \hat{f}_{i, jk}^2 u_{i, jk}(a, b, c)\label{eq:31}
\end{gather}
if $|a|$, $|b|$, $|c|$ $\leq 1$.
\end{lemma}
\begin{proof}
By \Cref{prop:1},
\begin{align}
\mathbb{E}[f(z_1, z_2, z_3)  f(z_1', z_2', z_3')]=\sum_{i, j, k} \hat{f}_{i, jk}^2 a^i b^j c^k.
\end{align}
\noindent We can use the invariance $\hat{f}_{i, jk}=\hat{f}_{i, kj}$ to write this as:
\begin{equation}\label{eq:13}
\mathbb{E}[f(z_1, z_2, z_3)  f(z_1', z_2', z_3')]=\sum_i \left[\sum_{j < k} \hat{f}_{i, jk}^2 (a^i b^j c^k + a^i b^k c^j ) +\sum_j \hat{f}_{i, jj}^2 a^i b^j c^j \right].
\end{equation}
\noindent This establishes \Cref{eq:31}.  For \Cref{eq:30}, we wish to calculate:
\begin{align}\label{eq:21}
\mathbb{E} \left[ \frac{ z_1 z_1' + z_2 z_2' + z_3 z_3'}{\sqrt{(z_1^2+z_2^2+z_3^2)((z_1')^2+(z_2')^2+(z_3')^2)}}\right]=\mathbb{E}[f(z_1, z_2, z_3)f(z_1', z_2', z_3')+f(z_2, z_1, z_3)f(z_2', z_1', z_3')\\
\nonumber +f(z_3, z_1, z_2)f(z_3', z_1', z_2')].
\end{align}
We can permute the variables in \Cref{eq:13} to obtain analogous expressions for 

\noindent $\mathbb{E} [ f(z_2, z_1, z_3) f(z_2', z_1', z_3')]$ and $\mathbb{E} [ f(z_3, z_1, z_2) f(z_3', z_1', z_2')]$.  Substituting these expressions into \Cref{eq:21} gives \Cref{eq:30}.
\end{proof}

We will also need a remainder bound for finite Hermite expansions.  We will state it for our Hermite expansion of interest, $\{\hat{f}_{i, jk}\}$, although it clearly generalizes to others. 
\begin{lemma}\label{lem:12}
Let $Q$ be a finite subset of $\mathbb{Z}_{\geq 0} \times \mathbb{Z}_{\geq 0} \times \mathbb{Z}_{\geq 0}$ such that $j\leq k$ for all $(i, j, k)\in Q$.  Then, 
\begin{align*}
\left|\sum_{\substack{i, j, k\\j\leq k}}\hat{f}_{i, jk}^2 p_{i, jk}(a, b, c)-\sum_{(i, j, k)\in Q}\hat{f}_{i, jk}^2 p_{i, jk}(a, b, c)\right| \leq 3 \left( 1/3-\sum_{(i, j, k) \in Q}2^{1-\delta_{jk}}\hat{f}_{i, jk}^2\right),
\end{align*}
\noindent where $\delta_{jk}$ is the standard discrete delta function and $|a|$, $|b|$, $|c|$ $\leq 1$.
\end{lemma}
\begin{proof}\footnote{There are some details here concerning convergence, etc. that we will not mention.  For reference, observe that the series $\sum_{i, j\leq k} \hat{f}_{i, jk}^2 p_{i, jk}(a, b, c)$ is absolutely convergent since $|p_{i, jk}(a, b, c)|$ is bounded and $\sum_{i, j, k}\hat{f}_{i, jk}^2=1/3$, so the convergence proofs are a simple exercise in real analysis\cite{R64}.  }
\begin{align}\label{eq:32}
\left|\sum_{\substack{i, j, k\\j\leq k}}\hat{f}_{i, jk}^2 p_{i, jk}(a, b, c)-\sum_{(i, j, k)\in Q}\hat{f}_{i, jk}^2 p_{i, jk}(a, b, c)\right| \leq \sum_{i, j\leq k} \hat{f}_{i, jk}^2 \delta_{(i, j, k)\in \neg Q} |p_{i, jk}(a, b, c)|,
\end{align}
\noindent where $\neg Q =\{(j, k, l)\in \mathbb{Z}_{\geq 0} \times \mathbb{Z}_{\geq 0} \times \mathbb{Z}_{\geq 0}: k\leq l \text{ and } (j, k, l)\notin Q\}$ and $\delta_{(i, j, k)\in \neg Q}$ is the standard discrete delta function which evaluates to $1$ if the condition is met and $0$ otherwise.  Observe that for $|a|$, $|b|$, $|c|$ $\leq 1$, $|p_{i, jk}(a, b, c)| \leq 6$ when $j\neq k$ and $|p_{i, jj}(a, b, c)| \leq 3$ so we can uniformly upper bound  $|p_{i, jk}(a, b, c)| \leq 3 \cdot{} 2^{1-\delta_{jk}}$.  Hence, we can upper bound \Cref{eq:32} as:
\begin{align}\label{eq:53}
\sum_{i, j\leq k} \hat{f}_{i, jk}^2 \delta_{(i, j, k)\in \neg Q} 3 \cdot{} 2^{1-\delta_{jk}}= \sum_{i, j\leq k} \hat{f}_{i, jk}^2 3 \cdot{} 2^{1-\delta_{jk}}-\sum_{(i, j, k)\in Q}  \hat{f}_{i, jk}^2 3 \cdot{} 2^{1-\delta_{jk}}\\
\nonumber =3\left(\sum_{i, j<k} \frac{\hat{f}_{i, jk}^2+\hat{f}_{i, kj}^2}{2}2^{1-\delta_{jk}}+\sum_{i, j} \hat{f}_{i, jj}^2-\sum_{(i, j, k) \in Q} \hat{f}_{i, jk}^2 2^{1-\delta_{jk}}\right),
\end{align}
\noindent where we used $\hat{f}_{i, jk}=\hat{f}_{i, kj}$.  We can now write \Cref{eq:53} as:
\begin{align*}
 =3\left(\sum_{i, j, k}\hat{f}_{i, jk}^2- \sum_{(i, j, k) \in Q} \hat{f}_{i, jk}^2 2^{1-\delta_{jk}}\right)=3\left(\frac{1}{3}- \sum_{(i, j, k) \in Q} \hat{f}_{i, jk}^2 2^{1-\delta_{jk}}\right).
\end{align*}

\end{proof}
We can use the same proof with the upper bound $|u_{i, jk}(a, b, c)| \leq 2^{1-\delta_{j, k}}$ to obtain:
\begin{lemma}\label{lem:13}
Let $Q$ be a finite subset of $\mathbb{Z}_{\geq 0} \times \mathbb{Z}_{\geq 0} \times \mathbb{Z}_{\geq 0}$ such that $k\leq l$ for all $(j, k, l)\in Q$.  Then, 
\begin{align*}
\left|\sum_{\substack{i, j, k\\j\leq k}}\hat{f}_{i, jk}^2 u_{i, jk}(a, b, c)-\sum_{(i, j, k)\in Q}\hat{f}_{i, jk}^2 u_{i, jk}(a, b, c)\right| \leq  1/3-\sum_{(i, j, k) \in Q}2^{1-\delta_{j, k}}\hat{f}_{i, jk}^2
\end{align*}
\noindent where $\delta_{j, k}$ is the standard discrete delta function and $|a|$, $|b|$, $|c|$ $\leq 1$.
\end{lemma}

%=====================================================================

\subsection{Optimizing Rational Functions over Polytopes}\label{sec:10}

The next set of important tools concerns optimizing rational functions over polytopes.  These observations will be used to simplify analysis, by reducing the parameter space we need to examine for a lower bound on the approximation factor.  

\begin{lemma}\label{lem:10}
Let $\mathcal{P}$ be some polytope, and $(A, B, C, a, b, c, k)$ be constants such that 
\begin{gather*}
k+[p, q, r] \cdot{} [A, B, C] \geq 0\\
{\rm and } \,\,k+[p, q, r] \cdot{} [a, b, c] \geq 0
\end{gather*}
\noindent for all $[p, q, r] \in \mathcal{P}$.  Then,
\begin{align}
\min_{[p, q, r]\in \mathcal{P}}\frac{k+[p, q, r] \cdot{} [A, B, C]}{k+[p, q, r] \cdot{} [a, b, c] } \geq \min_{[p, q, r]\in \mathcal{B}}\frac{k+[p, q, r] \cdot{} [A, B, C]}{k+[p, q, r] \cdot{} [a, b, c] },
\end{align}
\noindent where $\mathcal{B}$ is the set of extreme points for which the objective is defined.  
\end{lemma}
\begin{proof}
By definition any $[p, q, r]$ can be written as a convex combination of the extreme points of $\mathcal{P}$: $[p, q, r] =\sum_i \lambda_i [p_i, q_i, r_i]$.  Re-write both the numerator and denominator under this decomposition:
\begin{equation*}
\frac{\sum_i \lambda_i(k+[p_i, q_i, r_i] \cdot{} [A, B, C])}{\sum_i \lambda_i(k+[p_i, q_i, r_i] \cdot{} [a, b, c] )}.
\end{equation*}
Apply the elementary fact that if $\{a_1, ..., a_n, b_1, ..., b_n\}$ are non-negative constants, then:
\begin{equation*}
\frac{a_1+...+a_n}{b_1+...+b_n}\geq \min_{i:b_i\neq 0} \frac{a_i}{b_i}.
\end{equation*}
\noindent The lemma follows. 
\end{proof}

The above lemma indicates that if we have a restriction $[p, q, r] \in \mathcal{P}$ we need only examine the extreme points for a lower bound.  We can take it a step further and use the symmetries of our specific problem to demonstrate that a single extreme point is sufficient in each case, for the quadratic part of the objective.  We obtain the following:

\begin{lemma}\label{lem:17}
Fix $k\in \{1, 2, 3\}$.  Let $\mathcal{B}_1$ be the set of extreme points of $\mathcal{S}$, let $ \mathcal{B}_2$ be the set of extreme points of $\mathcal{T}$, and let $\mathcal{B}_3$ be the set of extreme points of $-\mathcal{S}$.  Let us further define 
\begin{gather*}
A= \mathbb{E} \left[ \frac{z_1 z_1'}{\sqrt{(z_1^2+z_2^2+z_3^2)((z_1')^2+(z_2')^2+(z_3')^2)}}\right], \\
\nonumber B= \mathbb{E} \left[ \frac{z_2 z_2'}{\sqrt{(z_1^2+z_2^2+z_3^2)((z_1')^2+(z_2')^2+(z_3')^2)}}\right], \\
\nonumber {\rm and } \,\,C= \mathbb{E} \left[ \frac{z_3 z_3'}{\sqrt{(z_1^2+z_2^2+z_3^2)((z_1')^2+(z_2')^2+(z_3')^2)}}\right],
\end{gather*}
\noindent where $[z_1, z_2, z_3, z_1', z_2', z_3'] \sim\mathcal{N}(0, \Sigma'(a, b, c))$.  Then,
\begin{equation*} 
\min_{\substack{[a, b, c] \in \mathcal{S}\\ [p, q, r]\in \mathcal{B}_k}} \frac{k+[p, q, r]\cdot{} [A, B, C]}{k+[p, q, r]\cdot{} [a, b, c]} = \min_{[a, b, c]\in \mathcal{S}} \frac{k+[p_k, q_k, r_k]\cdot{} [A, B, C]}{k+[p_k, q_k, r_k]\cdot{} [a, b, c]},
\end{equation*}
where 
\begin{equation*}
[p_k, q_k, r_k] =\begin{cases}
[-1, -1, -1] \text{ if }k=1\\
[2, 0, 0] \text{ if } k=2\\
[1, 1, 1] \text{ if } k=3
\end{cases}.
\end{equation*}

\end{lemma}
\begin{proof}
It is easy to see that the region $\mathcal{S}$ is invariant under permutation of coordinates, as well as the linear transformation $[a, b, c] \rightarrow [-a, -b, c]$.  Hence,
\begin{align*}
\min_{[a, b, c] \in \mathcal{S}} \frac{k+[p, q, r] \cdot{} [A, B, C]}{k+[p, q, r] \cdot{} [a, b, c]}=\min_{[a, b, c] \in \mathcal{S}} \frac{k+[-p, -q, r] \cdot{} [-A, -B, C]}{k+[-p, -q, r] \cdot{} [-a, -b, c]}\\
=\min_{[-a, -b, c] \in \mathcal{S}} \frac{k+[-p, -q, r] \cdot{} [A, B, C]}{k+[-p, -q, r] \cdot{} [a, b, c]}
=\min_{[a, b, c] \in \mathcal{S}} \frac{k+[-p, -q, r] \cdot{} [A, B, C]}{k+[-p, -q, r] \cdot{} [a, b, c]}.
\end{align*}
For every $k$, all the points in $\mathcal{B}_k$ are related by a permutation of coordinates followed by $[p, q, r] \rightarrow [-p, -q, r]$.  Hence, all points in $\mathcal{B}_k$ result in the same value when we minimize and we are free to pick an arbitrary extreme point. 

\end{proof}

\subsection{Technical Estimates for Functions of Interest}\label{sec:11}
\knote{Double check I didnt mess anything up when I changed def of q functions.}
There are several functions we will need bounds on.  The first set of functions (\Cref{lem:9} - \Cref{lem:7}) come from truncations of Hermite expansions, and the second set (\Cref{lem:20}) is related to the (exact) expansion for the linear part of the objective.  In this direction let us define the following three functions:
\begin{align*}
q_1(a, b, c)=\hat{f}_{1, 00}^2 p_{1, 00}(a, b, c) +\hat{f}_{1, 02}^2 p_{1, 02}(a, b, c) +\hat{f}_{3, 00}^2 p_{3, 00}(a, b, c)\\
\nonumber +3 (1/3-\hat{f}_{1, 00}^2-2\hat{f}_{1, 02}^2-\hat{f}_{3, 00}^2),\\
q_2(a, b, c)=\hat{f}_{1, 00}^2 a+\hat{f}_{1, 02}^2 a(b^2+c^2)+\hat{f}_{3, 00}^2 a^3\\
\nonumber -  (1/3 - \hat{f}_{1, 00}^2 - 2\hat{f}_{1, 02}^2 -\hat{f}_{3, 00}^2),\\
\nonumber \text{ and  }q_{3}(a, b, c)=\hat{f}_{1, 00}^2 p_{1, 00}(a, b, c) +\hat{f}_{1, 02}^2 p_{1, 02}(a, b, c) +\hat{f}_{3, 00}^2 p_{3, 00}(a, b, c)\\
\nonumber -3 (1/3-\hat{f}_{1, 00}^2-2\hat{f}_{1, 02}^2-\hat{f}_{3, 00}^2).
\end{align*}
\noindent Each of these functions will correspond to a different rank, $k$.  Since $[p, q, r]$ changes in each case, we will need different functions to lower bound the expectation, i.e. we may need to add or subtract the remainder term based on the signs of $[p, q, r]$.  For completeness, note the following analytic values of Hermite coefficients from \Cref{eq:50}:
\begin{equation*}
    \hat{f}_{1, 00}^2=\frac{8}{9 \pi},\,\,\,\,\,\,\,\,\,\,\, \hat{f}_{1, 02}^2=\frac{4}{225 \pi}, \,\,\,\,\,\,\text{and}\,\,\,\,\, \hat{f}_{3, 00}^2=\frac{4}{75 \pi}. 
\end{equation*}
With these definitions in hand, we can give the lower bounds we use for functions of interest:
\begin{lemma}[Rank $1$ Lemma]\label{lem:9}
It holds that 
\begin{align*}
\min_{[a, b, c] \in \mathcal{S}}\frac{1-\sum_{i, j\leq k}\hat{f}_{i, jk}^2 p_{i, jk} (a, b, c)}{1-a-b-c} \geq \frac{1-q_{1}(-1, -1, -1)}{4}=\frac{22}{15\pi}.
\end{align*}
\end{lemma}
\begin{lemma}[Rank $2$ Lemma]\label{lem:8}
It holds that 
\begin{align*}
\min_{[a, b, c] \in \mathcal{S}}\frac{1+\sum_{i, j\leq k}\hat{f}_{i, jk}^2 u_{i, jk} (a, b, c)}{1+a} \geq \frac{1+q_2(1, 0, 0)}{1+1}=\frac{1}{3}+\frac{24}{25\pi}.
\end{align*}
\end{lemma}
\begin{lemma}[Rank $3$ Lemma]\label{lem:7}
It holds that 
\begin{align*}
\min_{[a, b, c] \in \mathcal{S}}\frac{3+\sum_{i, j\leq k}\hat{f}_{i, jk}^2 p_{i, jk} (a, b, c)}{3+a+b+c} \geq \frac{3+q_3(1/3, 1/3, 1/3)}{3+1}=\frac{1}{2}+\frac{388}{405\pi}.
\end{align*}
\end{lemma}
\begin{proof}
See page \pageref{page:3} in the Appendix.  
\end{proof}

The final bound we need before proving the main technical lemmas follows:
\begin{lemma}\label{lem:20}
Let $a\in \mathbb{R}$ with $|a| \leq 1$, and let $b(a)=\frac{4a}{3\pi} \,_2F_1 \left[\begin{matrix}1/2, \,\, 1/2\\ 5/2 \end{matrix}; a^2 \right]$.  Then,
\begin{align*}
    b(a) \geq \begin{cases} (1/2) a \text{ if $a\leq 0$}\\ 4/(3\pi) a \text{ if $a\geq 0$}\end{cases}
\end{align*}
and
\begin{align*}
    b(a) \leq \begin{cases} 4/(3\pi) a \text{ if $a\leq 0$}\\ (1/2) a \text{ if $a\geq 0$}\end{cases}.
\end{align*}
\end{lemma}
\begin{proof}
If $a\geq 0$ it is easy to see that the first and second derivatives of $b(a)$ are positive, so the function is increasing and convex in this region.  this implies that $b'(0)a \leq b(a) \leq (b(1)-b(0))a$ in this region.  Similarly, it is easy to see that if $a\leq 0$ the first derivative is positive and the second is negative.  Hence, $b(a)$ is increasing and concave so the rest of the lemma follows.
\end{proof}

\subsection{Proof of \Cref{lem:26}}\label{sec:12}
Recall we are interested in evaluating
\begin{align*}
    \xi:=\min_{\substack{[a, b, c]\in \mathcal{S}\\ [p, q, r]\in \mathcal{P}_k}}&\frac{k+\mathbb{E} \left[ \frac{p z_1 z_1' +q z_2 z_2' +r z_3 z_3'}{\sqrt{(z_1^2+z_2^2+z_3^2)((z_1')^2+(z_2')^2+(z_3')^2)}}\right]}{k+a p+bq+cr}  .
\end{align*}

Observe that both the numerator and denominator correspond to $Tr(\mathcal{O}\phi)$ for some PSD observable $\mathcal{O}$ and some valid density matrix $\phi$ in all cases.  Hence, for fixed $[a, b, c]$ the numerator and denominator are non-negative and we can apply \Cref{lem:10}:
\begin{align*}
\xi \geq \min_{\substack{[a, b, c]\in \mathcal{S}\\ [p, q, r]\in \mathcal{B}_k}}\frac{k+\mathbb{E} \left[ \frac{p z_1 z_1' +q z_2 z_2' +r z_3 z_3'}{\sqrt{(z_1^2+z_2^2+z_3^2)((z_1')^2+(z_2')^2+(z_3')^2)}}\right]}{k+a p+bq+cr},
\end{align*}
where $\mathcal{B}_k$ is the set of extreme points of $\mathcal{P}_k$.  Now let us apply \Cref{lem:17}, which implies that we need only consider a particular $[p_k, q_k, r_k]$ for each case.  We obtain:
\begin{align*}
\xi \geq \min_{[a, b, c]\in \mathcal{S}}\frac{k+\mathbb{E} \left[ \frac{p_k z_1 z_1' +q_k z_2 z_2' +r_k z_3 z_3'}{\sqrt{(z_1^2+z_2^2+z_3^2)((z_1')^2+(z_2')^2+(z_3')^2)}}\right]}{k+a p_k+bq_k+cr_k},
\end{align*} 
for 
\begin{equation*}
    [p_k, q_k, r_k]=
    \begin{cases}
        [-1, -1, -1] \text{ if $k=1$}\\
        [2, 0, 0] \text{ if $k=2$}\\
        [1, 1, 1] \text{ if $k=3$} 
    \end{cases}.
\end{equation*}
Now we can expand in the Hermite polynomials according to \Cref{lem:11}:
\begin{align*}
\xi \geq \begin{cases} \min_{[a, b, c] \in \mathcal{S}} \frac{1-\sum_{i, j \leq k} \hat{f}_{i, jk}^2 p_{i, jk}(a, b, c)}{1-a-b-c} \text{ if $k=1$}\\
\min_{[a, b, c] \in \mathcal{S}} \frac{2+2\sum_{i, j \leq k} \hat{f}_{i, jk}^2 u_{i, jk}(a, b, c))}{2+2a} \text{ if $k=2$}\\
 \min_{[a, b, c] \in \mathcal{S}} \frac{3+\sum_{i, j \leq k} \hat{f}_{i, jk}^2 p_{i, jk}(a, b, c)}{3+a+b+c} \text{ if $k=3$}\end{cases}.
\end{align*}
%And apply \Cref{lem:12} and \Cref{lem:13} with $Q=\{(1, 0, 0), (1, 0, 2), (3, 0, 0)\}$ to obtain:
%\begin{align*}
%\min_{(a, b, c) \in \mathcal{S}} \frac{1-\sum_{i, j \leq k} \hat{f}_{i, jk}^2 p_{i, jk}(a, b, c)}{1-a-b-c} \geq \min_{(a, b, c) \in \mathcal{S}} \frac{1-q_{-}(a, b, c)}{1-a-b-c} \text{ if $k=1$} \\
%\min_{(a, b, c) \in \mathcal{S}} \frac{2+2\sum_{i, j \leq k} \hat{f}_{i, jk}^2 u_{i, jk}(a, b, c)}{2+2a} \geq \min_{(a, b, c) \in \mathcal{S}} \frac{1+s(a, b, c)}{1+a} \text{ if $k=2$}\\
%\min_{(a, b, c) \in \mathcal{S}} \frac{3+\sum_{i, j \leq k} \hat{f}_{i, jk}^2 p_{i, jk}(a, b, c)}{3+a+b+c} \geq \min_{(a, b, c) \in \mathcal{S}} \frac{3+q_{+}(a, b, c)}{3+a+b+c} \text{ if $k=3$}
%\end{align*}
Finally, we can apply Lemmas \ref{lem:12}-\ref{lem:13} and Lemmas \ref{lem:9}-\ref{lem:7}.

\subsection{Proof of \Cref{lem:21}}\label{sec:13}
Recall we are interested in finding a lower bound for:
\begin{align}\label{eq:58}
     \xi:= \min \bigg(k + \mathbb{E} \left[\frac{p z_1 z_1'+q z_2 z_2'+r z_3 z_3'}{\sqrt{(z_1^2+z_2^2+z_3^2)((z_1')^2+(z_2')^2+(z_3')^2)}} \right] +\mathbb{E} \left[ \frac{t_i x_1 s_1}{\sqrt{x_1^2+x_2^2+x_3^2}|s_1|}\right]\\
    \nonumber +\mathbb{E} \left[\frac{t_j y_1 s_2}{\sqrt{y_1^2+y_2^2+y_3^2}|s_2|} \right] \bigg) /\bigg(k+ap +bq+cr+t_i d_i +t_j d_j \bigg)
    \end{align}
    \vspace{-25 pt}
    \begin{gather}
    \nonumber s.t. -k \leq ap+bq+rc \leq 4-k,\\
    \nonumber -k \leq ap+bq+rc+t_id_i+t_j d_j \leq 4-k,\\
    \nonumber |t_id_i| \leq l,\\
    \nonumber \text{ and }|t_j d_j| \leq l.
\end{gather}

We can evaluate the expectations according to \Cref{lem:19} and \Cref{lem:11} to write the objective of \Cref{eq:58} as:
\begin{align*}
    \xi=\min \bigg(k + p \sum_{i, j\leq k} \hat{f}_{i, jk}^2 u_{i, jk}(a, b, c)+q \sum_{i, j\leq k} \hat{f}_{i, jk}^2 u_{i, jk}(b, a, c)+r \sum_{i, j\leq k} \hat{f}_{i, jk}^2 u_{i, jk}(c, a, b)\\
    \nonumber +t_i \frac{4 d_i}{3 \pi}\,_2 F_1 \left[ \begin{matrix} 1/2, \, 1/2\\ 5/2 \end{matrix}; d_i^2 \right]+t_j \frac{4 d_i}{3 \pi}\,_2 F_1 \left[ \begin{matrix} 1/2, \, 1/2\\ 5/2 \end{matrix}; d_i^2 \right] \bigg)/\bigg(k+ap +bq+cr+t_i d_i +t_j d_j \bigg)\\
    \end{align*}
    \vspace{-45 pt}
    \begin{align*}
    s.t. \,\,... \,.
\end{align*}

We can apply \Cref{lem:13} with $Q=\{(1, 0, 0)\}$\footnote{Note that the remainder term, $3(1/3-8/(9\pi))$, can be taken negative for a lower bound becuase the denominator is always non-negative by the constraints.}, and \Cref{lem:20} to obtain the minimization:
\begin{gather*}
    \xi \geq \min \frac{k+\frac{8}{9\pi}(ap+bq+rc)-3(\frac{1}{3}-\frac{8}{9 \pi})+c_i t_i d_i +c_j t_j d_j}{k+ap+bq+rc+t_i d_i+t_j d_j} \\
    s.t. -k \leq ap+bq+rc \leq 4-k,\\
    -k \leq ap+bq+rc+t_id_i+t_j d_j \leq 4-k,\\
    c_i, c_j \in \left\{\frac{1}{2}, \frac{4}{3\pi}\right\},\\
    |t_i d_i| \leq l,\\
    \text{ and }|t_j d_j| \leq l.
\end{gather*}
Let us ``coarse grain'' the optimization by setting $v_1=ap+bq+rc$, $v_2=t_id_i$ and $v_3=t_jd_j$. The optimization we obtain by coarse graining is a lower bound on the original optimization by the following:  Introduce variables $v_1$, $v_2$ and $v_3$ into the program with constraints $v_1=ap+bq+rc$, $v_2=t_i d_i$ and $v_3=t_j d_j$.  We can rewrite the objective and the original constraints so that they are written entirely in $(v_1, v_2, v_3)$, while maintaining an equivalent problem (we have simply renamed things).  Then, disregard the constraints $v_1=ap+bq+rc$, $v_2=t_i d_i$ and $v_3=t_j d_j$.  Since we are de-constraining the problem the new objective is non-increasing, and the variables $(a, b, c, p, q, r, t_i, t_j)$ are redundant and can be ignored.  The program obtained from this is:
\begin{gather}\label{eq:54}
    \xi \geq \min \frac{k+\frac{8}{9\pi}v_1-3(\frac{1}{3}-\frac{8}{9 \pi})+c_i v_2 +c_j v_3}{k+v_1+v_2+v_3} \\
    \nonumber s.t. -k \leq v_1 \leq 4-k,\\
    \nonumber -k \leq v_1+v_2+v_3 \leq 4-k,\\
    \nonumber c_i, c_j \in \left\{\frac{1}{2}, \frac{4}{3\pi}\right\},\\
    \nonumber |v_2| \leq l,\\
    \nonumber \text{ and }|v_3| \leq l.
\end{gather}
We will proceed by solving the above optimization problem for any fixed choice of $(c_i, c_j)$.  For any fixed $(c_i, c_j)$, $[v_1, v_2, v_3]$ is constrained to a polytope, so we may apply \Cref{lem:10} (the denominator is non-negative by the constraints, and it is easy to check that the numerator is always non-negative) to lower bound the optimization problem as:
\begin{align*}
    \min ... \,\,s.t. [v_1, v_2, v_3] \in \mathcal{Q}_k \geq \min ... \,\,s.t. [v_1, v_2, v_3] \in \mathcal{B}_k,
\end{align*}
where $\mathcal{Q}_k$ is the polytope defined via the constraints of \Cref{eq:54}, and $\mathcal{B}_k$ is the corresponding set of extreme points.  Hence, for any fixed $(c_i, c_j)$, we have a lower bound on the optimization by checking extreme points of some polytope.  By checking these points and checking all possible values of $c_i$, $c_j$ we obtain the following minima:
\begin{equation*}
    \xi \geq \begin{cases} 2/\pi-1/4 \approx 0.387 \text{ if $k=1$}\\   16/(9\pi) \approx 0.565 \text{ if $k=2$}\\  3/8+11/(9\pi) \approx 0.764 \text{ if $k=3$} \end{cases}.
\end{equation*}

\section{Proof of Main Theorems}
 We will first provide a proof of \Cref{thm:1}, since techniques used in this proof will be similar to the general case.  

\subsection{Proof of \Cref{thm:1} (strictly quadratic case)}\label{sec:1}
The proof will follow the outline in \Cref{sec:analysis-overview}, but here we will point to the formal components needed in each step. Suppose we are given an instance of \Cref{prob:QLHP} and let 
\begin{equation}
\rho=\bigotimes_{i=1}^n \left( \frac{\mathbb{I}+\theta_{iX} \sigma^1 +\theta_{iY}\sigma^2 +\theta_{iZ}\sigma^3}{2}\right)
\end{equation}
\noindent be the (random) output of the approximation algorithm (\Cref{alg:2}).  The total expected cost is:
\begin{align}
\mathbb{E} \left[ \sum_{ij} Tr[H_{ij} \rho]\right]= \sum_{ij} \mathbb{E}\left[Tr[H_{ij} \rho]\right].
\end{align}
Now observe that, if we can find a worst case lower bound on 
\begin{equation}\label{eq:49}
\frac{\mathbb{E}[Tr[H_{ij} \rho]]}{rank(H_{ij})/4+\text{Tr}[M^* C_{ij}]},
\end{equation}
\noindent we would be able to uniformly bound each term, and hence be able to provide \Cref{thm:1}.  The majority of the technical work presented here is aimed at precisely such a goal\footnote{The careful reader will notice there is an edge case, where $rank(H_{ij})/4+\text{Tr}[M C_{ij}]=0$.  However, for this case the SDP is earning $0$, and since $H_{ij}\semigeq 0$, any $\rho $ has objective at least the SDP.}.  Note that we can WLOG assume that $H_{ij}$ is a strictly quadratic projector since the positive constant of proportionality will cancel from the numerator and denominator.

Let $D$ be a nonzero off-diagonal block of $C_{ij}$ corresponding to qubits $i$ and $j$.  Formally, $D\in \mathbb{R}^{3\times 3}$ such that $D_{k, l}=C_{ij}(\sigma_i^k, \sigma_j^l)$ for all $k,l$ in $[3]$.  Let $M^*$ be the optimal moment matrix with a Cholesky vectors  $\{\mathbf{v}_{ik}\}$ for all $i\in [n], k\in [3]$ and let $\mathbf{v}_0 $ be the Cholesky vector corresponding to index $\mathbb{I}$.  By this we mean that $\mathbf{v}_{ik}^T \mathbf{v}_{jl}=M^*(\sigma_i^k, \sigma_j^l)$ and $\mathbf{v}_{0}^T \mathbf{v}_{jl}=M^*(\mathbb{I}, \sigma_j^l)$.  Define $V_i=[\mathbf{v}_{i1}, \mathbf{v}_{i2}, \mathbf{v}_{i3}]$ and $V_j=[\mathbf{v}_{j1}, \mathbf{v}_{j2}, \mathbf{v}_{j3}]$.  We can rewrite our expectation of interest as:
\begin{equation}
\mathbb{E}[Tr[H_{ij} \rho]]=\frac{rank(H_{ij})}{4}+2\mathbb{E}  \left[ \frac{\mathbf{r}^T V_i D V_j^T\mathbf{r}}{||V_i^T \mathbf{r}|| \,\, ||V_j^T \mathbf{r}||}\right],
\end{equation}
where the factor of $2$ is used to account for the symmetry of $C_{ij}$.  Note, as we have previously mentioned, $sign(\mathbf{v}_0^T \mathbf{r})$ is squared, so it has no effect on the objective.  Hence we can rewrite \Cref{eq:49}:
\begin{equation}
    \frac{\mathbb{E}[Tr[H_{ij} \rho]]}{rank(H_{ij})/4+\text{Tr}[M C_{ij}]}=\frac{k+\mathbb{E}  \left[ \frac{\mathbf{r}^T V_i (8 D) V_j^T\mathbf{r}}{||V_i^T \mathbf{r}|| \,\, ||V_j^T \mathbf{r}||}\right]}{k+Tr[V_i (8 D) V_j^T]}.
\end{equation}
Observe that $V_i^T V_j$ must correspond to a valid quadratic $2$-moment for density matrix $\rho_{i j}$ in the relaxation, since the SDP is constrained in that way (recall the SDP is defined so that the moment matrix is consistent with {\it local} density matrices).  Similarly, $8 D$ must correspond to a $2$-moment for a strictly quadratic projector of rank $k$ by our assumptions on the problem instance (compare with \Cref{eq:52}).  Hence, we can apply \Cref{lem:6} (note that $C$ in the context of that Lemma is the same as $8D$ in the context of this proof), which implies
\begin{align}
\frac{\mathbb{E}[Tr[H_{ij} \rho]]}{rank(H_{ij})/4 +\text{Tr} [M^* C_{ij}]}\geq \min_{\substack{[a, b, c]\in \mathcal{S}\\ [p, q, r]\in \mathcal{P}_k}}\frac{k+\mathbb{E} \left[ \frac{p z_1 z_1' +q z_2 z_2' +r z_3 z_3'}{\sqrt{(z_1^2+z_2^2+z_3^2)((z_1')^2+(z_2')^2+(z_3')^2)}}\right]}{k+a p+bq+cr},
\end{align}

\noindent where for fixed $[a, b, c]$, $[\mathbf{z}, \mathbf{z}'] \sim \mathcal{N}(0, \Sigma'(a, b, c))$ and $\mathcal{P}_k =\begin{cases}
-\mathcal{S} \text{ if } k=3\\
\mathcal{T} \text{ if } k=2\\
\mathcal{S} \text{ if } k=1
\end{cases}
$.  Note that we have met exactly the conditions of \Cref{lem:26}, so we may apply it to obtain the theorem.  

%
%Then apply \Cref{lem:2}:
%\begin{equation}
%\frac{\mathbb{E}[(\mathbb{I}_{ij}-\ket{\psi_{ij}}\bra{\psi_{ij}})\rho_{ij}]}{3/4 +\text{Tr} [M C^e]} \geq \min_{\substack{(a, b, c)\in \mathcal{S}\\ (p, q, r)\in B}}\frac{3+\mathbb{E} \left[ \frac{p z_1 z_1' +q z_2 z_2' +r z_3 z_3'}{\sqrt{(z_1^2+z_2^2+z_3^2)((z_1')^2+(z_2')^2+(z_3')^2)}}\right]}{3+a p+bq+cr}
%\end{equation}
%\noindent where $B=\{(1, 1, 1), (-1, -1, 1), (-1, 1, -1), (1, -1, -1)\}$.  Now apply the discussion at the end of \Cref{sec:2}, which shows that WLOG:
%\begin{equation}
%\frac{\mathbb{E}[(\mathbb{I}_{ij}-\ket{\psi_{ij}}\bra{\psi_{ij}})\rho_{ij}]}{3/4 +\text{Tr} [M C^e]} \geq \min_{\substack{(a, b, c)\in \mathcal{S}}}\frac{3+\mathbb{E} \left[ \frac{ z_1 z_1' + z_2 z_2' + z_3 z_3'}{\sqrt{(z_1^2+z_2^2+z_3^2)((z_1')^2+(z_2')^2+(z_3')^2)}}\right]}{3+a +b+c}
%\end{equation}
%
%Then expand in Hermite polynomials as described in \Cref{sec:3}:
%\begin{equation}
%\frac{\mathbb{E}[(\mathbb{I}_{ij}-\ket{\psi_{ij}}\bra{\psi_{ij}})\rho_{ij}]}{3/4 +\text{Tr} [M C^e]} \geq \min_{\substack{(a, b, c)\in \mathcal{S}}}\frac{3+\sum_{i, j\leq k} \hat{f}_{i, jk}^2 p_{i, jk}(a, b, c)}{3+a +b+c}
%\end{equation}
%
%\Cref{lem:5} and \Cref{lem:4} then imply 
%\begin{equation}
%\frac{\mathbb{E}[(\mathbb{I}_{ij}-\ket{\psi_{ij}}\bra{\psi_{ij}})\rho_{ij}]}{3/4 +\text{Tr} [M C^e]} \geq q(1/3, 1/3, 1/3) \approx 0.802
%\end{equation}
%\noindent where $q(a, b, c)$ is defined in \Cref{eq:18}.  The theorem follows.  

\subsection{Proof of \Cref{thm:7}}
Just as before, we will lower bound the expected approximation factor by arguing a worst case lower bound on the expected approximation factor of a single edge.  Let $H_{ij}=P_{ij}\otimes \mathbb{I}_{[n]\setminus\{i, j\}}$ be some projector of rank $k$ and let $C_{ij}$ be the corresponding value matrix defined by \Cref{eq:52}. 

Just as in the previous proof, let $D\in \mathbb{R}^{3\times 3}$ be the ``quadratic part'' of $C_{ij}$ with elements defined as $D_{k, l}=C_{ij}(\sigma_i^k, \sigma_j^l)$ for all $k$, $l\in [3]$.  Let $\mathbf{w}_i/\mathbf{w}_j$ be the 1-local part of $C_{ij}$ corresponding to qubit $i/j$.  Specifically, $\mathbf{w}_i$ and $\mathbf{w}_j$ are vectors in $\mathbb{R}^3$ satisfying 
\begin{align}
(\mathbf{w}_i)_k=C_{ij}(\sigma_i^k, \mathbb{I}),\\
\text{ and }(\mathbf{w}_j)_l=C_{ij}(\mathbb{I}, \sigma_j^l).
\end{align}
Let $M^*$ be the optimal moment matrix with a Cholesky vectors  $\{\mathbf{v}_{ik}\}$ for all $i\in [n], k\in [3]$ and let $\mathbf{v}_0 $ be the Cholesky vector corresponding to index $\mathbb{I}$.  By this we mean that $\mathbf{v}_{ik}^T \mathbf{v}_{jl}=M(\sigma_i^k, \sigma_j^l)$ and $\mathbf{v}_{0}^T \mathbf{v}_{jl}=M(\mathbb{I}, \sigma_j^l)$.  Lastly, we will define $V_i=[\mathbf{v}_{i1}, \mathbf{v}_{i2}, \mathbf{v}_{i3}]$ and $V_j=[\mathbf{v}_{j1}, \mathbf{v}_{j2}, \mathbf{v}_{j3}]$, matrices composed of Cholesky vectors for a single qubit.  

Just as in \Cref{thm:1}, we can write the expected approximation factor as the quotient of the expected objective from the rounding algorithm and the value that the relaxation obtains.  We can write:
\begin{align}
    & \frac{\mathbb{E} [Tr[H_{ij} \rho]]}{k/4+Tr[C_{ij} M]}\\
    \nonumber & =\bigg( \frac{k}{4}+2\mathbb{E}\left[ \frac{\mathbf{r}^T V_i D V_j^T \mathbf{r}}{||V_i^T\mathbf{r}||\,\, ||V_j^T \mathbf{r}||}\right]+2\mathbb{E}\left[\frac{(\mathbf{r}^T V_i \mathbf{w}_i)( \mathbf{v}_0^T \mathbf{r})}{||V_i^T \mathbf{r}|| \,\, |\mathbf{v}_0^T \mathbf{r}|} \right]+2\mathbb{E}\left[\frac{(\mathbf{r}^T V_j \mathbf{w}_j)( \mathbf{v}_0^T \mathbf{r})}{||V_j^T \mathbf{r}|| \,\, |\mathbf{v}_0^T \mathbf{r}|} \right]\bigg)\\
    \nonumber &\hspace{4.5 cm} / \bigg(\frac{k}{4}+2Tr[V_i D V_j^T] +2 Tr[V_i \mathbf{w}_i \mathbf{v}_0^T]+2Tr[V_j \mathbf{w}_j \mathbf{v}_0^T] \bigg).
\end{align}

Multiplying the numerator and denominator by $4$ we can find an expected approximation factor of 
\begin{align}\label{eq:57}
    =\frac{k+\mathbb{E}\left[ \frac{\mathbf{r}^T V_i (8 D) V_j^T \mathbf{r}}{||V_i^T\mathbf{r}||\,\, ||V_j^T \mathbf{r}||}\right]+\mathbb{E}\left[\frac{\mathbf{r}^T V_i (8 \mathbf{w}_i) \mathbf{v}_0^T \mathbf{r}}{||V_i^T \mathbf{r}|| \,\, |\mathbf{v}_0^T \mathbf{r}|} \right]+\mathbb{E}\left[\frac{\mathbf{r}^T V_j (8 \mathbf{w}_j) \mathbf{v}_0^T \mathbf{r}}{||V_j^T \mathbf{r}|| \,\, |\mathbf{v}_0^T \mathbf{r}|} \right]}{k+Tr[V_i (8 D) V_j^T] + Tr[V_i (8 \mathbf{w}_i) \mathbf{v}_0^T]+Tr[V_j (8 \mathbf{w}_j) \mathbf{v}_0^T]}.
\end{align}

We can apply \Cref{lem:6} and \Cref{lem:18} to convert \Cref{eq:57} to a standard form:
\begin{align}
    =\frac{k + \mathbb{E} \left[\frac{p z_1 z_1'+q z_2 z_2'+r z_3 z_3'}{\sqrt{(z_1^2+z_2^2+z_3^2)((z_1')^2+(z_2')^2+(z_3')^2)}} \right] +\mathbb{E} \left[ \frac{t_i x_1 s_1}{\sqrt{x_1^2+x_2^2+x_3^2}|s_1|}\right] +\mathbb{E} \left[\frac{t_j y_1 s_2}{\sqrt{y_1^2+y_2^2+y_3^2}|s_2|} \right] }{k+ap +bq+cr+t_i d_i +t_j d_j},
\end{align}
where random variables are distributed according to:\knote{At this point maybe remind the reader of physical interpretations of u's and d's.  }
\begin{gather*}
    [z_1, z_2, z_3, z_1', z_2', z_3']\sim \mathcal{N}(0, \Sigma'(a, b, c)),\\
    [s_1, x_1, x_2, x_3] \sim \mathcal{N}(0, \Sigma_1),\\
    \text{ and }[s_2, y_1, y_2, y_3] \sim \mathcal{N}(0, \Sigma_2)
\end{gather*}
with
\begin{align}
    \Sigma_1=\begin{bmatrix} 1 & d_i & 0 & 0 \\d_i & 1 & 0 & 0\\ 0 & 0 & 1 & 0\\0 & 0 & 0 & 1\end{bmatrix} \,\,\,\text{ and }\,\,\, \Sigma_2=\begin{bmatrix} 1 & d_j & 0 & 0 \\d_j & 1 & 0 & 0\\ 0 & 0 & 1 & 0\\0 & 0 & 0 & 1\end{bmatrix}.
\end{align}
Now let us note several constraints on the parameters.  Just as in the proof of \Cref{thm:1}, \Cref{lem:6} implies that $[a, b, c] \in \mathcal{S}$ and $[p, q, r] \in \mathcal{P}_k$ where $\mathcal{P}_k$ is defined in \Cref{eq:66}.  \Cref{eq:63} implies that $-k \leq ap+bq+rc \leq 4-k$.  Note further that $ap+bq+rc+t_i d_i +t_j d_j=4 Tr(P \rho)-k$ for some $2$ qubit projector $P$ of rank $k$ and some density matrix on two qubits (here we are crucially using the constraints that force the moment matrix to correspond to a valid $2$ qubits marginals), hence \Cref{eq:60} implies $-k \leq ap+bq+rc+t_i d_i +t_j d_j\leq 4-k$.  Finally, let us note that Equations \ref{eq:65}, \ref{eq:64} imply that $|t_i d_i|, |t_j d_j| \leq 1$ if $k=1$ or $3$ and $|t_i d_i|, |t_j d_j| \leq 2$ if $k=2$.  Let us say that $|t_i d_i|, |t_j d_j|\leq l$ where $l$ is $1$ if $k=1$ or $3$ and $l$ is $2$ if $k=2$.  Hence, we have met the conditions required for \Cref{lem:21} which completes the proof.  
\section{Conclusion}

In this work we have demonstrated several new approximation algorithms for interesting cases of the $2$-Local Hamiltonain problem.  As is the theme in many works \cite{B07, G12, B16}, we have given evidence that the geometry of $2$-Local interactions can drastically effect approximability for traceless Hamiltonians since we demonstrate the the bipartite case has a constant factor approximation algorithm and the unconstrained case is known to have no constant factor algorithm \cite{B19}.  In addition to this, we have given a novel approximation algorithm/analysis for the $2$-Local Hamiltonian with local terms which are also projectors.  This is especially interesting given the the scarcity of approximation algorithms for quantum problems.  Indeed, the rank $3$ case, has been open for some time \cite{G12, H20}.  Furthermore, we have provided new techniques for rounding to product states that we believe will have additional applications in quantum information.  Our rounding algorithm is quite ``natural'' given the solution of the SDP, and the ability to understand the expectation through Hermite polynomial analysis seems likely to extend to other Hamiltonians/problems.

Future work includes tightening our analysis with more sophisticated methods.  As stated in the introduction, we believe the strictly quadratic case is a very interesting special case of the general problem, and given the symmetry of the problem, the analysis may be amenable to techniques from algebraic geometry.  In particular, we hope to replace the tedious analysis from lemmas \ref{lem:9}-\ref{lem:7} with a more sophisticated approach that can handle higher order Hermite expansions.  Such a method would also probably allow for a better analysis of the approximation algorithm itself, potentially allowing us to prove what we believe to be the attained approximation factor for our rounding algorithm.  Another natural direction is removing the assumption we have that the terms are projectors, and replacing it with generic PSD terms.  Implicit in this work is a $0.387$-approximation to such problems, by writing a given PSD term as a positive combination of projectors, the question is whether or not a better approximation can be obtained, perhaps with additional assumptions.

\subsection{Acknowledgments }

Sandia National Laboratories is a multimission laboratory managed and operated by National Technology and Engineering Solutions of Sandia, LLC., a wholly owned subsidiary of Honeywell International, Inc., for the U.S. Department of Energy’s National Nuclear Security Administration under contract DE-NA-0003525.  This work was supported by the U.S. Department of Energy, Office of Science, Office of Advanced Scientific Computing Research, Accelerated Research in Quantum Computing and Quantum Algorithms Teams programs.
\bibliography{mybib}{}
\bibliographystyle{alpha}

\appendix 
\noindent {\huge \textbf{Appendix} }
\vspace{10 pt}

The Appendix covers a range of results.  First we present connections between classical constraint satisfaction problems and quantum local Hamiltonian problems.  Next, we list some combinatorial identities (\Cref{sec:17}) used in this work.  These identities are all standard themselves, of follow from standard techniques \cite{P97}.  Third, we cover some calculations relevant to the Hermite expansions used in the paper (\Cref{sec:18}).  In the paper we were primarily concerned with the $3$-dimensional function $f$ (\Cref{eq:75}), but here we give Hermite coefficients for the analogous $r$-dimensional function for arbitrary $r\in \mathbb{Z}_{\geq 0}$.  This calculation allows us to reproduce technical tools from another work \cite{B11}.  Then, in \Cref{sec:19}, we give our calculus-based lower bounds for the functions of interest.  The following section (\Cref{sec:20}) concerns some proofs for the properties of quantum states/projectors we used in the paper.  Lastly, in \Cref{sec:21} we give an approximation algorithm for the bipartite, traceless $2$-Local Hamiltonian problem which achieves a novel approximation factor.

\section{Connecting $2$-Local Hamiltonians and Classical $2$-CSPs}
\label{sec:classical-quantum}
Here we motivate why $2$-Local Hamiltonian is a generalization of classical $2$-CSPs.  We start with the Pauli matrices as defined in \Cref{eq:paulis} in \Cref{sec:quantum-notation}. We will use the notions of a qubit and classical Boolean variable interchangeably and assume we have $n$ such objects.  Here we favor different notation than in \Cref{sec:quantum-notation} to better emphasize the connection between boolean polynomials and Hamiltonians. 

We use the notation $P_i$ to denote a Pauli matrix $P \in \{X=\sigma^1,Y=\sigma^2,Z=\sigma^3\}$ acting on qubit $i$, i.e. $P_i := \mathbb{I} \otimes \mathbb{I} \otimes \ldots \otimes P \otimes \ldots \otimes \mathbb{I} \in \mathbb{C}^{2^n \times 2^n}$, where the $P$ occurs at position $i$.  As an example comparing notation, for $n=4$, $X_1Y_3 = \sigma_1^1 \otimes \sigma_3^2 \otimes \mathbb{I}_{2,4} = X \otimes \mathbb{I} \otimes Y \otimes \mathbb{I}$.  This notation will allow us to avoid explicitly using tensor-product notation, which is replaced with matrix multiplication; moreover, it will allow us to express Hermitian matrices as multilinear polynomials.

For $n=2$, consider the matrix $C = \frac{1}{2}(\mathbb{I}-Z_1Z_2) = \ket{01}\bra{01} + \ket{10}\bra{10}$:
\begin{equation*}
C =
    \begin{bmatrix}
    0 & 0 & 0 & 0\\
    0 & 1 & 0 & 0\\
    0 & 0 & 1 & 0\\
    0 & 0 & 0 & 0\\
    \end{bmatrix}
\text{ and }Q = 
    \begin{bmatrix}
    0 & 0 & 0 & 0\\
    0 & \frac{1}{2} & -\frac{1}{2} & 0\\
    0 & -\frac{1}{2} & \frac{1}{2} & 0\\
    0 & 0 & 0 & 0\\
    \end{bmatrix}.
\end{equation*}
Finding a maximum eigenvector of $C$ corresponds to solving the classical $2$-CSP Max $\frac{1}{2}(1-z_1z_2)$ over $z_1,z_2 \in \{\pm 1\}$.  This is because $C$ is diagonal matrix with Boolean entries, so that maximum eigenvectors correspond to $\ket{b_1b_2}$ where the corresponding diagonal entry of $C$ (i.e., $\bra{b_1b_2}C\ket{b_1b_2}$) is 1.   If we associate $\ket{0}$ with the value 1 and $\ket{1}$ with -1, then $\frac{1}{2}(1-z_1z_2)$ is the Boolean indicator function for the diagonal of $C$.  More generally, if we are given a \emph{density matrix} $\rho \in \mathbb{C}^{4 \times 4}$, satisfying $\rho \succeq 0$ and $\text{Tr}[\rho] = 1$, then $\text{Tr}[C\rho] \in [0,1]$ may be viewed as the expected value of $\frac{1}{2}(1-z_1z_2)$ over a probability distribution specified by the diagonal of $\rho$.  In this way, any rank $r$ projector in $\mathbb{C}^{4 \times 4}$ corresponds to Boolean clause on 2 variables, where $r$ indicates the number of satisfying assignments.  Such a clause may be extended to act on variables $i$ and $j$ among $n$ variables by taking $C_{ij} \otimes \mathbb{I}_{[n]\setminus\{i,j\}}$ ($C_{ij}$ denotes that $C$ acts on qubits $i$ and $j$).  Summing over such extended clauses yields a $2$-local Hamiltonian instance corresponding to a Boolean $2$-CSP.  Theorem~\ref{thm:6} demonstrates how Max $2$-QSAT generalizes Max $2$-SAT.

Multilinear polynomials over $\mathbb{I}$ and $Z_i$ of degree $k$ are diagonal $k$-local Hamiltonians encoding classical $k$-CSPs, while multilinear polynomials over $\mathbb{I}, X_i, Y_i, Z_i$ of degree $k$ are more general non-diagonal $k$-local Hamiltonians that are able to represent quantum phenomena. 

We obtain a quantum generalization of Boolean $2$-CSPs by considering non-diagonal matrices (in the standard basis).  Let $Q=\frac{1}{4}(\mathbb{I}-X_1X_2-Y_1Y_2-Z_1Z_2)$ as depicted above.  Here $C$ represents Max Cut on a single edge, while $Q$ is a quantum generalization of Max Cut~\cite{G19}. The unique maximum eigenvector of $Q$ is $\ket{\psi} = \frac{1}{\sqrt{2}}(\ket{01}-\ket{10})$, with corresponding density matrix $\rho = \ket{\psi}\bra{\psi}$, achieving an eigenvalue of $\text{Tr}[Q\rho]=1$.  This is a non-classical solution, which in general may require an exponential amount of information for representation.  In contrast \emph{product states}, which exhibit no quantum entanglement, may be represented as
\begin{equation*}
    \rho = \prod_i \frac{1}{2}(\mathbb{I}+\alpha_i X_i + \beta_i Y_i + \gamma_i Z_i).
\end{equation*}
Such $\rho$ has $\text{Tr}[\rho]=1$ by construction, and $\rho \succeq 0$ is guaranteed if $\alpha_i^2 + \beta_i^2 + \gamma_i^2 \leq 1$ for all $i$.  A product state maximizing $\text{Tr}[Q\rho]$ is $\rho = \frac{1}{4}(\mathbb{I}+Z_1)(\mathbb{I}-Z_2) = \ket{01} \bra{01}$, achieving a value of $\frac{1}{2}$.  Thus for $Q$ there is a gap of $\frac{1}{2}$ between the best product state and the best arbitrary quantum density matrix.

\section{Identities}\label{sec:17}

We will make use of several identities concerning the generalized Hypergeometric function.  The following standard technique will be used in proving our combinatorial identities:
\begin{theorem}[\cite{P97}]\label{thm:4}
\begin{enumerate}
    \item  Let $Q=\sum_{m=0}^p t_m$ be some finite sum.  If 
    \begin{equation}\label{eq:55}
    \frac{t_{m+1}}{t_m}=\frac{(m+a)(m+b)}{(m+c)(m+1)}
    \end{equation} 
    \noindent for $m=\{0, ..., p-1\}$, then $Q=t_0 \,\,_2F_1 \left[\begin{matrix} a, \,\, b\\ c\end{matrix} ; \,1\right]$.
    \item  Let  $Q=\sum_{m=0}^\infty t_m x^m$. If \Cref{eq:55} holds for all $m$, then $Q=t_0 \,\,_2F_1 \left[\begin{matrix} a, \,\, b\\ c\end{matrix} ; \,x\right]$.
\end{enumerate}
\end{theorem}

Now, we list several identities that will be used in proving the results in the paper:
\begin{lemma}\label{lem:16}
\,

\begin{enumerate}
\item \begin{equation}\label{eq:77}
\int_0^\pi \sin^a t \cos^b t \,\,dt=\begin{cases}
\frac{2(a-1)!! (b-1)!!}{(b+a)!!} \text{a is odd, b is even}\\
\frac{\pi(a-1)!! (b-1)!!}{(b+a)!!} \text{a is even, b is even}
\end{cases}.
\end{equation}

\item  \begin{equation}\label{eq:78}
\int_0^{2\pi}\sin^a t \cos^b t dt =\frac{2\pi (a-1)!! (b-1)!!}{ (a+b)!!} \text{for a and b even.}
\end{equation}.

\item  \cite{A48}  If $n$ is an integer:
\begin{equation}\label{eq:79}
\Gamma(n+1/2)=\sqrt{\pi} \frac{(2n-1)!!}{2^n}.
\end{equation}

%\item  \cite{S72}  \knote{Delete this, but be careful about the refs attached to it.  }The Hermite polynomials have an explicit description as:
%\begin{align}
%    h_n(z)=\frac{\sqrt{n!}}{2^{n/2}} \sum_{l=0}^{n/2} \frac{(-1)^{n/2-l} (\sqrt{2}z)^{2 l}}{(2l)! (n/2-l)!}
%    {\rm \text{ if $n$ is even, }} \\ 
%    \nonumber {\rm and} \,\,\, h_n(z)=\frac{\sqrt{n!}}{2^{n/2}} %\sum_{l=0}^{(n-1)/2}\frac{(-1)^{(n-1)/2-l} (\sqrt{2} z)^{2 l+1}}{(2l+1)!((n-1)/2-l)!} \text{ if $n$ is odd.}
%\end{align}

\item\cite{P97} (Gauss's Hypergeometric Theorem) If $c > a+b$ (and $a, b, c\in \mathbb{R}$):
\begin{equation}\label{eq:80}
         \,_2F_1 \left[\begin{matrix} a, \,\, b\\ c\end{matrix} ; \,1\right]=\frac{\Gamma(c) \Gamma(c-a-b)}{\Gamma(c-a) \Gamma(c-b)}.
\end{equation}

\item  If $r$ is odd:
\begin{equation}\label{eq:81}
    \sum_{m=0}^p \frac{(-1)^m (2m+r-1)!!}{(2m+r)!!} \binom{p}{m}=\frac{(r-1)!!}{(2p+r)(2p+r-2)...(2p+1)}.
\end{equation}
\item  Generalized Vandermonde's identitiy:
\begin{equation}\label{eq:82}
\sum_{k_1+...+k_p=m} \binom{n_1}{k_1} \binom{n_2}{k_2}... \binom{n_p}{k_p}=\binom{n_1+...+n_p}{m}.
\end{equation}
\item  
\begin{equation}\label{eq:83}
\sum_{q=0}^\infty a^{2q+1} \frac{(2q+1)!}{2^{2q+1}(1+2q)^2 (3+2q) (q!)^2}=\frac{a}{6} \,\,_2 F_1 \left[\begin{matrix} 1/2, \,\, 1/2 \\ 5/2\end{matrix}; a^2 \right].
\end{equation}
\end{enumerate}
\end{lemma}
\begin{proof}
    \Cref{eq:77} and \Cref{eq:78} are easy to establish with integration by parts and induction, and \Cref{eq:82} follows from a simple counting argument.  To prove \Cref{eq:81} we use \Cref{thm:4}, item $1$.  Let 
    \begin{equation}
    \sum_{m=0}^p \frac{(-1)^m (2m+r-1)!!}{(2m+r)!!} \binom{p}{m}=\sum_{m=0}^p t_m.
    \end{equation}
    We can calculate:
    \begin{equation}
    \frac{t_{m+1}}{t_m}=\frac{(m-p)(m+(r+1)/2)}{(m+(r+2)/2) (m+1)}.
    \end{equation}
    \Cref{thm:4} then implies that:
    \begin{equation}
    \sum_{m=0}^p t_m= t_0  \,\,_2F_1 \left[\begin{matrix} -p, \,\, (r+1)/2\\ (r+2)/2\end{matrix} ; \,1\right]=t_0\frac{\Gamma((r+2)/2) \Gamma(p+1/2) }{\Gamma((r+2)/2+p) \Gamma(1/2)} .
    \end{equation}
    where we applied \Cref{eq:80} for the last equality.  Now apply \Cref{eq:79} and our assumption that $r$ is odd to obtain:
    \begin{equation}
    \sum_{m=0}^p t_m=t_0 \frac{r!! (2p-1)!!}{(r+2p)!!}=\frac{(r-1)!!}{(2p+r)(2p+r-2)...(2p+1)}.
    \end{equation}
    We prove \Cref{eq:83} in the same way.  In the context of \Cref{thm:4}, item $2$ we  can take write: 
    \begin{equation}
    \sum_{q=0}^\infty a^{2q+1} \frac{(2q+1)!}{2^{2q+1}(1+2q)^2 (3+2q) (q!)^2}=a\sum_{q=0}^\infty t_q (a^2)^q,
    \end{equation}
    and calculate:
    \begin{align}
        \frac{t_{q+1}}{t_q}= \frac{\frac{(2(q+1)+1)!}{2^{2(q+1)+1} (1+2(q+1))^2(3+2(q+1))[(q+1)!]^2}}{\frac{(2q+1)!}{2^{2q+1} (1+2q)^2 (3+2q)(q!)^2}}=\frac{(q+1/2)^2}{(q+5/2)(q+1)},
    \end{align}
    where $t_0=a/6$.  \Cref{eq:83} follows.  
\end{proof}

\section{Calculation of Hermite Coefficients}\label{sec:18}
We will make use of these identities to determine the Hermite coefficients for our function of interest.  First, let us compute the moments:
\begin{lemma}\label{lem:22}
Let $\mathbf{z}\sim \mathcal{N}(0, \mathbb{I})$.  If $b_1, ..., b_r$ are even,
\begin{equation} 
\mathbb{E}\left[\frac{z_1^{b_1} z_2^{b_2}...z_r^{b_r}}{\sqrt{z_1^2+...+z_r^2}}\right]=c_r \frac{(b_1+...+b_r +r-3)!! (b_1-1)!! (b_2-1)!! ... (b_r-1)!!}{(b_1+b_2+...+b_r+r-2)!!},
\end{equation}
where $c_r$ is $\sqrt{2/\pi}$ if $r$ is odd and $\sqrt{\pi/2}$ if $r$ is even.
\end{lemma}
\begin{proof}
From the definition of the multivariate Gaussian,
\begin{align*}
\xi:= \mathbb{E} \left[ \frac{z_1^{b_1} z_2^{b_2} ... z_d^{b_d}}{\sqrt{z_1^2+z_2^2+...+z_d^2}}\right]=\int  \frac{z_1^{b_1} z_2^{b_2} ... z_d^{b_d}}{\sqrt{z_1^2+z_2^2+...+z_d^2}} \left( \frac{1}{(2 \pi)^{d/2}} e^{-||z||^2/2}\right) dz.
\end{align*}
Let us re-write this using generalized spherical coordinates:
\begin{align*}
\xi = \frac{1}{(2 \pi)^{d/2}}  \int \frac{1}{r} \bigg(r\cos\phi_1\bigg)^{b_1} \bigg(r \sin\phi_1 \cos\phi_2\bigg)^{b_2} \bigg(r\sin \phi_1 \sin\phi_2 \cos\phi_3\bigg)^{b_3} ... \bigg(r \sin\phi_1 \\
\nonumber ... \sin \phi_{d-3} \cos \phi_{d-2}\bigg)^{b_{d-2}} \bigg(r \sin\phi_1 ... \sin \phi_{d-2}\cos\phi_{d-1}\bigg)^{b_{d-1}} \bigg(r\sin \phi_1 ... \sin\phi_{d-2} \sin \phi_{d-1}\bigg)^{b_d}\\
\nonumber e^{-r^2/2}\bigg(r^{d-1} \sin^{d-2}\phi_1 \sin^{d-3} \phi_2 ... \sin\phi_{d-2}\bigg) dr d\phi_1 ... d\phi_{d-1},
\end{align*}
where $\phi_1, ... \phi_{d-2} \in [0, \pi]$ and $\phi_{d-1}\in [0, 2 \pi]$.  Now we can break it up into smaller integrals:
\begin{align*}
\xi =\frac{1}{(2 \pi)^{d/2}}\int_0^\infty r^{b_1+b_2+...+b_d +d-2}e^{-r^2/2} dr \int_0^\pi \cos^{b_1} \phi_1 (\sin\phi_1)^{b_2+...+b_d+d-2} d\phi_1 \\
\nonumber \int_0^\pi\cos^{b_{2}} \phi_2 (\sin\phi_2)^{b_3+...+b_d +d-3} d\phi_2 \int_0^\pi \cos^{b_3} \phi_3 (\sin\phi_3)^{b_4+...+b_d+d-4}d\phi_3 \\
\nonumber ...\int_0^\pi \cos^{b_{d-3}}\phi_{d-3} (\sin \phi_{d-3})^{b_{d-2} +b_{d-1}+b_d +2} d\phi_{d-3}\int_0^\pi \cos^{b_{d-2}}\phi_{d-2} (\sin\phi_{d-2})^{b_{d-1}+b_d+1} d\phi_{d-2}\\
\nonumber \cdot{}\int_0^{2 \pi} \cos^{b_{d-1}}\phi_{d-1} \sin^{b_d} \phi_{d-1}d\phi_{d-1}.
\end{align*}
Now we will apply \Cref{eq:77} and \Cref{eq:78}.  There are two cases here depending on if $d$ is even or odd.  If d is even:
\begin{align*}
\xi =\frac{1}{(2 \pi)^{d/2}}\left[ 2^{(b_1+...+b_d+d-3)/2} \Gamma\left(\frac{b_1+...+b_d+d-1}{2}\right)\right] \left[\frac{\pi (b_1-1)!! (b_2+...+b_d+d-3)!!}{(b_1+b_2+...+b_d+d-2)!!} \right]\\
\nonumber \left[\frac{2 (b_{2}-1)!! (b_3+...+b_d+d-4)!!)}{(b_2+b_3+...+b_d+d-3)!!}\right]\left[\frac{\pi (b_3-1)!! (b_4+...+b_d+d-5)!!}{(b_3+...+b_d+d-4)!!}\right]\\
\nonumber ... \left[\frac{\pi (b_{d-3}-1)!! (b_{d-2}+b_{d-1}+b_d+1)!!)}{(b_{d-3}+...+b_d+2)!!}\right]\left[\frac{2(b_{d-2}-1)!! (b_{d-1}+b_d)!!}{(b_{d-2}+b_{d-1}+b_d+1)!!}\right]\\
\nonumber \left[\frac{2\pi (b_{d-1}-1)!!(b_d-1)!!}{(b_{d-1}+b_d)!!}\right].
\end{align*}
Note that the product telescopes.  Since $d$ is even, we can count a factor $(2\pi)^{(d-2)/2+1}=(2\pi)^{d/2}$.  Continuing, 
\begin{align*}
\xi =\left[ 2^{(b_1+...+b_d+d-3)/2} \Gamma\left(\frac{b_1+...+b_d+d-1}{2}\right)\right]\frac{(b_1-1)!! (b_2-1)!! ... (b_d-1)!!}{(b_1+b_2+...+b_d+d-2)!!}\\
\nonumber =\sqrt{\pi} 2^{(b_1+...+b_d+d-3)/2} \frac{(b_1+...+b_d+d-3)!!}{2^{(b_1+...+b_d+d-2)/2}}\frac{(b_1-1)!! (b_2-1)!! ... (b_d-1)!!}{(b_1+b_2+...+b_d+d-2)!!}\\
\nonumber =\sqrt{\frac{\pi}{2}}\frac{(b_1+...+b_d+d-3)!!(b_1-1)!! (b_2-1)!! ... (b_d-1)!!}{(b_1+b_2+...+b_d+d-2)!!}.
\end{align*}
In the odd case, the constant factor we get is $2 (2 \pi)^{(d-3)/2+1}$, so the overall constant is $\sqrt{2/\pi}$.  We end up getting the same thing with a different constant:
\begin{align*}
\xi=\sqrt{\frac{2}{\pi}}\frac{(b_1+...+b_d+d-3)!!(b_1-1)!! (b_2-1)!! ... (b_d-1)!!}{(b_1+b_2+...+b_d+d-2)!!}.
\end{align*}
\end{proof}
Given these moments we can explicitly calculate the Hermite coefficients:
\begin{lemma}\label{lem:25}
Let $f(\mathbf{z})=z_1/\sqrt{z_1^2+...+z_r^2}$ for $r$ odd, and let it have Hermite expansion $f(\mathbf{z})=\sum_\mu \hat{f}_{\mu} h_{\mu}(\mathbf{z})$.  Then, if $\mu_1$ is odd and $\mu_i$ is even for all $i\in \{2, ..., r\}$,
\begin{equation}
    \hat{f}_\mu =\sqrt{\frac{2}{\pi}} \frac{\sqrt{\mu_1! ... \mu_r!}(-1)^p}{(\mu_1-1)!! \mu_2!! ... \mu_r!!}\frac{(r-1)!!}{(2p+r)(2p+r-2)...(2p+1)}
\end{equation}
where $p=(\mu_1+...+\mu_r-1)/2$.  Otherwise, $\hat{f}_\mu=0$.
\end{lemma}

\begin{proof}
The fact that $\hat{f}_\mu=0$ when $\mu_1$ is even or at least one of $\mu_i$ is odd for $i\in\{2, ..., r\}$ follows from the same idea in the paper, i.e. the integral of an anti-symmetric function is zero.  For the other values of $\mu$, use \Cref{def:3} to obtain:
\begin{align*}
    \hat{f}_\mu=\mathbb{E}\bigg[\frac{z_1}{\sqrt{z_1^2+...+z_r^2}} \frac{\sqrt{\mu_1 !}}{2^{\mu_1/2}}\sum_{l_1=0}^{(\mu_1-1)/2}\frac{(-1)^{(\mu_1-1)/2-l_1} (\sqrt{2}z_1)^{2l_1+1}}{(2l_1+1)! \left(\frac{\mu_1-1}{2}-l_1 \right)!}  \frac{\sqrt{\mu_2!}}{2^{\mu_2/2}}\sum_{l_2=0}^{\mu_2/2}\frac{(-1)^{\mu_2/2-l_2} (\sqrt{2}z_2)^{2l_2}}{(2l_2)! \left( \frac{\mu_2}{2}-l_2\right)!}   \\
    \nonumber ... \frac{\sqrt{\mu_r!}}{2^{\mu_r/2}}\sum_{l_r=0}^{\mu_r/2}\frac{(-1)^{\mu_r/2-l_r}(\sqrt{2}z_r)^{2 l_r}}{(2l_r)!\left( \frac{\mu_r}{2}-l_r\right)!}\bigg]
\end{align*}
\begin{align*}
    =\frac{\sqrt{\mu_1 !... \mu_r!}}{2^{(\mu_1+...+\mu_r)/2}} \sum_{l_1, ..., l_r} \frac{(-1)^{(\mu_1+...+\mu_r-1)/2-(l_1+...+l_r)}\sqrt{2}^{2(l_1+...+l_r)+1}}{\left[ (2 l_1+1)! (2l_2)! ... (2 l_r)!\right] \left[\left(\frac{\mu_1-1}{2}-l_1\right)! \left( \mu_2/2-l_2\right)! ... \left( \mu_r/2-l_r\right)!\right]}\\
    \nonumber \cdot{} \mathbb{E} \left[\frac{z_1^{2l_1+2} z_2^{2 l_2} ...z_r^{2 l_r}}{\sqrt{z_1^2+...+z_r^2}} \right].
\end{align*}
By \Cref{lem:22},
\begin{align*}
    \hat{f}_\mu=\frac{\sqrt{\mu_1 !... \mu_r!}}{2^{(\mu_1+...+\mu_r)/2}} \sum_{l_1, ..., l_r} \frac{(-1)^{(\mu_1+...+\mu_r-1)/2-(l_1+...+l_r)}\sqrt{2}^{2(l_1+...+l_r)+1}}{\left[ (2 l_1+1)! (2l_2)! ... (2 l_r)!\right] \left[\left(\frac{\mu_1-1}{2}-l_1\right)! \left( \mu_2/2-l_2\right)! ... \left( \mu_r/2-l_r\right)!\right]}\\
    \nonumber \cdot{} c_r \frac{(2l_1+2+2l_2+...+2l_r+r-3)!! (2l_1+1)!!(2l_2-1)!!...(2l_r-1)!!}{(2l_1+2+2l_2+...+2l_r+r-2)!!}
\end{align*}
\begin{align}\label{eq:76}
    =c_r\frac{\sqrt{\mu_1 !... \mu_r!}}{2^{(\mu_1+...+\mu_r)/2}}  (-1)^{(\mu_1+...+\mu_r-1)/2}\sum_{l_i}\frac{\sqrt{2}^{2(l_1+...+l_r)+1} (2l_1+...+2l_r+r-1)!!(-1)^{l_1+...+l_r}}{(2 l_1)!! (2l_2)!!... (2 l_r)!! \left(\frac{\mu_1-1}{2}-l_1 \right)! (\mu_2/2-l_2)! ... (\mu_r/2-l_r)!}\\
    \nonumber \cdot{} \frac{1}{(2l_1+...+2l_r+r)!!}.
\end{align}
Now let us apply the identity:
\begin{align*}
    (2 l)!! (q-l)!=2^l l! (q-l)!=2^l \frac{q!}{\binom{q}{l}},
\end{align*}
to rewrite \Cref{eq:76} as:
\begin{align*}
    \hat{f}_\mu=\frac{\sqrt{2} c_r\sqrt{\mu_1 !... \mu_r!}}{2^{(\mu_1+...+\mu_r)/2}}  (-1)^{(\mu_1+...+\mu_r-1)/2}\sum_{l_i}\frac{2^{l_1+...+l_r}(2l_1+...+2l_r+r-1)!! \binom{(\mu_1-1)/2}{l_1} \binom{\mu_2/2}{l_2}... \binom{\mu_r/2}{l_r}}{2^{l_1+...+l_r} \left( \frac{\mu_1-1}{2}\right)! (\mu_2/2)! ... (\mu_r/2)! (2l_1+...+2l_r+r)!!}\\
    \end{align*}
    \vspace{-1 cm}
    \begin{align*}
    =\frac{\sqrt{2} c_r\sqrt{\mu_1!...\mu_r!}(-1)^{(\mu_1+...+\mu_r-1)/2}}{2^{(\mu_1+...+\mu_r)/2} \left( \frac{\mu_1-1}{2}\right)!(\mu_2/2)!...(\mu_r/2)!}\sum_{m=0}^{(\mu_1+...+\mu_r-1)/2} (-1)^m \frac{(2m+r-1)!!}{(2m+r)!!} \\
    \nonumber \cdot{}\sum_{l_1+...+l_r=m} \binom{(\mu_1-1)/2}{l_1} \binom{\mu_2/2}{l_2}... \binom{\mu_r/2}{l_r}.
\end{align*}
Now apply the generalized Vandermonde's equality, \Cref{eq:82}:
\begin{align*}
    \hat{f}_\mu=\frac{\sqrt{2} c_r\sqrt{\mu_1!...\mu_r!}(-1)^{(\mu_1+...+\mu_r-1)/2}}{2^{(\mu_1+...+\mu_r)/2} \left( \frac{\mu_1-1}{2}\right)!(\mu_2/2)!...(\mu_r/2)!} \sum_{m} (-1)^m \frac{(2m+r-1)!!}{(2m+r)!!} \binom{(\mu_1-1)/2+\mu_2/2+...+\mu_r/2}{m}.
\end{align*}
Let us set $p=(\mu_1-1)/2+\mu_2/2+...+\mu_r/2$.  Using \Cref{eq:81},
\begin{align*}
    \hat{f}_\mu=\frac{\sqrt{2}c_r\sqrt{\mu_1!...\mu_r!}(-1)^{p}}{2^{(\mu_1+...+\mu_r)/2} \left( \frac{\mu_1-1}{2}\right)!(\mu_2/2)!...(\mu_r/2)!} \frac{(r-1)!!}{(2p+r)(2p+r-2)+...+(2p+1)}\\
    \nonumber =\frac{\sqrt{2} c_r \sqrt{\mu_1! ... \mu_r!} (-1)^p (r-1)!!}{2^{(\mu_1+...+\mu_r)/2}(\mu_1-1)!! (\mu_2)!!... (\mu_r)!!2^{-p}(2p+r)...(2p+1)}\\
    \nonumber =\frac{c_r \sqrt{\mu_1!... \mu_r!} (-1)^p (r-1)!!}{(\mu_1-1)!!\mu_2!!...\mu_r!!(2p+r)...(2p+1)}.
\end{align*}
Note that $c_r=\sqrt{2/\pi}$ since $r$ is odd.  
\end{proof}
We can do a ``sanity check'' for the calculations presented thus far by giving a proof sketch of a specialization of our result which yields a result from \cite{B11}:
\begin{proposition}
\begin{equation*}
    \sum_{i, j, k} \hat{f}_{i, jk}^2 a^{i+j+k}=\frac{8}{3\pi}  \,_2F_1 \left[\begin{matrix} 1/2, \,\, 1/2\\ 5/2\end{matrix} ; \,a^2\right]
\end{equation*}
\end{proposition}

\begin{proof}
{\it (sketch)}  It is sufficient to collect terms $\hat{f}_{i, jk}$ with fixed $i+j+k$ and compare them to the Taylor series given in the statement of Lemma 5.2.1. from \cite{B11}.  For this we must demonstrate that:
\begin{equation*}
    \sum_{i+j+k=2l+1} \frac{i! j! k!}{[(i-1)!! j!! k!!]^2}=\frac{1}{3} \frac{(2l-1)!!(1+2l)(3+2l)}{(2l)!!},
\end{equation*}
where we have canceled some common factors.  This can be rewritten as:
\begin{equation*}
    =\sum_{i=[0:2:2l]}\left( \frac{(i+1)!}{(i!!)^2 }\sum_{j=[0:2:2l-i]} \frac{j! (2l-i-j)!}{[j!! (2l-i-j)!!]^2}\right),
\end{equation*}
where we have redefined $i$ and used the (MATLAB) notation $i=[0:2:2l]$ to indicate a sum from $0$ to $2l$ by $2$.  The second sum can be proven to be $1$ using  \Cref{thm:4} and Gauss's theorem (\Cref{eq:80}).  The remainder of the proof follows from the identity
\begin{equation*}
\sum_{q=0}^l \frac{(2q+1)!}{[(2q)!!]^2}=\frac{1}{3}  \frac{(2l+3)!!}{(2l)!!},
\end{equation*}
which is simple to establish with induction.  
\end{proof}
\label{page:2}
\begin{proof}[Proof of \Cref{lem:19}]
We have two functions of interest here:
\begin{gather*}
	g(z)=sign(z)=\sum_i \hat{g}_i h_i(z)\\
	\nonumber \text{ and }f(z_1, z_2, z_3)=\frac{z_1}{\sqrt{z_1^2+z_2^2+z_3^2}}=\sum_{i, j, k}\hat{f}_{i, jk} h_i(z_1)h_j(z_2)h_k(z_3).
\end{gather*}
for $\{\hat{g}_i\}$ and $\{\hat{f}_{i, jk}\}$ Hermite expansions given in \Cref{lem:23} and \Cref{lem:24}.  We want:
\begin{align*}
\mathbb{E}[g(z) f(z_1, z_2, z_3)]=\sum_{i, j, k}\hat{f}_{i, jk} \hat{g}_i \mathbb{E}[h_i(z_1) h_i(z)] \mathbb{E} [h_j(z_2)] \mathbb{E}[h_k(z_3)]=\sum_{i \,\, odd} \hat{f}_{i, 00}\hat{g}_i a^i,
\end{align*}
where in the second equality we used that fact that $\delta_{0j}=\mathbb{E}[h_0(z_2)h_j(z_2)]$.  If we apply the identity $(i-1)!!=[(i-1)/2]! \,\,2^{\frac{i-1}{2}}$ and use the previously derived values of $\hat{f}_{i, 00}$ and $\hat{g}_i$, the sum is:
\begin{align*}
\sum_{i \,\, odd} a^i \hat{g}_i \,\hat{f}_{i, 00}=\sum_{i \,\, odd} a^i \left( \frac{\sqrt{i!} (-1)^{q}}{2^{i/2}} \frac{2}{\sqrt{\pi} q!(1+2 q)}\right)\left( \frac{4}{\sqrt{\pi}}\frac{\sqrt{i! }}{2^{(i)/2}} \frac{(-1)^{q}}{\left( \frac{i-1}{2}\right)! (1+2q)(3+2 q)} \right)\\
\nonumber =\frac{8}{\pi}\sum_{i\,\, odd} a^i\frac{i!}{2^i(1+2q)^2(3+2q)(q!)^2}.
\end{align*}
where $q=(i-1)/2$.  We can equivalently sum over $q$ to get a sum from 0 to infinity by 1, (i.e. not every other integer):
\begin{align*}
=\frac{8}{\pi}\sum_{q=0}^\infty a^{2q+1} \frac{(2q+1)!}{2^{2q+1}(1+2q)^2(3+2q)(q!)^2}=\frac{8}{\pi} \frac{a}{6}\,_2F_1 \left[\begin{matrix} 1/2, \,\, 1/2\\ 5/2\end{matrix} ; \,a^2\right] \\
\nonumber =\frac{4a}{3\pi}\,_2F_1 \left[\begin{matrix} 1/2, \,\, 1/2\\ 5/2\end{matrix} ; \,a^2\right],
\end{align*}
where we used \Cref{eq:83}.

\end{proof}

We also present a proof of this fact for completeness:
\label{page:4}
\begin{proof}[Proof of \Cref{prop:1}, Item $3$]

\begin{align}\label{eq:20}
\left| \mathbb{E}[f(\mathbf{z}) f(\mathbf{z}')]-\sum_{i, j, k \leq N} \hat{f}_{i, j, k}^2 a^i b^j c^k\right|^2\\
\nonumber = \left|\mathbb{E} \left[\left( f(\mathbf{z})-\sum_{i, j, k \leq N} \hat{f}_{i, j, k} h_i (\mathbf{z}_1) h_j(\mathbf{z}_2) h_k(\mathbf{z}_3)\right) \left( f(\mathbf{z}')+\sum_{i, j, k \leq N} \hat{f}_{i, j, k} h_i (\mathbf{z}_1') h_j(\mathbf{z}_2') h_k(\mathbf{z}_3')\right) \right]\right|^2\\
\nonumber \leq \mathbb{E} \left[ \left(f(\mathbf{z})-\sum_{i, j, k \leq N} \hat{f}_{i, j, k} h_i (\mathbf{z}_1) h_j(\mathbf{z}_2) h_k(\mathbf{z}_3) \right)^2\right] \mathbb{E} \left[\left( f(\mathbf{z}')+\sum_{i, j, k \leq N} \hat{f}_{i, j, k} h_i (\mathbf{z}_1') h_j(\mathbf{z}_2') h_k(\mathbf{z}_3')\right)^2 \right]
\end{align}
\noindent by Cauchy-Schwarz.  The second term is:
\begin{align*}
\mathbb{E} \left[ f(\mathbf{z}')^2+2 \sum_{i,j,k} \hat{f}_{i, j, k \leq N} h_i(\mathbf{z}_1') h_j(\mathbf{z}_2') h_k(\mathbf{z}_3') f(\mathbf{z}')\right]+\sum_{i, j, k \leq N}\hat{f}_{i, j, k}^2=\langle f, f \rangle +3\sum_{i, j, k \leq N} \hat{f}^2_{i, j, k}\\
\nonumber \leq 4 \langle f, f \rangle.
\end{align*}
\noindent So to get the RHS of \Cref{eq:20} less than $\epsilon$, take $N$ large enough that 
\begin{equation*}
\mathbb{E} \left[ \left(f(\mathbf{z})-\sum_{i, j, k \leq N} \hat{f}_{i, j, k} h_i (\mathbf{z}_1) h_j(\mathbf{z}_2) h_k(\mathbf{z}_3) \right)^2\right] \leq \frac{\epsilon}{4 \langle f,f \rangle }.
\end{equation*}
\noindent We know this can be done since $\sum_{i, j, k\leq N}\hat{f}_{i, j, k} h_i(\mathbf{z}_1) h_j(\mathbf{z}_2) h_k(\mathbf{z}_3) \rightarrow f$ in this norm.
\end{proof}

\section{Lower Bounds for the Expectation of an Edge}\label{sec:19}

For the result we have described, we crucially needed several lower bounds on the expectation of a single edge.  In this direction we present proofs of lemmas \ref{lem:9}-\ref{lem:7}:
\label{page:3}
\begin{proof}[Proof of \Cref{lem:9}]
\textbf{Rank 1 case}
We will divide the region $\mathcal{S}$ into two regions and use different order expansions in each region for a bound.  Define:
\begin{align*}
\mathcal{S}_1=conv \{[-1/2, -1/2, -1], [-1/2, -1, -1/2], [-1, -1/2, -1/2], [-1, -1, -1]\},\\
\text{ and }\mathcal{S}_2=\mathcal{S}\setminus \mathcal{S}_1.
\end{align*}
Let us apply \Cref{lem:12} with $Q=\{(1, 0, 0)\}$ to obtain 
\begin{equation*}
\frac{1-\sum_{i, j \leq k} \hat{f}_{i, jk}^2 p_{i, jk}(a, b, c)}{1-a-b-c} \geq \frac{1-\hat{f}_{1, 00}^2(a+b+c)-3(1/3-\hat{f}_{1, 00}^2)}{1-a-b-c}.
\end{equation*}
Define $x=a+b+c$ and observe that the RHS is a function of $x$, so define 
\begin{equation*}
l(x):=\frac{1-\hat{f}_{1, 00}^2 x-3(1/3-\hat{f}_{1, 00}^2)}{1-x}.
\end{equation*}
It is easy to see that this function is increasing as a function of $x$, and that $l(-2)\geq 0.47$, so we have the lower bound on $\mathcal{S}_2$:
\begin{equation}
\min_{(a, b, c) \in \mathcal{S}_2}\frac{1-\sum_{i, j \leq k} \hat{f}_{i, jk}^2 p_{i, jk}(a, b, c)}{1-a-b-c} \geq 0.47.
\end{equation}
It remains to bound the function over $\mathcal{S}_1$.  For this we will need a higher order Hermite expansion to get a good approximation factor.  Let us apply \Cref{lem:12} with $Q=\{(1, 0, 0), (1, 0, 2), (3, 0, 0)\}$ to obtain:
\begin{align}
\frac{1-\sum_{i, j \leq k} \hat{f}_{i, jk}^2 p_{i, jk}(a, b, c)}{1-a-b-c} \geq \bigg{(} 1-\hat{f}_{1, 00}^2 p_{1, 00}(a, b, c) -\hat{f}_{1, 02}^2 p_{1, 02}(a, b, c) -\hat{f}_{3, 00}^2 p_{3, 00}(a, b, c)\\
\nonumber -3 (1/3-\hat{f}_{1, 00}^2-2\hat{f}_{1, 02}^2-\hat{f}_{3, 00}^2)\bigg{)} / \bigg{(}1-a-b-c \bigg{)} =: q(a, b, c).
\end{align}
We will show that $q(a, b, c)$ is minimized at $[-1, -1, -1]$, demonstrating the lemma.  First calculate the partial derivative with respect to $a$:
\begin{equation*}
\frac{\partial q}{\partial a}= \frac{4 (115 + 6 a^3 - b^2 (1 + 2 b) + 2 a (-1 + b + c) (b + c) - 
   c^2 (1 + 2 c) + a^2 (-9 + 10 b + 10 c))}{225 (-1 + a + b + c)^2 \pi}.
\end{equation*}
We can easily check that this is never $0$ since the constant term is large enough that the numerator is always strictly positive:
\begin{align*}
115 + 6 a^3 - b^2 (1 + 2 b) + 2 a (-1 + b + c) (b + c) - 
   c^2 (1 + 2 c) + a^2 (-9 + 10 b + 10 c)\\
   \nonumber \geq 115 - 6 - (1 + 2 ) - 2  (3) (2) - (1 + 2 ) + (-9 - 10 - 10 )=62.
\end{align*}
Note that we applied the bounds $|a|\leq 1$, $|b| \leq 1$ and $|c|\leq 1$.  Hence there are no internal critical points.  
\begin{figure}
    \centering
    \includegraphics[scale=0.75]{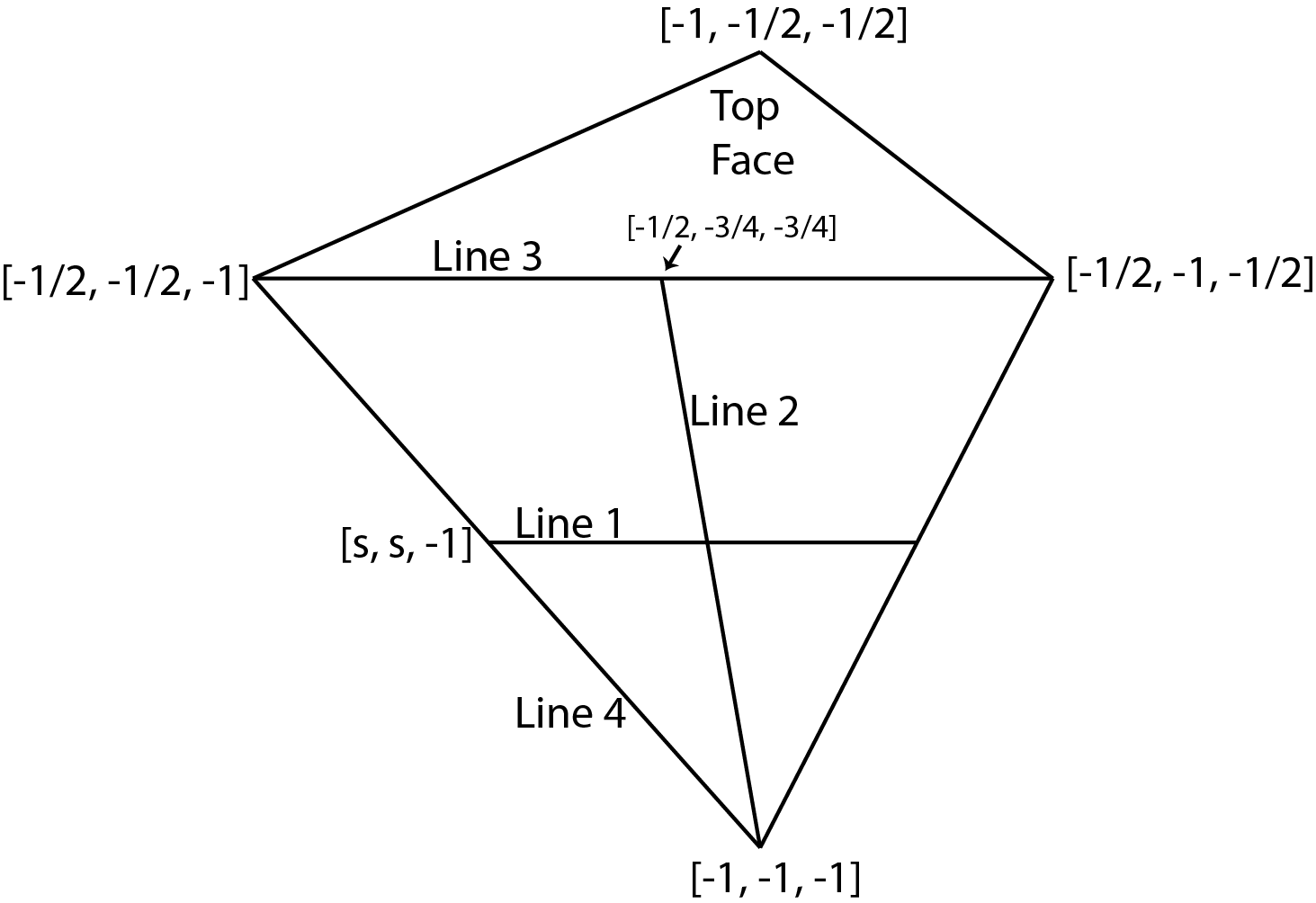}
    \caption{Relevant parameterizations for proof of \Cref{lem:9}, Region $\mathcal{S}_1$}
    \label{fig:1}
\end{figure}

Let us now consider a horizontal line  along one of the ``vertical faces'' (line $1$ in \Cref{fig:1}).  One such a line has the form $[a, b, c]=[s, s, -1]+t[0, -1, 1]$ for $s$ some fixed constant $s\in [-1/2, -1]$ and $t$ some parameter that varies along the line between $0$ and $s+1$.  We can calculate the derivative of $q$ along this line to obtain
\begin{align*}
\frac{d q(s, s-t, -1+t)}{dt}=-\frac{4 (-4 + 5 s) (1 + s - 2 t)}{225 \pi (-1 + s)}.
\end{align*}
The only potential critical points here are in the middle of the line when $t=(s+1)/2$ and when $s=4/5$.  The later point is outside of our region of interest, so we can ignore it.  For the first point it is clear we must examine the point which is in the middle of the line for each value of $s$.  For this consider the line which goes down the ``middle'' of a vertical face (line $2$ in \Cref{fig:1}).  This line can be parameterized as $[a, b, c]=[-1, -1, -1]+t[1/2, 1/4, 1/4]$ for $t$ between $0$ and $1$.  We can once again evaluate the derivative of $q$ along this line:
\begin{align*}
\frac{dq}{dt}=\frac{560 + t (344 + t (-109 + 11 t))}{450 \pi (-4 + t)^2}.
\end{align*}
We can again see that the constant term is large enough that the numerator is always strictly positive, so the only possible critical points are the endpoints.  

It remains to check the ``top face'', line $3$, and line $4$.  Observe that we can parameterize the top face as $[a, b, -2-a-b]$ under some suitable restriction of $a$ and $b$.  The gradient under this parameterization can be written as
\begin{align*}
\frac{\partial q(a, b, -2-a-b)}{\partial a}=\frac{8 (2 + 2 a + b) (8 + 3 b)}{675 \pi},\\
\frac{\partial q(a, b, -2-a-b)}{\partial b}=\frac{8 (8 + 3 a) (2 + a + 2 b)}{675 \pi}.
\end{align*}
The only possible critical point is $[-2/3, -2/3, -2/3]$.  

For the line $3$ parameterize as $[a, b, c]=[-1/2, -1/2, -1]+t[0, -1/2, 1/2]$ and calculate:
\begin{equation*}
\frac{d q(-1/2, -1/2(1+t), -1+t/2)}{dt}=\frac{13 (-1 + 2 t)}{675 \pi}.
\end{equation*}
Only critical point is in the middle of the line.  

Lastly, we need to check line $4$: $[a, b, c]=[-1, -1, -1]+t[1, 1, 0]$.  The derivative is:
\begin{equation*}
\frac{d q(-1+t, -1+t, -1)}{dt}=\frac{4 (35 + 52 t - 37 t^2 + 8 t^3)}{225 \pi (-2 + t)^2}.
\end{equation*}
Finding the roots of the numerator (which can be determined analytically by a well known formula), it is clear there are no critical points along the line, so we only need to check its endpoints.

The previous analysis as well as the symmetries of the function $q(a, b, c)$ imply that we need only check the following points to find a minimum:
\begin{align*}
q(-2/3, -2/3, -2/3)\approx 0.508,\\
\nonumber q(-1, -1, -1)\approx 0.467,\\
\nonumber q(-1/2, -1/2, -1)\approx 0.510,\\
\nonumber \text{ and }q(-1/2, -3/4, -3/4)\approx 0.509.
\end{align*}
The lemma follows.

\end{proof}

\begin{proof}[Proof of \Cref{lem:8}]
As before let us lower bound this expectation using \Cref{lem:13} with $Q=\{(1, 0, 0), (1, 0, 2), (3, 0, 0)\}$.
\begin{equation}
\frac{1+\sum_{i, j\leq k} \hat{f}_{i, jk}^2 u_{i, jk}(a, b, c)}{1+a} \geq \frac{1+\hat{f}_{1, 00}^2 a+\hat{f}_{1, 02}^2 a(b^2+c^2)+\hat{f}_{3, 00}^2 a^3-  (1/3 - \hat{f}_{1, 00}^2 - 2 \hat{f}_{1, 02}^2 -\hat{f}_{3, 00}^2)}{1+a}.
\end{equation}
Since $b^2+c^2\geq 0$, we can uniformly lower bound this expression as:
\begin{equation}
\frac{1+\hat{f}_{1, 00}^2 a+\hat{f}_{3, 00}^2 a^3-  (1/3 - \hat{f}_{1, 00}^2 - 2 \hat{f}_{1, 02}^2 -\hat{f}_{3, 00}^2)}{1+a}=:l(a).
\end{equation}
It is easy to check that the derivative is always negative, so the function is minimized at $a=1$.  The lemma follows since $l(1)\approx 0.639$.
\end{proof}

\begin{proof}[Proof of \Cref{lem:7}]
This proof is very similar to the proof of the rank $1$ case, however the function behaves a bit differently so different regions are needed.  Let 
\begin{align*}
\mathcal{S}_1= conv\{[-1, -1, -1], [r, r, -1], [r, -1, r], [-1, r, r]\},\\
\text{ and }\nonumber \mathcal{S}_2=\mathcal{S}\setminus \mathcal{S}_1,
\end{align*}
\noindent \textbf{where $r=0.6$ here and throughout the proof}.  Just as before, we will use a low order Hermite expansion for $\mathcal{S}_1$ and a higher order expansion along with some tedious analysis for $\mathcal{S}_2$.  Let us first apply \Cref{lem:12} with $Q=\{(1, 0, 0)\}$ to obtain:
\begin{equation*}
\frac{3+\sum_{i, j\leq k}\hat{f}_{i, jk}^2 p_{i, jk}(a, b, c))}{3+a+b+c} \geq \frac{3+\hat{f}_{1, 00}^2 (a+b+c)-3(1/3-\hat{f}_{1, 00}^2)}{3+a+b+c}.
\end{equation*}
The RHS is a function of $x:=a+b+c$, so we define:
\begin{equation*}
l(x):= \frac{3+\hat{f}_{1, 00}^2 x-3(1/3-\hat{f}_{1, 00}^2)}{3+x}.
\end{equation*}
In $\mathcal{S}_1$ $x$ ranges from $-3$ to $2r-1$.  It is easy to see $l$ is decreasing as a function of $x$, so we can obtain a lower bound on $\mathcal{S}_1$ of $l(2r-1)\approx 0.907$.

\begin{figure}
    \centering
    \includegraphics{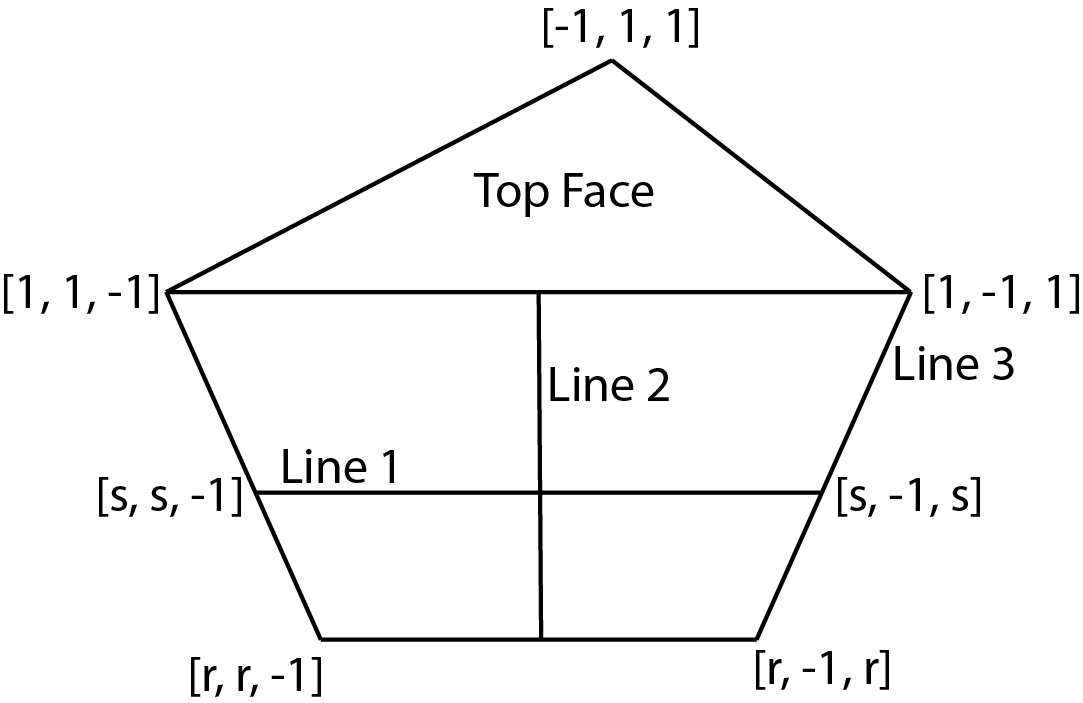}
    \caption{Relevant parameterizations for proof of \Cref{lem:7}, region $\mathcal{S}_2$}
    \label{fig:2}
\end{figure}

Now we turn to a bound on the function in $\mathcal{S}_2$.  First apply \Cref{lem:12} with $Q=\{(1, 0, 0), (1, 0, 2), (3, 0, 0)\}$ to obtain:
\begin{align*}
\frac{3+\sum_{i, j\leq k}\hat{f}_{i, jk}^2 p_{i, jk}(a, b, c))}{3+a+b+c} \geq \bigg{(} 3+\hat{f}_{1, 00}^2 p_{1, 00}(a, b, c) +\hat{f}_{1, 02}^2 p_{1, 02}(a, b, c) +\hat{f}_{3, 00}^2 p_{3, 00}(a, b, c)\\
\nonumber -3 (1/3-\hat{f}_{1, 00}^2-2\hat{f}_{1, 02}^2-\hat{f}_{3, 00}^2)\bigg{)} / \bigg{(}3+a+b+c \bigg{)} =: q(a, b, c).
\end{align*}
First, let us examine {\it inside} $\mathcal{S}_2$.  Just as before we can take a partial derivative:
\begin{align*}
\frac{\partial q}{\partial a}=\bigg{(} 2 (-30 + 12 a^3 + 6 b^2 - 4 b^3 + 6 c^2 - 4 c^3 + 
   a^2 (54 + 20 b + 20 c) + 4 a (b^2 + c (3 + c) \\ 
   + b (3 + 2 c)) - 225 \pi)\bigg{)}/\bigg{(}225 (3 + a + b + c)^2 \pi \bigg{)},
\end{align*}
\noindent and check that the constant term has large enough magnitude that the numerator is always negative:
\begin{align*}
\bigg{(} 2 (-30 + 12 a^3 + 6 b^2 - 4 b^3 + 6 c^2 - 4 c^3 + 
   a^2 (54 + 20 b + 20 c) + 4 a (b^2 + c (3 + c)  + b (3 + 2 c))\\ 
   \nonumber - 225 \pi)\bigg{)}\leq \bigg{(} 2 (-30 + 12  + 6  + 4  + 6  + 4  + (54 + 20  + 20 ) + 4  (1 +  (3 + 1) +  (3 + 2 ))
   \\ \nonumber  - 225 \pi)\bigg{)}\approx -1100.
\end{align*}
\noindent So, the numerator is always negative and we know the denominator is always positive ($\mathcal{S}_2$ does not contain $[-1, -1, -1]$) so their quotient is always negative.  

Consider a vertical face (see \Cref{fig:2}).  We can parameterize a horizontal line (line $1$) across the face as $[a, b, c]=[s, s-t, -1+t]$.  We can take the derrivative of $q(s, s-t, -1+t)$ and simplify to obtain:
\begin{align*}
\frac{dq(s, s-t, -1+t) }{dt}= -\frac{4 (-4 + 5 s) (1 + s - 2 t)}{225 \pi (1 + s)}.
\end{align*}
Critical points occur only when $t=(1+s)/2$ or $s=4/5$.  For $s=4/5$, this calculation shows that $q$ is constant along the line, hence minima can be taken to be the endpoints of the line.  The case $t=(1+s)/2$ corresponds to a vertical line through the center of the face (line $2$ in \Cref{fig:2}).  We can parameterize this vertical line as $[a, b, c]=[t+r(1-t), \frac{r-1}{2}(1-t), \frac{r-1}{2}(1-t)]$.  We calculate:
\begin{align*}
\frac{d q(t+r(1-t), \frac{r-1}{2}(1-t), \frac{r-1}{2}(1-t))}{d t}=\bigg{(} (-1 + r) (225 \pi + (21 + \\
\nonumber 22 r (-1 + t) - 22 t) (1 + r + t - r t)^2) \bigg{)}/\bigg{(}225 \pi (1 + r + t - r t)^2\bigg{)}.
\end{align*}
Just as before, the constant term in the numerator is large enough that the derivative never changes sign, so the function must be minimized at an endpoint.  

Now let us consider the ``top'' and ``bottom'' faces in the \Cref{fig:2}.  They are both parameterized as $[a, b, w-a-b]$ where $w=1$ for the top face and $w=2 r-1$ for the bottom face.  The two components of the gradient can be factored as:
\begin{gather*}
\frac{\partial q}{\partial a}\propto 8 (3 b - 4 w) (2 a + b - w)\\
\nonumber \text{ and  }\frac{\partial q}{\partial b}\propto 8 (3 a - 4 w) (a + 2 b - w)
\end{gather*}
\noindent On the top face, the only solution is $\{[1/3, 1/3, 1/3]\}$.  On the bottom face, we have critical points when 
\begin{equation*}
(a, b) \in \{(w/3, w/3), (4w/3, -5w/3),  (-5w/3, 4w/3), (4w/3, 4w/3)\}.
\end{equation*}

The only region of $\mathcal{S}_2$ which has not yet been covered explicitly, or implicitly by symmetry is line $3$ in the \Cref{fig:2}.  We can parameterize this line as $(a, b, c)=(r+t, r+t, -1)$ for $t\in [0, 1-r]$.  We take the derivative with respect to $t$ to obtain:
\begin{align*}
\frac{dq(r+t, r+t, -1)}{dt}=\frac{4(-5+8 r+8 t)(1+r+t)^2-225 \pi}{225 \pi (1+r+t)^2}.
\end{align*}
\noindent Once again the constants in the numerator are large enough that the derivative never changes sign.  Our analysis and symmetry of the function imply that we need to only check $5$ points:
\begin{align*}
q(1, 1, -1) \approx 0.811,\\
\nonumber q(w/3, w/3, w/3) \approx 0.934,\\
\nonumber q(4w/3, -5w/3, 4w/3) \approx 0.934,\\
\nonumber q(r, r, -1) \approx 0.933,\\
\nonumber q(1/3, 1/3, 1/3) \approx 0.805,\\
\nonumber q(r, (r-1)/2, (r-1)/2) \approx 0.935,\\
\nonumber \text{ and  }q(1, 0, 0) \approx 0.808.
\end{align*}

\end{proof}

\section{Technical Quantum Facts}\label{sec:20}
In this subsection we present the proofs for the quantum facts we used in the paper.  
\label{page:1}
\begin{proof}[Proof of \Cref{lem:14}]
For this proof we will have the Bell states take their usual definition:
\begin{gather}\label{eq:51}
\ket{\Phi^+}=\frac{\ket{00}+\ket{11}}{\sqrt{2}}, \,\,\,\,\,\,\,\,\,\,\,\,\, \ket{\Phi^-}=\frac{\ket{00}-\ket{11}}{\sqrt{2}},\\
\nonumber \ket{\Psi^+}=\frac{\ket{01}+\ket{10}}{\sqrt{2}}, \,\,\,\,\,\,\,\text{ and }\,\,\,\,\,\, \ket{\Psi^-}=\frac{\ket{01}-\ket{10}}{\sqrt{2}}.
\end{gather}
If $rank(P)=1$ then $P=\ket{\phi}\bra{\phi}$.  Write $\ket{\phi}$ in the Bell basis: $\ket{\phi}=\alpha_1 \ket{\Phi^+}+\alpha_2 \ket{\Phi^-}+\alpha_3 \ket{\Psi^+}+\alpha_4 \ket{\Psi^-}$.  Then,
\begin{gather*}
\Gamma_{11}=\text{Tr} [\sigma^1\otimes \sigma^1 \ket{\phi} \bra{\phi}]=|\alpha_1|^2-|\alpha_2|^2+|\alpha_3|^2-|\alpha_4|^2,\\
\nonumber \Gamma_{22}=\text{Tr} [\sigma^2\otimes \sigma^2 \ket{\phi} \bra{\phi}]=-|\alpha_1|^2+|\alpha_2|^2+|\alpha_3|^2-|\alpha_4|^2,\\
\nonumber  \text{ and }\Gamma_{33}=\text{Tr} [\sigma^3\otimes \sigma^3 \ket{\phi} \bra{\phi}]=|\alpha_1|^2+|\alpha_2|^2-|\alpha_3|^2-|\alpha_4|^2.
\end{gather*}
So, $[\Gamma_{11}, \Gamma_{22}, \Gamma_{33}]=|\alpha_1|^2 [1, -1, 1]+|\alpha_2|^2 [-1, 1, 1]+|\alpha_3|^2[1, 1, -1]+|\alpha_4|^2 [-1, -1, -1]$.  

If $rank(P)=2$, $P=\ket{\phi_1} \bra{\phi_1} +\ket{\phi_2} \bra{\phi_2}$.  Once again let us write these states in the Bell basis:
\begin{gather*}
\ket{\phi_1}=\alpha_1 \ket{\Phi^+}+\alpha_2 \ket{\Phi^-}+\alpha_3 \ket{\Psi^+}+\alpha_4 \ket{\Psi^-}\\
\text{ and }\ket{\phi_2}=\beta_1 \ket{\Phi^+}+\beta_2 \ket{\Phi^-}+\beta_3 \ket{\Psi^+}+\beta_4 \ket{\Psi^-},
\end{gather*}
where $\sum_i \alpha_i^* \beta_i=0$ since $\ket{\phi_1}$ and $\ket{\phi_2}$ must be orthogonal.  $[\Gamma_{11}, \Gamma_{22}, \Gamma_{33}]$ can be evaluated in the same way:
\begin{gather*}
\Gamma_{11}=\text{Tr} [\sigma^1\otimes \sigma^1 P]=|\alpha_1|^2-|\alpha_2|^2+|\alpha_3|^2-|\alpha_4|^2+|\beta_1|^2-|\beta_2|^2+|\beta_3|^2-|\beta_4|^2,\\
\nonumber \Gamma_{22}=\text{Tr} [\sigma^2\otimes \sigma^2 P]=-|\alpha_1|^2+|\alpha_2|^2+|\alpha_3|^2-|\alpha_4|^2-|\beta_1|^2+|\beta_2|^2+|\beta_3|^2-|\beta_4|^2,\\
\nonumber  \text{ and }\Gamma_{33}=\text{Tr} [\sigma^3\otimes \sigma^3 P]=|\alpha_1|^2+|\alpha_2|^2-|\alpha_3|^2-|\alpha_4|^2+|\beta_1|^2+|\beta_2|^2-|\beta_3|^2-|\beta_4|^2.
\end{gather*}
So, $[\Gamma_{11}, \Gamma_{22}, \Gamma_{33}]=(|\alpha_1|^2+|\beta_1|^2) [1, -1, 1]+(|\alpha_2|^2+|\beta_2|^2) [-1, 1, 1]+(|\alpha_3|^2+|\beta_3|^2)[1, 1, -1]+(|\alpha_4|^2+|\beta_4|^2) [-1, -1, -1]$.  Observe that the hull $\mathcal{T}$ is defined by the following inequalities:
\begin{align}
\label{eq:40}[x, y, z]\in \mathcal{T} &\Leftrightarrow -2 \leq x+y+z \leq 2,\\
\label{eq:41}& \,\,\,\,\,\,\,\,   -2 \leq x-y+z \leq 2,\\
\label{eq:42}& \,\,\,\,\,\,\,\,   -2 \leq -x+y+z \leq 2,\\
\label{eq:43}& \,\,\,\,\,\,\,\,   -2 \leq x+y-z \leq 2.
\end{align}
We will prove this case by demonstrating our derived expression for $[\Gamma_{11}, \Gamma_{22}, \Gamma_{33}]$ satisfies all these equations.  

\textbf{\Cref{eq:40}}
\begin{align*}
\Gamma_{11}+\Gamma_{22}+\Gamma_{33}=(|\alpha_1|^2+|\beta_1|^2) +(|\alpha_2|^2+|\beta_2|^2)+(|\alpha_3|^2+|\beta_3|^2)-3(|\alpha_4|^2+|\beta_4|^2)\\
=(1-|\alpha_4|^2)+(1-|\beta_4|^2)-3(|\alpha_4|^2+|\beta_4|^2)=2(1-2(|\alpha_4|^2+|\beta_4|^2)).
\end{align*}
The case follows since $0 \leq |\alpha_4|^2+|\beta_4|^2 \leq 1$.  To see this, complete the vectors $\boldsymbol{\alpha}$ and $\boldsymbol{\beta}$ to a basis for $\mathbb{C}^4$: $\{\boldsymbol{\alpha}, \boldsymbol{\beta}, \boldsymbol{\gamma}, \boldsymbol{\eta}\}$.  Also let $e_4=[0, 0, 0, 1]$.  Then, $Tr((\ket{\boldsymbol{\alpha}}\bra{\boldsymbol{\alpha}}+\ket{\boldsymbol{\beta}}\bra{\boldsymbol{\beta}}+\ket{\boldsymbol{\gamma}}\bra{\boldsymbol{\gamma}}+\ket{\boldsymbol{\eta}}\bra{\boldsymbol{\eta}})\ket{e_4}\bra{e_4})=1$.

\textbf{\Cref{eq:41}}
\begin{align*}
\Gamma_{11}-\Gamma_{22}+\Gamma_{33}=3(|\alpha_1)^2+|\beta_1|^2-(|\alpha_2|^2+|\beta_2|^2)-(|\alpha_3|^2+|\beta_3|^2)-(|\alpha_4|^2+|\beta_4|^2)\\
=3(|\alpha_1|^2+|\beta_1|^2-(2-|\alpha_1|^2-|\beta_1|^2))=2(2(|\alpha_1|^2+|\beta_1|^2)-1).
\end{align*}
This case follows since $0 \leq |\alpha_1|^2 +|\beta_1|^2 \leq 1$ by the same argument.

\textbf{\Cref{eq:42}}
\begin{align*}
-\Gamma_{11}+\Gamma_{22}+\Gamma_{33}=(-1)(|\alpha_1)^2+|\beta_1|^2)+3(|\alpha_2|^2+|\beta_2|^2)-(|\alpha_3|^2+|\beta_3|^2)-(|\alpha_4|^2+|\beta_4|^2)\\
=2(2(|\alpha_2|^2+|\beta_2|^2)-1).
\end{align*}
This case follows since $0 \leq |\alpha_2|^2 +|\beta_2|^2 \leq 1$ by the same argument.  

\textbf{\Cref{eq:43}}
\begin{align*}
\Gamma_{11}+\Gamma_{22}-\Gamma_{33}=(-1)(|\alpha_1)^2+|\beta_1|^2)-(|\alpha_2|^2+|\beta_2|^2)+3(|\alpha_3|^2+|\beta_3|^2)-(|\alpha_4|^2+|\beta_4|^2)\\
=2(2(|\alpha_3|^2+|\beta_3|^2)-1).
\end{align*}
This case follows since $0 \leq |\alpha_3|^2 +|\beta_3|^2 \leq 1$ by the same argument.

If $\text{rank}(P)=3$, then $P=\mathbb{I}-\ket{\phi}\bra{\phi}$ so if the moment matrix for $\ket{\phi}\bra{\phi}$ has values $[\Gamma_{11}, \Gamma_{22}, \Gamma_{33}]$, then the moment matrix for $P$ has values $[\Gamma_{11}', \Gamma_{22}', \Gamma_{33}']=-[\Gamma_{11}, \Gamma_{22}, \Gamma_{33}]$.  It is easy to verify that $\{-[\Gamma_{11}, \Gamma_{22}, \Gamma_{33}]: [\Gamma_{11}, \Gamma_{22}, \Gamma_{33}]\in \mathcal{S}\}=-\mathcal{S}$ as defined so \Cref{eq:59} is proven.  To prove \Cref{eq:60}, note:
\begin{equation*}
0 \leq Tr(P \rho) \leq 1 \Rightarrow 0 \leq 4 Tr(P \rho) \leq 4 \Rightarrow -k \leq 4Tr(P \rho)-k \leq 4-k.
\end{equation*}

To prove \Cref{eq:63} set $[\Sigma_{11}, \Sigma_{22}, \Sigma_{33}]=[a, b, c]$ and $[\Gamma_{11}, \Gamma_{22}, \Gamma_{33}]=[p, q, r]$.  We will show that if $[a, b, c]$ is restricted to some polytope, $\mathcal{P}_1$, and $[p, q, r]$ is restricted to some polytope, $\mathcal{P}_2$, then
\begin{equation}\label{eq:61}
    \min_{\substack{ [a, b, c] \in\mathcal{P}_1 \\  [p, q, r] \in \mathcal{P}_2}} [a, b, c] \cdot{} [p, q, r] \geq \min_{\substack{ [a, b, c] \in\mathcal{B}_1\\  [p, q, r] \in \mathcal{B}_2}} [a, b, c] \cdot{} [p, q, r],
\end{equation}
where $\mathcal{B}_i$ is the set of extreme points of $\mathcal{P}_i$.  Similarly, we will show \begin{equation}\label{eq:62}
    \max_{\substack{  [a, b, c] \in\mathcal{P}_1\\ [p, q, r] \in \mathcal{P}_2}} [a, b, c] \cdot{} [p, q, r] \leq \max_{\substack{  [a, b, c] \in\mathcal{B}_1\\ [p, q, r] \in \mathcal{B}_2}} [a, b, c] \cdot{} [p, q, r].
\end{equation}
For \Cref{eq:61} fix $[p, q, r] \in \mathcal{P}_2$ and write it as a convex combination of extreme points $[p, q, r] =\sum_i \lambda_i [p_i, q_i, r_i]$.  
\begin{equation}
 \min_{ [a, b, c] \in\mathcal{P}_1} \sum_i \lambda_i [a, b, c] \cdot{} [p_i, q_i, r_i].
 \end{equation}
Since $\sum_i \lambda_i=1$ we may lower bound this quantity by 
\begin{equation*}
    \geq \min_{\substack{  [a, b, c] \in\mathcal{P}_1\\  [p', q', r'] \in \mathcal{B}_2}} [a, b, c] \cdot{} [p', q', r']. 
\end{equation*}
Since this argument holds for all $(p, q, r)$ we can uniformly lower bound:
\begin{equation*}
     \min_{\substack{ [a, b, c] \in\mathcal{P}_1\\ [p, q, r] \in \mathcal{P}_2}} [a, b, c] \cdot{} [p, q, r] \geq \min_{\substack{ [a, b, c] \in\mathcal{P}_1\\ [p', q', r'] \in \mathcal{B}_2}} [a, b, c] \cdot{} [p', q', r'].
\end{equation*}
For fixed $[p, q, r] \in \mathcal{B}_2$ apply the same argument to $[a, b, c]$ to get \Cref{eq:61}.  The argument for \Cref{eq:62} is similar.  Applying \Cref{eq:61} and \Cref{eq:62} to the relevant polytopes provides \Cref{eq:63}.

\Cref{eq:64} follows from \cite{G16}.  To prove \Cref{eq:65} if $k=1$ or $3$ note that in this case the $1$-local parts of the projector are proportional to the $1$-local parts for a pure state, hence \cite{G16} provides this case.  For $k=2$ observe that the $1$-local part of the $2$-moment is the sum of $1$-local parts from two pure states.  Hence, by the triangle inequality the length of the $1$-local part is at most $2$.  

\end{proof}

We also give the map between standard (classical)  and Max 2-QSAT:
\begin{theorem}\label{thm:6}
 Max 2-QSAT generalizes Max 2-SAT.
\end{theorem}
\begin{proof}
 An instance of Max 2-SAT corresponds to a set of variables $\{x_i\}_{i=1}^n$ as well as a multiset of clauses $E=\{(y_i, y_j)\}$ where each $y_i$ is either $x_i$ or $\neg x_i$.  Let $OPT$ be the maximum number of clauses which are satisfiable with some Boolean assignment to the variables $\{x_i\}$.  Define the following function:
\begin{equation*}
b(y_i)=\begin{cases}
1 \text{ if $y_i=\neg x_i$}\\
0 \text{ if $y_i=x_i$}
\end{cases},
\end{equation*}
and consider the following mapping between clauses and Hamiltonians, 
\begin{equation*}
(y_i, y_j) \leftrightarrow \frac{3}{4}\mathbb{I}_{ij}-(-1)^{b(y_i)}\frac{\sigma^3_i \otimes \mathbb{I}_j}{4} -(-1)^{b(y_j)} \frac{\mathbb{I}_i \otimes \sigma^3_j}{4}  -(-1)^{b(y_i)+b(y_j)} \frac{\sigma^3 \otimes \sigma^3}{4}.
\end{equation*}
Denote the $2$-Local term on the RHS as $H_e(y_i, y_j)$.  Then, we claim 
\begin{equation*}
OPT=\lambda_{max}\left( \sum_{(y_i, y_i) \in E} H_e(y_i, y_j)\right).
\end{equation*}
This holds because the Hamiltonian is diagonal in the computational basis, hence its largest eigenstate can be assumed to be a computational basis state.  Such a state has $0$ and $1$ values corresponding to $False$ and $True$ assignements.  

\end{proof}

The next result we present is that our constants are nearly tight for quadratic Hamiltonians, in that our approximation factors are about as good as we can expect for approximation algorithms which yield product states.

\begin{theorem}\label{thm:gap-examples}
There exist instances of the $2$-Local Hamiltonian problem, $\{H_e\}$ with strictly quadratic projectors where each $H_e$ has rank $k$ such that
\begin{equation}
\max_{\ket{\phi_1}\in \mathbb{C}^2, \ket{\phi_2}\in \mathbb{C}^2} \bra{\phi_1} \otimes \bra{\phi_2} \left( \sum_e H_e\right) \ket{\phi_1}\otimes \ket{\phi_2} \leq \beta(k) \lambda_{max} \left( \sum_e H_e\right),
\end{equation}
where $
\beta(k)=\begin{cases}
1/2 \text{ if $k=1$}\\
2/3 \text{ if $k=2$}\\
5/6 \text{ if $k=3$}\end{cases}
$.
\end{theorem}
\begin{proof}
Let the Bell states be defined as in \Cref{eq:51}.  
If $k=1$ consider the Hamiltonian $H=\ket{\Psi^-}\bra{\Psi^-}$.  From, e.g., \cite{G19}, we know that if we optimize $\ket{\phi}$ over product states we can get at most $1/2$ for $\bra{\phi} H\ket{\phi}$, while it is apparent that  $\lambda_{max}=1$ since we can take the singlet, $\ket{\Phi^-}$.

If $k=2$ consider the Hamiltonian (recall the problem formulation allows multi-edges):
\begin{align}
H=\frac{1}{3}(\ket{\Psi^- }\bra{\Psi^-}+\ket{\Phi^+}\bra{\Phi^+})+\frac{1}{3}(\ket{\Psi^- }\bra{\Psi^-}+\ket{\Phi^-}\bra{\Phi^-})\\
\nonumber
+\frac{1}{3}(\ket{\Psi^- }\bra{\Psi^-}+\ket{\Psi^+}\bra{\Psi^+})
\end{align}
\vspace{-0.9 cm}
\begin{align}
=\frac{1}{3}\mathbb{I}+\frac{2}{3}\ket{\Psi^-}\bra{\Psi^-}.
\end{align}
Similarly, we can expect objective at most $1/2$ for the singlet if we optimize over product states, hence $\max_{\ket{\phi}\in PROD} \bra{\phi} H\ket{\phi}=2/3$.  Observe that we can achieve $\lambda_{max}=1$ by taking the singlet once again.

If $k=3$ consider:
\begin{align}
    H=\frac{1}{3}\left(\mathbb{I}-\ket{\Phi^+}\bra{\Phi^+}\right)+\frac{1}{3}\left(\mathbb{I}-\ket{\Phi^-}\bra{\Phi^-}\right)+\frac{1}{3}\left(\mathbb{I}-\ket{\Psi^+}\bra{\Psi^+}\right)
    \end{align}
    \vspace{-0.75 cm}
    \begin{align}
    =\frac{2}{3}\mathbb{I}+\frac{1}{3}\ket{\Psi^-} \bra{\Psi^-}.
\end{align}
We get at most $5/6$ for product states in the same way, while $\lambda_{max}=1$.

\end{proof}

\section{Traceless, Bipartite, and Strictly Quadratic Hamiltonians}\label{sec:21}

The final result we present is an approximation algorithm for strictly quadratic (\Cref{def:2}) instances of traceless $2$-Local Hamiltonian on bipartite interaction graphs.  Since this result is relatively self-contained and very different from the main results of this paper, we present both the statements and the analysis here.  We demonstrate that this problem can be approximated using Krivine rounding \cite{K77}, with the same product-state ansatz as in \cite{B19}.  More generally we show how traceless instances of $2$-LH can be approximated using an approximation algorithm for a related classical $2$-CSP.

An instance, $\{H_e\}$, of QLH on $n$ qubits (\Cref{prob:QLH}) is $\emph{traceless}$ if $\text{Tr}[\sum_e H_e] = 0$.  For a traceless instance, we may assume that each term $H_e$ is traceless by replacing it with $H'_e = H_e  -\frac{1}{2^n}\text{Tr}[H_e]\mathbb{I}$, which preserves the maximum eigenvalue of $H=\sum_e H_e$.  
%the corresponding Hamiltonian problem is traceless if and only if finding the largest eigenvalue of $H=\sum_e H_e$ is equivalent to a $2$-local Hamiltonian problem where each $H_e$ is itself traceless (a particular $H_e$ in the input can be specified as the sum of two such terms, but then we would create local Hamiltonians $H_e$ and $H_{e'}$ to define an equivalent problem).  
We will also assume for our result that each $H_e$ is strictly quadratic, so that we may assume each $H_e=w_e \sigma_i^k \otimes \sigma_j^l \otimes \mathbb{I}_{[n]\setminus \{i, j\}}$ for $w_e \neq 0$ and $k,l \in [3]$.  Given such a Hamiltonian, we will create a ``Pauli'' interaction graph $G=(V, E)$ where the set of vertices have a one-to-one correspondence with nontrivial single qubit Pauli operators.  If there are $n$ qubits there would be $3n$ such vertices, one for each $\sigma_i^k$ for $i\in [n]$ and $k \in [3]$.  We place an edge between vertices if there is a term $H_e$ in the description of the local Hamiltonian problem which contains both Pauli matrices, i.e.\ vertex $(i, k)$ is connected to $(j, l)$ if there is an $H_e=w_e \sigma_i^k \otimes \sigma_j^l \otimes \mathbb{I}_{[n]\setminus \{i, j\}}$.  If the corresponding graph is bipartite, then we say that the corresponding (strictly quadratic) $2$-local Hamiltonian is traceless and bipartite.  For ease of reference, we provide a formal definition:
\begin{definition}[Pauli Interaction Graph]\label{def:1}
Let $\{H_e\}$ be a set of $2$-local terms on $n$ qubits where each $H_e=w_e \sigma_i^k \otimes \sigma_j^l\otimes \mathbb{I}_{[n]\setminus \{i, j\}}$ with $w_e \neq 0$ and $k,l\in [3]$.  Let $V=\{(i, k) \mid i\in [n] \text{ and } k\in [3]\}$.  Construct a set of pairs $E\subseteq [3n]\times [3n]$ with $(ik, jl) \in E$ if and only if there exists $e$ such that $H_e=w_e \sigma_i^k \otimes \sigma_j^l\otimes \mathbb{I}_{[n]\setminus \{i, j\}}$.  We say that $\{H_e\}$ is strictly quadratic, bipartite and traceless if $G=(V, E)$ is bipartite.
\end{definition}
\noindent Our general approach is to solve a classical CSP on the Pauli interaction graph of a $2$-local Hamiltonian.  In particular, we seek to maximize the weight earned from the edges $\sigma_i^k \otimes \sigma_j^l\otimes \mathbb{I}_{[n]\setminus \{i, j\}}$, where we now assign values in $\{\pm 1\}$ to each variable $\sigma_i^k$, represented by a vertex in $V$. Such an assignment is converted to a product state by treating the unit vector $[\sigma_i^1,\sigma_i^2,\sigma_i^3]/\sqrt{3}$ as a Bloch vector for qubit $i$. We note that this approach works for any traceless instance of $2$-LH, where a classical $\alpha$-approximation algorithm for the classical $2$-CSP on the Pauli interaction graph yields a classical $\frac{\alpha}{3}$-approximation for the original $2$-LH instance.

The strictly quadratic, bipartite, and traceless case of the $2$-Local Hamiltonian problem is still $QMA$-hard, since many (strictly quadratic) families of Hamiltonians retain $QMA$-hardness even when the edges are restricted to a $2$-$d$ lattice \cite{P15}.  For this problem, the relevant relaxation is a weakening of our main SDP, \Cref{prob:5}.  We opt to drop many of the constraints present in \Cref{prob:5} for consistency with the result we use as a black box~\cite{K77,BJ10}.  Given a local term of the form $H_e=w_e \sigma_i^k \otimes \sigma_j^l \otimes \mathbb{I}_{[n]\setminus \{i, j\}}$, where in the present context we allow $w_e$ to be negative, let $D_e\in \mathbb{R}^{3n \times 3n}$ be a symmetric matrix indexed by nontrivial single-qubit Pauli operators (i.e.\ $[n]\times [3]$), such that:
\begin{gather}\label{eq:67}
    D_e(\sigma_i^k, \sigma_j^l)=D_e(\sigma_j^l, \sigma_i^k)=w_e/2,\\
    \nonumber \text{and  } D_e=0 \text{ otherwise}.
\end{gather}

\begin{problem}\label{prob:4}
Given a strictly quadratic instance of QLH on $n$ qubits with terms $\{H_e\}$, for each $H_e$ let $D_e$ be defined as in \Cref{eq:67}.  Solve the following SDP:
\begin{alignat}{2}
    \max \sum_e &\text{Tr}\mathrlap{[D_e M]}\\
    \notag s.t.\qquad M(\sigma_i^k, \sigma_i^k)&=1 \quad && \forall (i,k)\in V,\\
    \notag \mathcal{S}(\mathbb{R}^{3n \times 3n}) \ni M &\semigeq 0,
\end{alignat}
where $\mathcal{S(\cdot)}$ refers to the class of symmetric matrices.
\end{problem}

Given this relaxation, which is equivalent to the relaxation used by Goemans and Williamson~\cite{G95}, we appeal to a classical rounding scheme, used as a black box, to obtain the result.  Given some graph $G=(V, E)$ with $|V|=n$, let $A$ be a symmetric matrix with rows and columns indexed by elements of $V$.  For each $v\in V$ we associate a $\mathbf{w}_v\in \mathbb{R}^n$.  The main observation is the following:

\begin{theorem}[\cite{K77} and Remark 2 on p.~95 in \cite{BJ10}]\label{thm:5}
If $G$ is a bipartite graph, then there is a polynomial-time randomized algorithm which outputs variables $\{z_v\}_{v\in V}$ where each $z_v \in \{\pm 1\}$ and 
\begin{equation*}
\mathbb{E} \left[\sum_{uv\in E} A_{uv}z_u z_v \right]\geq \frac{2(\ln(1+\sqrt{2}))}{\pi}\sum_{uv\in E} A_{uv}\mathbf{w}_u \cdot \mathbf{w}_v.
\end{equation*}
\end{theorem}
The above gives an approximation algorithm for the Grothendieck problem if the underlying graph is bipartite.  The right hand side corresponds to the semidefinite relaxation of some NP-hard optimization problem and the left hand side corresponds to some feasible (but likely not optimal) solution to the optimization problem.  We obtain the following corollary:

\begin{corollary}
Let $H=\sum_e H_e$ be traceless, strictly quadratic, and bipartite in the sense of \Cref{def:1}.  Then, there is a randomized polynomial-time algorithm which produces a pure product state $\ket{\psi}$ such that:
\begin{equation}
\mathbb{E} \left[\bra{\psi} \sum_e H_e \ket{\psi} \right] \geq \frac{2 \ln(1+\sqrt{2})}{3\pi} \lambda_{max}  \left(\sum_e H_e \right).
\end{equation}
\end{corollary}

%The technique used in the proof is referred to as ``Krivine rounding'' in the literature REF.  In this algorithm, the vectors $\mathbf{w}_v$ are mapped to $\phi(\mathbf{w}_v)\in \mathbb{R}^N$ for $N>>n$ in such a way that $\phi(\mathbf{w}_v)^T \phi(\mathbf{w}_u)$ ``simulates'' some function using its power series.  This function is used to invert the expectation so that $[\mathbf{r}^T\phi(\mathbf{w}_v)][\mathbf{r}^T\phi(\mathbf{w}_u)]=c\mathbf{w}_v^T \mathbf{w}_u$ where $\mathbf{r}$ is multivariate normal and sufficiently large.  In this way, one can find an algorithm that produces 
%Imagine we have two fixed vectors $\mathbf{w}_v$ and $\mathbf{w}_u$.  If we take $x_v=sign(\mathbf{w}_v^T \mathbf{r})$ and $x_u=sign(\mathbf{w}_u^T \mathbf{r})$ then we obtain $\mathbb{E}[x_v x_u]=s(\theta)$, some function of the overlap between the two vectors $\mathbf{w}_v$ and $\mathbf{w}_u$.  Essentially the algorithm takes as input the vectors $\mathbf{w}_v$, and embeds them into a larger space to simulate taking the inverse of $s(\theta)$.

\begin{proof}
Construct a $3n\times3n$ matrix $A$ by taking
\begin{equation*}
A(\sigma_i^k, \sigma_j^l)=Tr\left[(\sigma_i^k \otimes \sigma_j^l \otimes \mathbb{I}_{[n]\setminus \{i, j\}})\sum_e H_e\right]/2^{n+1}.
\end{equation*}

Observe that the matrix $A$ is bipartite in the same sense as \Cref{thm:5}: Let $G=(V, E)$ be the graph described previously, where we have a vertex for each single qubit Pauli matrix and vertices are connected if they interact through $H$.  We have a partition of the vertices into two sets $B_1$ and $B_2$ where no edge links two vertices in the same $B_i$.  Then, $A(\sigma_i^k, \sigma_j^k)=0$ if [$(i, k) \in B_1$ and $(j, l) \in B_1$] or [$(i, k) \in B_2$ and $(j, l) \in B_2$].  

Now let us define the following optimization problem:
\begin{align}\label{eq:68}
    \max_{\substack{z_{ik}\in \{\pm 1\}\\\forall (i, k) \in V}} \sum_{(ik, jl) \in E} A(\sigma_i^k, \sigma_j^l) z_{ik}z_{jl}.
\end{align}
We can take its natural semidefinite relaxation \cite{G95} to obtain:
\begin{align*}
    \max_{\substack{M(\sigma_i^k, \sigma_i^k)=1\ \forall i, k\\\mathcal{S}(\mathbb{R}^{3n\times 3n}) \ni M \semigeq 0}}\ \sum_{(ik, jl)\in E }A(\sigma_i^k, \sigma_j^l) M(\sigma_i^k, \sigma_j^l),
\end{align*}
\noindent or written in terms of the Cholesky vectors:
\begin{align*}
     \max_{\substack{\mathbf{w}_{ik}\in \mathbb{R}^{3n}\\ ||\mathbf{w}_{ik}||=1\ \forall (i, k) \in V}}\ \sum_{(ik, jl) \in E} A(\sigma_i^k, \sigma_j^l) \mathbf{w}_{ik}\cdot \mathbf{w}_{jl}.
\end{align*}
Observe that \Cref{eq:68} is exactly the relaxation for the $2$-Local Hamiltonian problem, \Cref{prob:4}, since $\sum_e D_e=A$.  \Cref{thm:5}, Krivine rounding, provides a random set of variables $z_{ik}\in \{\pm 1\}$ such that:
\begin{equation*}
\mathbb{E}\left[ \sum_{(ik, jl) \in E} A(\sigma_i^k, \sigma_j^l) z_{ik}z_{jl} \right] \geq \frac{2 \ln(1+\sqrt{2})}{\pi} \lambda_{max}  \left(\sum_e H_e \right).
\end{equation*}
Dividing both sides by $3$, we obtain:
\begin{equation*}
\mathbb{E}\left[ \sum_{(ik, jl) \in E} A(\sigma_i^k, \sigma_j^l) \frac{z_{ik}}{\sqrt{3}}\frac{z_{jl}}{\sqrt{3}} \right] \geq \frac{2 \ln(1+\sqrt{2})}{3\pi} \lambda_{max}  \left(\sum_e H_e \right).
\end{equation*}

Note that the objective of the LHS corresponds to a valid quantum product state just as in \cite{B19}:
\begin{equation*}
\rho=\bigotimes_{i=1}^n \left( \frac{\mathbb{I} +(z_{i1}/\sqrt{3})\sigma^1 +(z_{i2}/\sqrt{3})\sigma^2 +(z_{i3}/\sqrt{3})\sigma^3 }{2}\right).
\end{equation*}
\end{proof}

\end{document}